%% file: main-wmso.tex
\documentclass[11pt]{article}

\input{yv-preamble.tex}

\input{fc-commands.tex}

\input{fz-commands.tex}

\usepackage{enumerate,xspace,stmaryrd,xcolor}
\usepackage[shortlabels]{enumitem}
\usepackage[e]{esvect}

\usepackage[author=anonymous]{fixme}
\FXRegisterAuthor{af}{aaf}{AF}
\FXRegisterAuthor{fc}{afc}{FC}
\FXRegisterAuthor{fz}{afz}{FZ}
\FXRegisterAuthor{yv}{ayv}{YV}

\title{%
Weak MSO: Automata and Expressiveness Modulo Bisimilarity%
\footnote{Emails: \texttt{fcarreiro@dc.uba.ar}, \texttt{facchini@mimuw.edu.pl}, \texttt{y.venema@uva.nl}, \texttt{fabio.zanasi@ens-lyon.fr}.}%
}

\author{%
    Facundo Carreiro\\%
    \emph{ILLC, Universiteit van Amsterdam, The Netherlands}\\ \\
    Alessandro Facchini\\%
    \emph{University of Warsaw, Poland}\\ \\
    Yde Venema\\%
    \emph{ILLC, Universiteit van Amsterdam, The Netherlands}\\ \\
    Fabio Zanasi\\%
    \emph{ENS Lyon, U. Lyon, CNRS, INRIA, UCBL, France}
}

\date{\small Last revision: \today}

\begin{document}

\maketitle

\input{abstract}

\input{introduction}

\section{Preliminaries}\label{sec:prel}
\input{preliminaries.tex}
\section{One-step logics, normal forms and continuity}\label{sec:onestep}
\input{one-step.tex}

\section{Automata for WMSO}\label{sec:aut}
\input{msoaut.tex}

\input{autchar.tex}

\section{From $\yvWMSO$-automata to $\yvWMSO$}\label{sec:aut-to-formula_wmso}
\input{aut-to-formula_wmso}

\section{Modal Characterization of WMSO}\label{sec:char}
\input{jw_intro}

\subsection{One-step translations}\label{pinvariant-fragment}
\input{pinvariant-fragment}

\subsection{From automata in $\yvcwAut(\ofo)$ to formulas in $\yvF$}\label{aut-to-formula}
\input{aut-to-formula}



\section{Conclusions}
\input{conclusions}


\bibliographystyle{plain}
\bibliography{main-wmso}

\end{document}

%% file: yv-preamble.tex
\usepackage{amssymb}
\usepackage{amsmath}
\usepackage{latexsym}

\newenvironment{tbs}{%
   \small\tt
   \begin{itemize}}{\end{itemize}}
\newcommand{\btbs}{\begin{tbs}}                                                                      
\newcommand{\etbs}{\end{tbs}}

\newcommand{\hide}[1]{}

\newtheorem{theorem}{Theorem}[section]

\newtheorem{fact}[theorem]{Fact}
\newtheorem{proposition}[theorem]{Proposition}
\newtheorem{lemma}[theorem]{Lemma}
\newtheorem{corollary}[theorem]{Corollary}
\newtheorem{defi}[theorem]{Definition}

\newtheorem{conv}[theorem]{Convention}
\newtheorem{rema}[theorem]{Remark}
\newtheorem{exam}[theorem]{Example}

\newenvironment{definition}{\begin{defi}\rm}{\end{defi}}

\newenvironment{remark}{\begin{rema}\rm}{\hfill $\lhd$\end{rema}}

\newenvironment{proof}{\begin{trivlist}\item[]{\bf
Proof.}}{\hfill {\sc qed}\end{trivlist}}
\newenvironment{proofsketch}{\begin{trivlist}\item[]{\bf
Proof sketch.}}{\hfill {\sc qed}\end{trivlist}}
\newenvironment{proofof}[1]{\begin{trivlist}\item[\hskip\labelsep{\bf
Proof~of~{#1}.\ }]}{\hspace*{\fill} {\sc qed}\end{trivlist}}

\newtheorem{claim2}{\sc Claim}
\newenvironment{claim}{\begin{claim2}\rm}{\end{claim2}\rm}
\newenvironment{claimfirst}{\setcounter{claim2}{0}
               \begin{claim2}\rm}{\end{claim2}\rm}

\newenvironment{pfclaim}{\begin{trivlist}\item[]{\sc Proof of
Claim.}}{\hfill {\mbox{$\blacktriangleleft$}}\end{trivlist}}

\newtheorem{exer}{Exercise}[section]

\newcommand{\mathstr}[1]{\mathbb{#1}}

\newcommand{\struc}[1]{(#1)}

\newcommand{\bbA}{\mathstr{A}}

\newcommand{\bbN}{\mathstr{N}}

\newcommand{\bbT}{\mathstr{T}}

\newcommand{\bis}{\mathrel{\mathchoice%
{\raisebox{.3ex}{$\,
  \underline{\makebox[.7em]{$\leftrightarrow$}}\,$}}%
{\raisebox{.3ex}{$\,
  \underline{\makebox[.7em]{$\leftrightarrow$}}\,$}}%
{\raisebox{.2ex}{$\,
  \underline{\makebox[.5em]{\scriptsize$\leftrightarrow$}}\,$}}%
{\raisebox{.2ex}{$\,
  \underline{\makebox[.5em]{\scriptsize$\leftrightarrow$}}\,$}}}}

\newcommand{\Prop}{\ensuremath{\mathsf{P}}}        

\newcommand{\ML}{\ensuremath{\mathrm{ML}}}       
\newcommand{\muML}{\ensuremath{\mu\ML}}    
       
\newcommand{\mso}{\mathrm{MSO}}       

\newcommand{\yvF}{\contAFMC}

\newcommand{\yvAut}{\mathit{Aut}}
\newcommand{\yvwAut}{\mathit{Aut}_{w}}
\newcommand{\yvcwAut}{\mathit{Aut}_{cw}}

\newcommand{\yvMSO}{\ensuremath{\mathrm{MSO}}}
\newcommand{\yvWMSO}{\ensuremath{\mathrm{WMSO}}}
\newcommand{\yvWFMSO}{\ensuremath{\mathrm{WFMSO}}}
\newcommand{\yvL}{\mathcal{L}}
\newcommand{\yvLo}{\yvL_{1}}
\newcommand{\yvmLo}{\yvL_{1}^{+}}
\newcommand{\yvdual}[1]{#1^{\delta}}

\renewcommand{\phi}{\varphi} 

\newcommand{\q}{a}

\newcommand{\ai}{\q_{I}}

\newcommand{\eloi}{\exists}
\newcommand{\abel}{\forall}

\newcommand{\haak}[1]{[\![#1]\!]}

\newcommand{\sse}{\subseteq}

\newcommand{\pw}{\wp}

\newcommand{\De}{\Delta}
\newcommand{\Om}{\Omega}
\newcommand{\al}{\alpha}

\newcommand{\om}{\omega}

\newcommand{\isdef}{\mathrel{:=}}

\newcommand{\Dom}{\mathsf{Dom}}
\newcommand{\Ran}{\mathsf{Ran}}


%% file: fc-commands.tex
\newcommand{\nat}{\ensuremath{\mathbb{N}}\xspace}
\newcommand{\tup}[1]{\langle{#1}\rangle}
\newcommand{\ext}[1]{\llbracket#1\rrbracket}
\newcommand{\prop}{\mathsf{P}}
\newcommand{\qprop}{\mathsf{Q}}

\newcommand{\model}{\mathbb{T}}
\newcommand{\osmodel}{\mathbf{D}}
\newcommand{\compset}[1]{\{#1\}}
\newcommand{\mc}{\mathcal}
\newcommand{\mmodels}{\Vdash}
\newcommand{\umods}{\mathfrak{M}_1}
\newcommand{\rest}{|}
\newcommand{\omegaunrav}[1]{{#1}^\omega}
\newcommand{\unravel}[1]{{#1}^e}

\newcommand{\MC}{\ensuremath{\mu\ML}\xspace}
\newcommand{\AFMC}{\ensuremath{\mathrm{AFMC}}\xspace}
\newcommand{\contAFMC}{\ensuremath{\mu_{c}\ML}\xspace}
\newcommand{\WMSO}{\wmso}
\newcommand{\wmso}{\ensuremath{\mathrm{WMSO}}\xspace}
\newcommand{\owmso}{\ensuremath{\mathrm{WMSO}_1}\xspace}

\newcommand{\fo}{\ensuremath{\mathrm{FO}}\xspace}

\newcommand{\ofo}{\ensuremath{\fo_1}\xspace}
\newcommand{\foe}{\ensuremath{\mathrm{FOE}}\xspace}

\newcommand{\ofoe}{\ensuremath{\foe_1}\xspace}
\newcommand{\cont}[2]{#1\mathrm{C}_{#2}}
\newcommand{\cocont}[2]{#1\overline{\mathrm{C}}_{#2}}
\newcommand{\monot}[2]{#1\mathrm{M}_{#2}}
\newcommand{\mclass}{K}
\newcommand{\llang}{\ensuremath{\mathcal{L}}\xspace}
\newcommand{\trees}{\ensuremath{\mathcal{T}}}

\newcommand{\qu}{\ensuremath{\exists^\infty}\xspace}
\newcommand{\dqu}{\ensuremath{\forall^\infty}\xspace}
\newcommand{\wqu}{\ensuremath{\mathbf{W}}\xspace}
\newcommand{\lqu}{\ensuremath{\fo^\infty}\xspace}
\newcommand{\lque}{\ensuremath{{\foe}^\infty}\xspace}
\newcommand{\mlque}{\ensuremath{\mu\lque}\xspace}
\newcommand{\clque}{\ensuremath{\mu_{c}\lque}\xspace}
\newcommand{\olque}{\ensuremath{\lque_1}\xspace}

\newcommand{\dbnf}{\nabla}										
\newcommand{\dbnfofo}[1]{\dbnf_{\fo}(#1)}						
\newcommand{\dbnfofoe}[2]{\dbnf_{\foe}(#1,#2)}					
\newcommand{\dbnfolque}[3]{\dbnf_{\lque}(#1,#2,#3)}				
\newcommand{\dbnfinf}[1]{\dbnf_\infty(#1)}

\newcommand{\mondbnfofo}[2]{\dbnf^{#2}_\fo(#1)}					
\newcommand{\mondbnfofoe}[3]{\dbnf^{#3}_{\foe}(#1,#2)}			
\newcommand{\mondbnfolque}[4]{\dbnf^{#4}_{\lque}(#1,#2,#3)}		
\newcommand{\mondbnfinf}[2]{\dbnf^{#2}_\infty(#1)}

\newcommand{\posdbnfofo}[1]{\dbnf^+_{\fo}(#1)}					
\newcommand{\posdbnfofoe}[2]{\dbnf^+_{\foe}(#1,#2)}				
\newcommand{\posdbnfolque}[3]{\dbnf^+_{\lque}(#1,#2,#3)}		
\newcommand{\posdbnfinf}[1]{\dbnf^+_\infty(#1)}

\newcommand{\fovar}{\mathsf{iVar}}

\newcommand{\ass}{g}
\newcommand{\val}{V}
\newcommand{\tscolors}{\sigma}
\newcommand{\tsval}{\tscolors^\flat}

\newcommand{\aut}{\ensuremath\mathbb{A}\xspace}
\newcommand{\tmap}{\Delta}
\newcommand{\pmap}{\Omega}
\newcommand{\ord}{\preceq}

\newcommand{\lthen}{\to}

\newcommand{\inc}{\sqsubseteq}
\newcommand{\here}[1]{{\Downarrow}#1}
\newcommand{\arediff}[1]{\mathrm{diff}(#1)}
\newcommand{\foeq}{\approx}
\newcommand{\tcont}{\vartriangle}
\newcommand{\tmono}{\circ}
\newcommand{\vlist}[1]{\vv{#1}}

\newcommand{\efgame}{\mathrm{EF}}
\newcommand{\egame}{\mathcal{E}}
\newcommand{\eloise}{\ensuremath{\exists}\xspace}
\newcommand{\abelard}{\ensuremath{\forall}\xspace}
\newcommand{\pmatches}[2]{\mathsf{PM}^{#1}_{#2}}
\newcommand{\win}{\text{\sl Win}}

%% file: fz-commands.tex
\renewcommand{\sc}{\scshape}
\newenvironment{Iff-RL}{\textbf{($\Rightarrow$)} }{\bigskip}
\newenvironment{Iff-LR}{\textbf{($\Leftarrow$)} }{}
\newcommand{\m}{\mathit}
\newcommand{\mb}{\mathbb}
\def \: {\colon}
\def \p {\wp}
\newcommand{\lift}[1]{{#1}^{\wp}}

\newcommand{\shA}{A^{\sharp}}
\newcommand{\shai}{a_{I}^{\sharp}}
\newcommand{\shDe}{\Delta^{\sharp}}

\newcommand{\V}{\tscolors}
\newcommand{\R}[1]{R[#1]}
\newcommand{\DeltaProj}{\widetilde{\Delta}} 
\newcommand{\MSO}{\mso}
\newcommand{\WFMSO}{\mathrm{WFMSO}}

\newcommand{\f}{F} 

\newcommand{\mgFOETrsym}{\star} 
\newcommand{\mgFOETr}[1]{(#1)^\mgFOETrsym} 

%% file: abstract.tex

\begin{abstract}
We prove that the bisimulation-invariant fragment of weak monadic second-order logic (WMSO) is equivalent to the fragment of the modal $\mu$-calculus where the application of the least fixpoint operator $\mu p.\varphi$ is restricted to formulas $\varphi$ that are continuous in $p$. Our proof is automata-theoretic in nature; in particular, we introduce a class of automata characterizing the expressive power of WMSO over tree models of arbitrary branching degree. The transition map of these automata is defined in terms of a logic $\olque$ that is the extension of first-order logic with a generalized quantifier $\qu$, where $\qu x. \phi$ means that there are infinitely many objects satisfying $\phi$. An important part of our work consists of a model-theoretic analysis of $\olque$.
\end{abstract}

%% file: introduction.tex

\section{Introduction}
\label{sec:intro}

\paragraph{Expressiveness modulo bisimilarity.}
This paper concerns the relative expressive power of some languages used for
describing properties of pointed labelled transitions systems, or Kripke
models.
The interest in such expressiveness questions stems from applications where
these structures model computational processes, and bisimilar pointed
structures represent the \emph{same} process.
Seen from this perspective, properties of transition structures are relevant
only if they are invariant under bisimilarity.
This explains the importance of bisimulation invariance results of the form
\begin{equation*}
M \equiv L / {\bis} \text{ (over $K$)}
\end{equation*}
stating that,  if one restricts attention to a certain class $K$ of transition
structures, one language $M$ is expressively complete with respect to the
relevant (i.e., bisimulation-invariant) properties that can be formulated in
another language $L$.
In this setting, generally $L$ is some rich yardstick formalism such as
first-order or monadic second-order logic, and $M$ is some modal-style
fragment of $L$, usually displaying much better computational behavior
than the full language $L$.

A seminal result in the theory of modal logic is van Benthem's Characterization
Theorem~\cite{vanBenthemPhD}, stating that every bisimulation-invariant
first-order formula $\alpha(x)$ is actually equivalent to (the standard
translation of) a modal formula:
\begin{equation*}
\ML \equiv \m{FO}/{\bis} \text{ (over the class of all LTSs)}.
\end{equation*}
Over the years, a wealth of variants of the Characterization Theorem have been
obtained.
For instance, Rosen proved that van Benthem's theorem is one of the few
preservation results that transfers to the setting of finite
models~\cite{rose:moda97}; for a recent, rich source of van Benthem-style
characterization results, see Dawar \& Otto~\cite{DawarO09}.
In this paper we are mainly interested is the work of Janin \&
Walukiewicz~\cite{Jan96}, who extended van Benthem's result to the setting
of fixpoint logics, by proving that the modal $\mu$-calculus ($\MC$) is the
bisimulation-invariant fragment of monadic second-order
logic (\yvMSO):
\begin{equation*}
\MC \equiv \yvMSO/{\bis} \text{ (over the class of all LTSs)}.
\end{equation*}

\paragraph{Bisimulation invariance for $\yvWMSO$.}
The yardstick logic that we consider in this paper is \emph{weak} monadic
second-order logic (\yvWMSO), a variant of monadic second-order logic where
the second-order quantifiers range over \emph{finite} subsets of the
transition structure rather than over arbitrary ones.
Our target will be to identify the bisimulation-invariant fragment of this
logic \yvWMSO.

Before moving on, we should stress the role of the ambient class $K$ in
bisimulation-invariance results.
Of particular importance in the setting of weak monadic second-order logic is
the difference between structures of finite versus arbitrary branching degree.
In the case of finitely branching models, it is not very hard to show that
$\yvWMSO$ is a (proper) fragment of \yvMSO, and it seems to be folklore that
$\yvWMSO/{\bis}$ corresponds to \AFMC, the alternation-free fragment of the
modal $\mu$-calculus.
For binary trees, this result was proved by Arnold \& Niwi{\'n}ski in
\cite{ArnoldN01}.
In the case of structures of arbitrary branching degree, however, $\yvWMSO$
and $\yvMSO$ have \emph{incomparable} expressive power.
The fact that, in particular, $\yvWMSO$ does not correspond to a fragment of
\yvMSO, is witnessed by the class of infinitely branching structures, which
is clearly \yvWMSO-definable, cannot be defined in \yvMSO, since every
\yvMSO-definable class of trees contains a finitely branching
tree.\footnote{As remarked in \cite{CateF11}, this follows from the automata characterization of MSO given in \cite{Walukiewicz96}.} 
For this reason, the relative expressive power of $\yvWMSO/{\bis}$ and
$\yvMSO/{\bis}$ is not a priori clear.
However, it is reasonable to think that $\yvWMSO/{\bis}$ is strictly \emph{weaker} than \AFMC: the class of well-founded trees, which is definable in \AFMC by
the simple formula $\mu p. \Box p$, is not definable in \yvWMSO.\footnote{This follows from the fact that $\yvWMSO$ can only define properties of trees that, from a topological point of view, are Borel, which is not the case of the class of trees defined by $\mu p. \Box p$--- see e.g. \cite{CateF11}.}
Incidentally, the question whether, conversely, there is a natural logic of
which the
bisimulation-invariant fragment corresponds to \AFMC was answered positively
by three of the present authors in~\cite{DBLP:conf/lics/FacchiniVZ13}, where they introduced
another variant of \yvMSO, called well-founded $\yvMSO$ (\yvWFMSO), and proved
that $\yvWFMSO/{\bis} \equiv \AFMC$ (over the class of all LTSs).

The main result that we shall prove in this paper states that the
bisimulation-invariant fragment of $\yvWMSO$ is equivalent to a certain,
fragment $\yvF$ of the modal $\mu$-calculus.
\begin{equation}
\label{eq-main}
\yvF \equiv \yvWMSO/{\bis}  \text{ (over the class of all LTSs)}.
\end{equation}
This fragment $\contAFMC$, which is strictly weaker than the alternation-free fragment
of $\muML$, is characterized by a certain restriction on the application of
fixpoint operators, which involves the notion of \emph{(Scott) continuity}.

Continuity, an interesting property that features naturally in the semantics
of many (fixpoint) logics, in fact plays a key role throughout this paper.
For its definition, we consider how the meaning $\haak{\phi}^{\model} \sse T$
of a formula $\phi$ in some structure $\model$ (with domain $T$) depends on the meaning of a fixed
proposition letter or monadic predicate symbol $p$.
This dependence can be formalized as a map $\phi^{\model}_{p}: \wp(T) \to
\wp (T)$, and if this map satisfies the condition
\begin{equation}
\label{eq-Sc}
\phi^{\model}_{p}(X) = \bigcup \Big\{ \phi^{\model}_{a}(X') \mid X'
\text{ is a finite subset of } X \Big\},
\end{equation}
we say that $\phi$ is \emph{continuous in $p$}.
The topological terminology stems from the observation that \eqref{eq-Sc}
expresses the continuity of the map $\phi^{\model}_{p}$ with respect to the Scott
topology on $\pw(T)$. 
If we look at concrete cases, this definition can be given a different reading:
if $\phi$ is a formula of the
modal $\mu$-calculus, \eqref{eq-Sc} means that $\phi$ holds at some state $s$
of $\model$ iff we can shrink the interpretation of the proposition letter $p$
to some finite subset of the original interpretation, in such a way that
$\phi$ holds at $s$ in the modified version of $\model$.


A syntactic \emph{characterization} of this property for the modal $\mu$-calculus
was obtained by Fontaine~\cite{Fontaine08,FV12}, and the definition of our fragment $\yvF$
uses this characterization as follows:
whereas in the full language of $\muML$ the only syntactic condition on the
formation of a formula $\mu p. \phi$ is that $\phi$ is \emph{positive} in $p$,
for the fragment $\contAFMC$ this condition is strengthened to the requirement that
$\phi$ is (syntactically) \emph{continuous} in $p$.
More precisely, the fragment $\yvF$ is defined as follows:

\begin{definition}
For each set $\qprop$ of
proposition letters, the fragment $\cont{\MC}{\qprop}$ of $\MC$ which is \emph{continuous in $\qprop$}
is given by the simultaneous induction
\begin{equation*}
   \varphi ::= q
   \mid \psi
   \mid \varphi \lor \varphi
   \mid \varphi \land \varphi
   \mid \Diamond \varphi
   \mid \mu p.\alpha
\end{equation*}
where $p\in\prop$, $q \in \qprop$, $\psi$ is a $\qprop$-free $\MC$-formula, and
$\alpha \in \cont{\MC}{\qprop\cup\{p\}}$.
The formulas of the fragment $\contAFMC$ are then given by the following induction:
\begin{equation*}
   \varphi ::= p \mid \lnot \varphi
    \mid \varphi \lor \varphi
    \mid  \Diamond \varphi
    \mid \mu p.\alpha
\end{equation*}
where $p \in \prop$, and $\alpha \in \cont{\MC}{p}$.
\end{definition}

In fact we will prove, analogous to the result by Janin \& Walukiewicz,
the following strong version of the characterization result~\eqref{eq-main},
which provides an explicit translation, mapping any
bisimulation-invariant formula $\phi$ in $\yvWMSO$ to an equivalent formula
$\phi^{\bullet}$ in $\contAFMC$.

\begin{theorem}
\label{t:m1}
There are effective translations $(-)^{\bullet}: \yvWMSO \to \contAFMC$ and
$(-)_{\bullet}: \contAFMC \to \yvWMSO$ such that 
\begin{enumerate}
\item A formula $\phi$ of $\yvWMSO$ is
bisimulation invariant if and only if $\phi \equiv \phi^{\bullet}$, and
\item $\psi \equiv \psi_{\bullet}$ for every formula $\psi \in \contAFMC$.
\end{enumerate}
\end{theorem}

To see how this theorem implies \eqref{eq-main}, observe that part (i)
shows that $\yvWMSO/{\bis} \leq \yvF$.
Part (ii) states that $\yvF \leq \WMSO$, so combined with the fact that
every formula in $\yvF \sse \MC$ is bisimulation invariant, this gives the
converse, $\yvF \leq \yvWMSO/{\bis}$.

\paragraph{Automata for $\yvWMSO$.}
As usual in this research area, our proof will be automata-theoretic in
nature.
More specifically, as the second main contribution of this paper, we
introduce a new class of parity automata that exactly captures the expressive
power of $\yvWMSO$ over the class of tree models of arbitrary branching degree.

Before we turn to a description of these automata, we first have a look at the
automata, introduced by Walukiewicz~\cite{Walukiewicz96}, corresponding to $\yvMSO$
(over tree models).
Fixing the set of proposition letters of our models as $\Prop$, we think of
$\wp(\Prop)$ as an \emph{alphabet} or set of \emph{colors}.
We can then define an \yvMSO-automaton as a tuple $\bbA = \tup{A, \De,
\Om, \ai}$, where $A$ is a finite set of states, $\ai$ an initial state, and $\Om:
A \to \bbN$ is a parity function.
The transition function $\De$ maps a pair $(a,c) \in A \times \p(\Prop)$ to a
sentence in the first-order language (with equality) $\ofoe(A)$, of which the
state space $A$ provides the set of (monadic) predicates.
For a more precise definition, let $\ofoe^+(A)$ denote the set of
those sentences in $\ofoe(A)$ where all predicates in $A$ occur only positively;
we require that $\De: A \times \wp(\Prop) \to \ofoe^+(A)$.

We shall refer to $\ofoe$ as the \emph{one-step language} of $\yvMSO$-automata,
and denote the class of $\yvMSO$-automata with $\yvAut(\ofoe)$.
The automata that we consider in this article run on labelled transition systems
and decide wether to accept or reject them. To take such decision we associate
an acceptance game for an
$\yvMSO$-automaton $\bbA$ and a transition system $\model$.
A match of this game consists of two players, $\eloi$ and $\abel$, moving a
token from one position to another.
When such a match arrives at a so-called \emph{basic} position, i.e., a
position of the form $(a,t) \in A \times T$, the players consider the
sentence $\De(a,c_{t}) \in \ofoe^+(A)$, where $c_{t} \in \wp(\Prop)$ is the color
of $t$ (that is, the set of proposition letters true at $t$).
At this position $\eloi$ has to turn the set $R[t]$ of successors of
$s$ into a \emph{model} for the formula $\De(a,c_{t})$ by coming up with an
interpretation $I$ of the monadic predicates $a \in A$ as subsets of $R[s]$,
so that the resulting first-order structure $(R[s],I)$ makes the
formula $\De(a,c_{t})$ true.


Walukiewicz's key result linking $\yvMSO$ to $\yvAut(\ofoe)$ states that
\begin{equation}
\yvMSO \equiv \yvAut(\ofoe)
 \text{ (over tree models)},
\end{equation}
and the proof of this result proceeds by inductively showing that every formula
$\phi$ in $\yvMSO$ can be effectively transformed into an equivalent
automaton $\bbA_{\phi} \in \yvAut(\ofoe)$.
For the details of this construction, a fairly intricate analysis of the
one-step logic $\ofoe$ is required, crucially involving various normal
forms of the sentences of $\ofoe(A)$.

In order to adapt this approach to the setting of \WMSO, observe that by
K\"onig's lemma, a subset of a tree $\model$ is finite iff it is both a subset of
a finitely branching subtree of $\model$ and \emph{noetherian}, that is, a subset
of a subtree of $\model$ that has no infinite branches.
This suggests that we may change the definition of $\yvMSO$-automata into one
of $\yvWMSO$-automata via two kinds of modifications, roughly speaking
corresponding to a horizontal and a vertical `dimension' of trees.

For the `vertical modification' we may turn to the literature on weak automata~\cite{MullerSaoudiSchupp92}.
The acceptance condition $\Om$ of a parity automaton $\bbA =
\tup{A, \De, \Om, \ai}$ is \emph{weak} if $\Om(a) = \Om(a')$ whenever
the states $a$ and $a'$ belong to the same strongly connected component
(SCC) of the automaton. To see that the notion of connected component is well-defined
observe that for $\bbA$ we can associate a directed graph on $A$
such that $a,b \in A$ are connected iff $b$ occurs in $\Delta(a,c)$ for some $c \in \wp(\prop)$.
Let $\yvwAut(\ofoe)$ denote the set of $\yvMSO$-automata with a weak parity
condition.
It was proved in \cite{Zanasi:Thesis:2012} (see also \cite{DBLP:conf/lics/FacchiniVZ13}) that
\begin{equation*}
\yvWFMSO \equiv \yvwAut(\ofoe) \text{ (over the class of all trees)},
\end{equation*}
with $\yvWFMSO$ denoting the earlier mentioned variant of $\yvMSO$ 
where second-order quantification is restricted to noetherian subsets of trees.
From this it easily follows that
\begin{equation*}
\yvWMSO \equiv \yvwAut(\ofoe) \text{ (over the class of finitely
branching trees)},
\end{equation*}
since the noetherian subsets of a finitely branching trees correspond to the
finite ones.
Over the class of all tree models, however, $\yvWMSO$ is \emph{not} equivalent
to $\yvwAut(\ofoe)$, 
as is witnessesed by the earlier mentioned
class of well-founded trees, which can be defined in $\AFMC \leq \yvWFMSO$,
but not in $\yvWMSO$.

The hurdle to take, in order to find automata for WMSO on trees of
\emph{arbitrary} branching degree, concerns the horizontal dimension; the
main problem lies in finding the right one-step language for
$\yvWMSO$-automata.
An obvious candidate for this language would be weak monadic second-order logic
itself, or more precisely, its variant $\owmso$ over the signature of monadic
predicates (corresponding to the automata states).
A very helpful observation, made by V\"a\"an\"anen~\cite{vaananen77}, states that
\[
\owmso \equiv \olque,
\]
where $\olque$ is the extension of $\ofoe$ with the generalized quantifier
$\qu$, where $\qu x. \phi$ meaning that there are \emph{infinitely} many
objects satisfying $\phi$.
Taking the \emph{full} language of $\owmso$ or $\olque$ as our one-step language
would give too much expressive power: since $\olque$ extends $\ofoe$,
we would find that, over tree models, $\yvwAut(\olque)$ extends
$\yvwAut(\ofoe)$, whereas we already saw that $\yvwAut(\ofoe) \equiv
\yvWFMSO$ is incomparable to $\yvWMSO$.
It is here that we will crucially involve the notion of \emph{continuity}.
The automata corresponding to $\yvWMSO$ will be of the form $\bbA = \tup{A, \De, \Om, \ai}$,
where the transition map $\De: A \times \wp(\Prop) \to
{\olque}^+(A)$ is subject to the following two constraints, for all $a,a' \in A$
belonging to the same strongly connected component of $A$:
\begin{description}
\itemsep 0 pt
\item[(weakness)] $\Om(a) = \Om(a')$, and
\item[(continuity)]
if $\Om(a)$ is odd (resp. even), then for each colour $c\in \wp(\Prop)$,
   $\De(a,c)$ is continuous (resp. co-continuous) in $a'$,
\end{description}
where co-continuity is a dual notion to continuity.
The class of these automata is denoted by $\yvcwAut(\olque)$.
Consequently, for a proper definition of these automata we need a
\emph{syntactic} characterization of the $\olque(A)$-sentences that are
(co-)continuous in one (or more) monadic predicates of $A$.

For this purpose, we conduct a fairly detailed model-theoretic study of the
logic $\olque$ which we consider to be the third main
contribution of our work.
Similar to the results for $\ofoe$, we provide normal forms for the
sentences of $\olque(A)$, and syntactic characterizations of the fragments
whose sentences are monotone (respectively continuous) in some monadic predicate $a \in A$.

To finish, we give constructions transforming \yvWMSO-formulas to
\yvWMSO-automata and vice-versa, witnessing that
\begin{equation}
\label{eq:m3}
\yvWMSO \equiv \yvcwAut(\olque) \text{ (over tree models)}.
\end{equation}

\paragraph{Proof of main result.}

To conclude our introduction we briefly sketch the proof of our main result, 
Theorem~\ref{t:m1}(1).
Roughly speaking, we follow the bisimulation-invariance proof by Janin \& 
Walukiewicz, which revolves around relating two distinct types of automata, 
which correspond, respectively, to the logics $\yvMSO$ and $\MC$.
More precisely, these two automaton types are given as $\yvAut(\ofoe)$ and
$\yvAut(\ofo)$, where the one-step languages are first-order 
logic respectively with and without equality.
What we will add to their proof is the insight from~\cite{Venxx} that the
required relation between $\yvAut(\ofoe)$ and $\yvAut(\ofo)$ already follows
from results relating the respective one-step languages.

In our setting, we need to identify automata corresponding to the fragment
$\yvF$.
For this purpose we introduce the class $\yvcwAut(\ofo)$ consisting of those
automata in $\yvAut(\ofo)$ that satisfy similar weakness and continuity 
conditions as the ones in $\yvcwAut(\olque)$:
\begin{equation}
\label{eq:autF}
\yvF \equiv \yvcwAut(\ofo) \text{ (over the class of all LTSs)}.
\end{equation}

As the key step in our proof then, we will provide a translation 
$(-)^{\bullet}: \olque \to \ofo$ which naturally induces a transformation 
$(-)^{\bullet}: \yvcwAut(\olque) \to \yvcwAut(\ofo)$.
As a consequence of the nice model-theoretic properties of the translation at 
the one-step level, the automaton transformation satisfies, for all 
transition systems $\bbT$:
\begin{equation}
\label{eq:crux-i}
\bbA^{\bullet} \text{ accepts } \bbT \text{ iff } \bbA \text{ accepts 
} \omegaunrav{\bbT}
\end{equation}
where $\omegaunrav{\bbT}$ is the $\om$-unravelling of $\bbT$.
It easily follows from \eqref{eq:crux-i} that a \yvWMSO-automaton $\bbA$
is bisimulation invariant iff $\bbA\equiv \bbA^{\bullet}$, and so 
Theorem~\ref{t:m1}(1) follows by \eqref{eq:m3} and \eqref{eq:autF}.

\paragraph{Overview of paper.}
In the next section we give a precise definition of the preliminaries required to understand this article. In Section~\ref{sec:onestep} we define the one-step logics that will be used through the paper and give normal forms and syntactic characterizations of their monotone and (co-)continuous fragments. In Section~\ref{sec:aut} we formally define \WMSO-automata and show that from every \WMSO-formula we can construct an equivalent \WMSO-automaton. In Section~\ref{sec:aut-to-formula_wmso} we prove the converse, that is, for every \WMSO-automaton we can construct an equivalent \WMSO-formula, this finishes the automata characterization of \WMSO over tree models. Finally, in Section~\ref{sec:char} we prove the main result of the paper, namely that the fragment $\contAFMC$ is the bisimulation-invariant fragment of \WMSO.

%% file: preliminaries.tex

\subsection{Transition systems and trees} \label{ssec:prelim_trees}

Throughout this article we fix a set $\prop$ of elements that will be called
\emph{proposition letters} and denoted with small Latin letters $p, q, \ldots$ .
We denote with $C$ the set $\wp (\prop)$ of \emph{labels} on $\prop$; it will be
convenient to think of $C$ as an \emph{alphabet}.
Given a binary relation $R \subseteq X \times Y$, for any element $x \in X$,
we indicate with $R[x]$ the set $\compset{ y \in Y \mid (x,y) \in R}$ while $R^+$
and $R^{*}$ are defined respectively as the transitive closure of~$R$ and
the reflexive and transitive closure of~$R$. The set $\Ran(R)$ is defined as $\bigcup_{x\in X}R[x]$.

A \emph{$C$-labeled transition system} (LTS) is a tuple $\model = \tup{T,R,\tscolors,s_I}$ where
$T$ is the universe or domain of $\model$, $\tscolors:T\to\wp(\prop)$ is a marking,
$R\subseteq T^2$ is the accessibility relation and $s_I \in T$ is a distinguished node.
We use $|\model|$ to denote the domain of $\model$.
%
Observe that the marking ${\tscolors:T\to\wp(\prop)}$ can be seen as a valuation $\tsval:\prop\to\wp (T)$ given by $\tsval(p) = \{s \in T \mid p\in \tscolors(s)\}$.

%
A \emph{$C$-tree} is a LTS in which every node can
be reached from $s_I$, and every node except $s_I$ has a unique predecessor;
the distinguished node $s_I$ is called the \emph{root} of $\model$.
Each node $s \in T$ uniquely defines a subtree of $\model$ with carrier
$R^{*}[s]$ and root $s$. We denote this subtree by ${\model.s}$.
We use the term \emph{tree language} as a synonym of class of $C$-trees.

The tree unravelling of an LTS $\model$ is given by $\unravel{\model} := \tup{T_P,R_P,\tscolors',s_I}$ where $T_P$ is the set of finite paths in $\model$ stemming from $s_I$, $R_P(t,t')$ iff $t'$ is an extension of $t$ and the color of a path $t\in T_P$ is given by the color of its last node in $T$. The $\omega$-unravelling $\omegaunrav{\model}$ of $\model$ is an unravelling which has $\omega$-many copies of each node different from the root.

A \emph{$p$-variant} of a transition system $\model = \tup{T,R,\tscolors,s_I}$
is a $\p (\prop\cup\{p\})$-transition system $\tup{T,R,\tscolors',s_I}$
such that $\tscolors'(s)\setminus\{p\} = \tscolors(s)$ for all $s \in T$.
Given a set $S \subseteq T$, we let $\model[p\mapsto S]$ denote the $p$-variant
where $p \in \tscolors'(s)$ iff $s \in S$.

Let $\varphi \in \llang$ be a formula of some logic $\llang$,
we use $\ext{\varphi} = \compset{\model \mid \model \models \varphi}$ to denote the class
of transition systems that make $\varphi$ true.
A class $\mclass$ of transition systems is said to be \emph{$\llang$-definable} if there
is a formula $\varphi \in \llang$ such that $\ext{\varphi} = \mclass$.
We use the notation $\varphi \equiv \psi$ to mean that $\ext{\varphi} = \ext{\psi}$ and given two logics
$\llang, \llang'$ we use $\llang \equiv \llang'$ when the $\llang$-definable and $\llang'$-definable
classes of models coincide.

\textit{Convention.}
Throughout this paper, we will only consider transition systems $\model$
in which $R[s]$ is non-empty for every node $s \in T$.
In particular this means that every tree we consider is \emph{leafless}.
All our results, however, can easily be lifted to the general case.

\subsection{Games}

We introduce some terminology and background on infinite games.
All the games that we consider involve two players called \emph{\'Eloise}
($\exists$) and \emph{Abelard} ($\forall$).
In some contexts we refer to a player $\Pi$ to specify a
a generic player in $\{\exists,\forall\}$.
Given a set $A$, by $A^*$ and $A^\omega$ we denote respectively the set of
words (finite sequences) and streams (or infinite words) over $A$.

A \emph{board game} $\mc{G}$ is a tuple $(G_{\exists},G_{\forall},E,\win)$,
where $G_{\exists}$ and $G_{\forall}$ are disjoint sets whose union
$G=G_{\exists}\cup G_{\forall}$ is called the \emph{board} of $\mc{G}$,
$E\subseteq G \times G$ is a binary relation encoding the \emph{admissible
moves}, and $\win \subseteq G^{\omega}$ is a \emph{winning condition}.
An \emph{initialized board game} $\mc{G}@u_I$ is a tuple
$(G_{\exists},G_{\forall},u_I, E,\win)$ where
$u_I \in G$ is the
\emph{initial position} of the game.
When $\win$ is  given by a parity function
$\pmap: G \to \omega$ we say that $\mc{G}$ is a parity game and sometimes
simply write $\mc{G}=(G_{\exists},G_{\forall},E,\pmap)$.

Given a board game $\mc{G}$, a \emph{match} of $\mc{G}$ is simply a path
through the graph $(G,E)$; that is, a sequence $\pi = (u_i)_{i< \alpha}$ of
elements of $G$, where $\alpha$ is either $\omega$ or a natural number,
and $(u_i,u_{i+1}) \in E$ for all $i$ with $i+1 < \alpha$.
A match of $\mc{G}@u_{I}$ is supposed to start at $u_{I}$.
Given a finite match $\pi = (u_i)_{i< k}$ for some $k<\omega$, we call
$\m{last}(\pi) := u_{k-1}$ the \emph{last position} of the match; the
player $\Pi$ such that $\m{last}(\pi) \in G_{\Pi}$ is supposed to move
at this position, and if $E[\m{last}(\pi)] = \emptyset$, we say that
$\Pi$ \emph{got stuck} in $\pi$.
A match $\pi$ is called \emph{total} if it is either finite, with one of the
two players getting stuck, or infinite. Matches that are not total are called
\emph{partial}.
Any total match $\pi$ is \emph{won} by one of the players:
If $\pi$ is finite, then it is won by the opponent of the player who gets stuck.
Otherwise, if $\pi$ is infinite, the winner is $\exists$ if $\pi \in
\win$, and $\forall$ if $\pi \not\in \win$.

Given a board game $\mc{G}$ and a player $\Pi$, let $\pmatches{G}{\Pi}$ denote
the set of partial matches of $\mc{G}$ whose last position belongs to player
$\Pi$.
A \emph{strategy for $\Pi$} is a function $f:\pmatches{G}{\Pi}\to G$.
A match $\pi  = (u_i)_{i< \alpha}$ of $\mc{G}$ is
\emph{$f$-guided} if for each $i < \alpha$ such that $u_i \in G_{\Pi}$ we
have that $u_{i+1} = f(u_0,\dots,u_i)$.
Let $u \in G$ and a $f$ be a strategy for $\Pi$.
We say that $f$ is a \emph{surviving strategy} for $\Pi$ in $\mc{G}@u$ if
\begin{enumerate}
  \item[(i)] For each $f$-guided partial match $\pi$ of $\mc{G}@u$, if $\m{last}(\pi)$
  is in $G_{\Pi}$ then $f(\pi)$ is legitimate, that is,
  $(\m{last}(\pi),f(\pi)) \in E$.
\end{enumerate}
We say that $f$ is a \emph{winning strategy} for $\Pi$ in $\mc{G}@u$ if, additionally, 
%
\begin{enumerate}
  \item[(ii)] $\Pi$ wins each $f$-guided total match of $\mc{G}@u$.
\end{enumerate}
If $\Pi$ has a winning winning strategy for $\mc{G}@u$ then $u$ is called a \emph{winning position} for $\Pi$ in $\mc{G}$.
The set of positions of $\mc{G}$ that are winning for $\Pi$ is denoted by $\win_{\Pi}(\mc{G})$.
A strategy $f$ is called \emph{positional} if $f(\pi) = f(\pi^{\prime})$ for each $\pi,\pi^{\prime} \in \Dom(f)$ with $\m{last}(\pi) = \m{last}(\pi^{\prime})$.
A board game $\mc{G}$ with board $G$ is \emph{determined} if $G = \win_{\exists}(\mc{G}) \cup \win_{\forall}(\mc{G})$, that is, each $u \in G$ is a winning position for one of the two players.

\begin{fact}[Positional Determinacy of Parity Games~\cite{EmersonJ91,Mostowski91Games}]
\label{THM_posDet_ParityGames}
For each parity game $\mc{G}$, there are positional strategies $f_{\exists}$
and $f_{\forall}$ respectively for player $\exists$ and $\forall$, such that
for every position $u \in G$ there is a player $\Pi$ such that $f_{\Pi}$ is a
winning strategy for $\Pi$ in $\mc{G}@u$.
\end{fact}
From now on, we always assume that each strategy we work with in parity games
is positional. Moreover, we think of a positional strategy $f_\Pi$ for player $\Pi$
as a function $f_\Pi:G_\Pi\to G$.

\input{paraut}

\subsection{Weak monadic second-order logic}

\begin{definition}
The \emph{weak monadic second-order logic} on a set of predicates $\prop$ is given by
$$\varphi ::= \here{p} \mid p \inc q \mid R(p,q) \mid \lnot\varphi \mid \varphi\lor\varphi \mid \exists p.\varphi$$
where $p,q \in \prop$. We denote this logic by $\wmso(\prop)$ and omit $\prop$ when the set of proposition letters is clear from context.
We  adopt the standard convention that no letter is both free and bound in $\varphi$.
\end{definition}

\begin{definition}\label{def:wmso}
Let $\model = \tup{T,R,\tscolors, s_I}$ be a LTS, the semantics of \wmso is defined as follows
\begin{align*}
\model \models \here{p} & \quad\text{ iff }\quad  \tsval(p) = \compset{s_I} \\
\model \models p \inc q & \quad\text{ iff }\quad  \tsval(p) \subseteq \tsval(q) \\
\model \models R(p,q) & \quad\text{ iff }\quad  \text{for every $s\in \tsval(p)$ there exists $t\in \tsval(q)$ such that $sRt$} \\
\model \models \lnot\varphi & \quad\text{ iff }\quad  \model \not\models \varphi \\
\model \models \varphi\lor\psi & \quad\text{ iff }\quad  \model \models \varphi \text{ or } \model \models \psi \\
\model \models \exists p.\varphi & \quad\text{ iff }\quad  \text{there is a \emph{finite} set $X \subseteq_\omega T$ such that $\model[p\mapsto X] \models \varphi$}
\end{align*}
\end{definition}


\begin{remark}
The reader may have expected a more standard two-sorted language for second-order logic, for example given by
$$
\varphi ::= p(x)
\mid R(x,y)
\mid x \foeq y
\mid \neg \varphi
\mid \varphi \lor \varphi
\mid \exists x.\varphi
\mid \exists p.\varphi
$$%
where $p \in \prop$, $x,y \in \fovar$ (individual variables), 
and $\foeq$ is the symbol for equality.
Both definitions can be proved to be equivalent, however, we choose to keep Definition~\ref{def:wmso} as it will be better suited to work with in the context of automata.
\end{remark}

\subsection{First-order logic with generalized quantifiers}
In this subsection we introduce an extension of first-order logic with so called generalized quantifiers. Our interest in this extension stems from the fact that it will allow us to define a variant of first-order logic that is expressively equivalent to weak monadic second order (see Section~\ref{sec:onestep}) and has nice technical features such as a normal form theorem.

Mostowski~\cite{Mostowski1957} defined (unary) generalized quantifiers as follows: a (unary) generalized quantifier ${\mathcal Q}$ is a collection of pairs $(J, X)$ with $X \subseteq J$, and satisfying the following condition
$$\text{If } \big( (J,X)\in {\mathcal Q}, \ |X|=|Y| \ \land \ | J \setminus X|=|K \setminus Y|\big) \text{ then } (K,Y)\in {\mathcal Q}.$$

%
\noindent The semantics of ${\mathcal Q}$ is defined by the following condition
$$
\model\models {\mathcal Q}x. \phi(x, \vlist{s}) \text{ iff } (T,\{t \mid \model\models \phi(t,\vlist{s}) \}) \in {\mathcal Q}.
$$
where we use $\vlist{s} := [s_1,\dots,s_n]$ to denote sequences or vectors of elements.

For the rest of this article we will focus on the generalized quantifier $\qu$ expressing that there exist infinitely many elements satisfying a certain condition. Formally, it is defined as
\[ \qu := \{(J,X) \mid |X| \geq \aleph_0\}.\]
The dual of $\qu$ is $\dqu =\{(J,X) \mid |J\setminus X| < \aleph_0\}$. It is worth observing what is the intended meaning of this quantifier: $\dqu x.\varphi$ expresses that there are \emph{at most finitely many} elements \emph{falsifying} the formula $\varphi$.
The extension of first-order logic with equality ($\foe$) or without equality ($\fo$)
obtained by adding $\qu$ to the corresponding first-order language is denoted respectively by $\lqu$ and $\lque$.



\subsection{The Modal $\mu$-Calculus.}\label{subsec:mu}
The language of the modal $\mu$-calculus ($\MC$) on $\prop$ is given by the following grammar:
\begin{equation*}
    \varphi\ ::= q \mid \neg q \mid \varphi \land \varphi \mid
    \varphi \lor \varphi \mid  \Diamond \varphi \mid \Box \varphi \mid
    \mu p.\varphi \mid \nu p.\varphi
\end{equation*}
%
where $p,q \in \prop$ and $p$ is positive in $\varphi$ (i.e., $p$ is not negated).
We use the standard convention that no variable is both free and bound in a formula and that every bound variable is fresh.
Let $p$ be a bound variable occuring in some formula $\varphi \in \MC$, we use $\delta_p$ to denote the binding definition of $p$, that is, the formula such that either $\mu p.\delta_p$ or $\nu p.\delta_p$ are subformulas of $\varphi$.

The semantics of this language is completely standard. Let $\model = \tup{T,R,\tscolors, s_I}$ be a transition system and $\varphi \in \MC$. We inductively define the \emph{meaning} $\ext{\varphi}^{\model}$ which includes the following clauses for the least $(\mu)$ and greatest ($\nu$) fixpoint operators:
\begin{align*}
  \ext{\mu x.\psi}^{\model}  & :=   \bigcap \{S \subseteq T \mid S \supseteq \ext{\psi}^{\model[x\mapsto S]} \}  \\
  \ext{\nu x.\psi}^{\model}  & :=   \bigcup \{S \subseteq T \mid S \subseteq \ext{\psi}^{\model[x\mapsto S]} \}
\end{align*}
We say that $\varphi$ is \emph{true} in $\model$ (notation $\model \mmodels \varphi$) iff $s_I \in \ext{\varphi}^{\model}$.

We will now describe the semantics defined above in game-theoretic terms. That is,
we will define the evaluation game $\egame(\varphi,\model)$ associated with a formula $\varphi \in \MC$ and a transition system $\model$. This game is played by two players (\eloise and \abelard) moving through positions $(\xi,s)$ where $\xi$ is a subformula of $\varphi$ and $s \in T$.
\begin{table}[h]
\centering
\begin{tabular}{|l|c|l|c|}
  \hline
  Position & Player & Admissible moves\\
  \hline
  $(\psi_1 \vee \psi_2,s)$ & $\exists$ & $\{(\psi_1,s),(\psi_2,s) \}$ \\
  $(\psi_1 \wedge \psi_2,s)$ & $\forall$ & $\{(\psi_1,s),(\psi_2,s) \}$ \\
  $(\Diamond\varphi,s)$ & $\exists$ & $\{(\varphi,t)\ |\ t \in R[s] \}$ \\
  $(\Box\varphi,s)$ & $\forall$ & $\{(\varphi,t)\ |\ t \in R[s] \}$ \\
  $(\mu p.\varphi,s)$ & $-$ & $\{(\varphi,s) \}$ \\
  $(\nu p.\varphi,s)$ & $-$ & $\{(\varphi,s) \}$ \\
  $(p,s)$ with $p$ bound in $\varphi$ & $-$ & $\{(\delta_p,s) \}$ \\
  $(\lnot q,s) \in \prop \times T$, $q$ free in $\varphi$ and $q \notin \tscolors(s)$ & $\forall$ & $\emptyset$\\
  $(\lnot q,s) \in \prop \times T$, $q$ free in $\varphi$ and $q \in \tscolors(s)$ & $\exists$ & $\emptyset$\\
  $(q,s) \in \prop \times T$, $q$ free in $\varphi$ and $q \in \tscolors(s)$ & $\forall$ & $\emptyset$\\
  $(q,s) \in \prop \times T$, $q$ free in $\varphi$ and $q \notin \tscolors(s)$ & $\exists$ & $\emptyset$\\
  \hline
\end{tabular}
\caption{Evaluation game for the modal $\mu$-calculus}
\label{egame_mucalc}
\end{table}
In an arbitrary position $(\xi,s)$ it is useful to think of
\eloise trying to show that $\xi$ is true at $s$, and of \abelard of trying to convince her that $\xi$ is false at $s$. The rules of the evaluation game are given in Table~\ref{egame_mucalc}.
Every finite match of this game is lost by the player that got stuck. To give a winning condition for an infinite match let $p$ be, of the bound variables of $\varphi$ that get unravelled infinitely often, the one that is the highest in the syntactic tree of $\varphi$. The winner of the match is \abelard if $p$ is a $\mu$-variable and \eloise if $p$ is a $\nu$-variable. We say that $\varphi$ is true in $\model$ iff \eloise has a winning strategy in $\egame(\varphi,\model)$.

\bigskip
Formulas of the modal $\mu$-calculus are classified according to their
\emph{alternation depth}, which roughly is given as the maximal length of
a chain of nested alternating least and greatest fixpoint operators~\cite{Niwinski86}.
The \emph{alternation-free fragment} of the modal $\mu$-calculus~($\AFMC$) is the collection of
$\MC$-formulas without nesting of least and greatest fixpoint operators.

\begin{definition}
  Let $\varphi$ be a formula of the modal $\mu$-calculus. We say that $\varphi\in\AFMC$ iff for all subformulas $\mu p.\psi_1$ and $\nu q.\psi_2$ we have that $p$ is not free in $\psi_2$ and $q$ is not free in $\psi_1$.
\end{definition}

It is not difficult to see that, over arbitrary transition systems, this fragment is
less expressive than the whole $\MC$. That is, there is a $\MC$-formula
$\varphi$ such that $\ext{\varphi}$ is not $\AFMC$-definable~\cite{Park79}.

In order to properly define the fragment $\contAFMC \subseteq \AFMC$ which is of critical importance in this article, we are particularly interested in the \emph{continuous} fragment of the modal $\mu$-calculus. As observed in Section~\ref{sec:intro}, the abstract notion of continuity can be given a concrete interpretation in the context of $\mu$-calculus.
\begin{definition}
Let $\varphi \in \MC$, and $q$ be a propositional variable. We say that \emph{$\varphi$ is continuous in $q$} iff for every transition system $\model$ there exists some finite $S \subseteq_\omega \tsval(q)$ such that
$$
\model \mmodels \varphi \quad\text{iff}\quad \model[q \mapsto S] \mmodels \varphi .
$$
\end{definition}

We can give a syntactic characterization of the fragment of $\MC$ that captures this property. Given a set $\qprop$ of propositional variables, we define the fragment of \MC \emph{continuous} in $\qprop$, denoted by $\cont{\MC}{\qprop}$, by induction in the following way
\begin{equation*}
   \varphi ::= q
   \mid \psi
   \mid \varphi \lor \varphi
   \mid \varphi \land \varphi
   \mid \Diamond \varphi
   \mid \mu p.\alpha
\end{equation*}
where $q \in \qprop$, $p \in \prop \setminus \qprop$, $\alpha \in \cont{\MC}{\qprop\cup\{p\}}$, and $\psi$ is a $\qprop$-free $\MC$-formula.

\begin{proposition}[\cite{Fontaine08,FV12}]\label{prop:FVcont}
A $\MC$-formula is continuous in $q$ iff it is equivalent to a formula in the fragment $\cont{\MC}{q}$.
\end{proposition}

Finally, we define $\contAFMC$ to be the fragment of $\MC$ where the use of the least fixed point operator is restricted to the continuous fragment. Formally, it is defined as follows.

\begin{definition}
Formulas of the fragment $\contAFMC$ are given by the following induction:
\begin{equation*}
   \varphi ::= q \mid \lnot \varphi
    \mid \varphi \lor \varphi
    \mid \Diamond \varphi
    \mid \mu p.\alpha
\end{equation*}
where $p,q \in \prop$, and $\alpha \in \cont{\MC}{p}$.
\end{definition}

\begin{proposition}
Let $\varphi \in \contAFMC$, the following hold
\begin{enumerate}[(1)]
\itemsep 0pt
\item $\varphi$ is an $\AFMC$-formula,
\item Every $\mu$-variable in $\varphi$ is existential (i.e., is only in the scope of diamonds), and dually every $\nu$-variable in $\varphi$ is universal (i.e., is only in the scope of boxes).
\end{enumerate}
\end{proposition}
\begin{proof}
Both points are proved by an easy induction on the complexity of a formula. For the first one,  it is enough to notice that if $\varphi \in \cont{\MC}{q} \cap \AFMC$, then $\mu q. \varphi \in \AFMC$ by definition of $\cont{\MC}{q} $.
\end{proof}

As an immediate consequence of Proposition \ref{prop:FVcont} we have the following:

\begin{corollary}\label{cor:cont}
For every $\contAFMC$-formula $\mu p. \varphi$, $\varphi$ is continuous in $p$.
\fznote{can we say also that $\nu p. \varphi$ is cocontinuous in $p$, on the base of the second part of Prop. 2.14(2)?}
\end{corollary}

\subsection{Bisimulation}
Bisimulation is a notion of behavioral equivalence between processes.
For the case of  transition systems, it is formally defined as follows.

\begin{definition}
Let $\model = \tup{T, R, \tscolors, s_I}$ and
$\model' = \tup{T', R', \tscolors', s'_I}$ be transition systems.
A \emph{bisimulation} is a relation $Z \subseteq T \times T'$
such that for all $(t,t^{\prime}) \in Z$ the following holds:
\begin{description}
  \itemsep 0 pt
  \item[(atom)] $p \in \tscolors(t)$ iff $p \in \tscolors^{\prime}(t')$ for all $p\in\prop$;
  \item[(forth)] for all $s \in R[t]$ there is $s^{\prime} \in R^{\prime}[t^{\prime}]$ such that $(s,s^{\prime}) \in Z$;
  \item[(back)] for all $s^{\prime} \in R^{\prime}[t^{\prime}]$ there is $s \in R[t]$ such that $(s,s^{\prime}) \in Z$.
\end{description}
Two pointed transition systems $\model$ and $\model^{\prime}$ are
\emph{bisimilar} (denoted $\model \bis \model^{\prime}$) if there is a
bisimulation $Z \subseteq T \times T^{\prime}$ containing $(s_I,s'_I)$.
\end{definition}

The following fact about tree unravellings will allow us to provide a proof of
Theorem~\ref{t:m1} by just proving it for tree languages.

\begin{fact}\label{prop:tree_unrav}
$\model$ and $\model^e$ are bisimilar, for every transition system $\model$.
\end{fact}

A class of transition systems $\mclass$ is \emph{bisimulation closed} if $\model
\bis \model^{\prime}$ implies that $\model \in \mclass$ iff $\model^{\prime}
\in \mclass$, for all $\model$ and $\model^{\prime}$.
A formula $\varphi \in \llang$ is \emph{bisimulation-invariant} if $\model \bis
\model^{\prime}$ implies that $\model \mmodels \varphi$ iff $\model^{\prime}
\mmodels \varphi$, for all $\model$ and $\model^{\prime}$.

\begin{fact}
Each $\MC$-definable class of transition systems is bisimulation closed.
\end{fact}

%% file: paraut.tex

\subsection{Parity automata}

We recall the definition of a parity automaton, adapted to our setting.
Since we will be comparing parity automata defined in terms of various
one-step languages, it makes sense to make the following abstraction.

\begin{definition}
Given a set $A$, we define an \emph{$A$-structure} to be a pair $\struc{D,\val}$
consisting of a domain $D$ and a valuation $\val: A \to \wp D$.
Depending on context, elements of $A$ will be called \emph{monadic predicates}
or \emph{propositional variables}.

A \emph{one-step language} is a map $\yvLo$ assigning to each set $A$ a set
$\yvLo(A)$ of objects called \emph{one-step formulas} over $A$.
We require that $\yvLo(\bigcap_{i} A_{i}) = \bigcap_{i} \yvLo(A_{i})$,
so that for each $\phi \in \yvLo(A)$ there is a smallest $A_{\phi} \sse A$ such
that $\phi \in \yvLo(A_{\phi})$; this $A_{\phi}$ is the set of propositional
variables that \emph{occur} in $\phi$.

We assume that such a one-step language $\yvLo$ comes with a \emph{truth}
relation: given an $A$-structure $\struc{D,\val}$, a formula $\phi \in \yvLo$
is either \emph{true} or \emph{false} in $\struc{D,\val}$, denoted by,
respectively, $\struc{D,\val} \models \phi$ and $\struc{D,\val} \not\models \phi$.
We also assume that $\yvLo$ has a \emph{positive fragment} $\yvmLo$
characterizing monotonicity in the sense that a formula $\phi \in \yvLo(A)$ is
(semantically) monotone iff it is equivalent to a formula $\phi' \in
\yvmLo(A)$.
\end{definition}

The one-step languages $\yvLo$ featuring in this paper all are induced by
well-known logics.
Examples include first-order logic (with and without equality), first-order
logic extended with the infinity quantifier, and fragments of these
languages.

\begin{definition}
\label{def:parity_aut}
Let $\yvLo$ be some one-step language.
A \emph{parity automaton} based on $\yvLo$ and 
alphabet $C$ is a tuple $\bbA = \tup{A,\tmap,\pmap,a_I}$ such that $A$ is a
finite set of states of the automaton, $a_I \in A$ is the initial state,
$\tmap: A\times C \to \yvmLo(A)$
is the transition map, and $\pmap: A \to \nat$ is the parity map.
The collection of such automata will be denoted by $\yvAut(\yvLo)$.
\end{definition}

Acceptance and rejection of a transition system by an automaton is defined
in terms of the following parity game.

\begin{definition}
Given $\bbA = \tup{A,\tmap,\pmap,a_I}$ in $\yvAut(\yvLo)$ and a transition
system $\bbT = \tup{T,R,\tscolors,s_I}$, the \emph{acceptance game}
$\mathcal{A}(\bbA,\bbT)$ of $\bbA$ on $\model$ is the parity game defined
according to the rules of the following table.
%
\begin{center}
\begin{tabular}{|l|c|l|c|} \hline
Position & Player & Admissible moves & Parity
\\\hline
    $(a,s) \in A \times T$
  & $\exists$
  & $\{\val : A \to \wp(R[s]) \mid (R[s],\val) \models \tmap (a, \tscolors(s)) \}$
  & $\pmap(a)$
\\
    $\val : A \rightarrow \wp(T)$
  & $\forall$
  & $\{(b,t) \mid t \in \val(b)\}$
  & $\max(\pmap[A])$
\\ \hline
 \end{tabular}
\end{center}

A transition system $\bbT$ is \emph{accepted} by $\bbA$ if $\exists$ has
a winning strategy in $\mathcal{A}(\bbA,\model)@(a_I,s_I)$, and \emph{rejected}
if $(a_I,s_I)$ is a winning position for $\abel$.
\end{definition}


Many properties of parity automata are determined at the one-step level.
An important example concerns the notion of complementation.

\begin{definition}
\label{d:bdual1}
Two one-step formulas $\phi$ and $\psi$ are each other's \emph{Boolean dual}
if for every structure $\struc{D,\val}$ we have
\[
\struc{D,\val} \models \phi \text{ iff } \struc{D,\val^{c}} \not\models \psi,
\]
where $\val^{c}$ is the valuation given by $\val^{c}(a) \mathrel{:=} D
\setminus \val(a)$, for all $a$.
A one-step language $\yvLo$ is \emph{closed under Boolean duals} if for every
set $A$, each formula $\phi \in \yvLo(A)$ has a Boolean dual $\yvdual{\phi}
\in \yvLo(A)$.
\end{definition}

Following ideas by~\cite{DBLP:conf/calco/KissigV09}\fzwarning{Yde, is this the right reference?}, we can use Boolean duals, together with a
\emph{role switch} between $\abel$ and $\eloi$, in order to define a
negation or complementation operation on automata.

\begin{definition}
\label{d:caut}
Assume that, for some one-step language $\yvLo$, the map $\yvdual{(-)}$
provides, for each set $A$, a Boolean dual $\yvdual{\phi} \in \yvLo(A)$ for each
$\phi \in \yvLo(A)$.
Given $\bbA = \tup{A,\tmap,\pmap,a_I}$ in $\yvAut(\yvLo)$ we define its
\emph{complement} $\overline{\bbA}$ as the automaton
$\tup{A,\yvdual{\De},\yvdual{\Om},\ai}$,
where $\yvdual{\De}(a,c) \isdef \yvdual{(\De(a,c))}$, and $\yvdual{\Om}(a)
\isdef 1 + \Om(a)$, for all $a \in A$ and $c \in C$.
\end{definition}

\begin{proposition}
\label{PROP_complementation}
Let $\yvLo$ and $\yvdual{(-)}$ be as in the previous definition.
For each automaton $\bbA \in \yvAut(\yvLo)$ and each transition structure
$\bbT$ we have that
\begin{equation*}
\overline{\bbA} \text{ accepts } \bbT \text{ iff }
\bbA \text{ rejects } \bbT.
\end{equation*}
\end{proposition}

The proof of Proposition~\ref{PROP_complementation} is based on the fact
that the power\fcwarning{`power' undefined.} of $\eloi$ in $\mathcal{A}(\overline{\bbA},\bbT)$ is the same
as that of $\abel$ in $\mathcal{A}(\bbA,\bbT)$.

As an immediate consequence of this proposition, one may show that if the
one-step language $\yvLo$ is closed under Boolean duals, then the class
$\yvAut(\yvLo)$ is closed under taking complementation.
Further on we will use Proposition~\ref{PROP_complementation} to show that
the same may apply to some subsets of $\yvAut(\yvLo)$.

%% file: one-step.tex
  Given a set of names $A$ we use $\llang_1(A)$ to denote an arbitrary one-step language based on $A$.
  One-step formulas are interpreted over unary (or one-step) models $\osmodel = (D,\val:A\to\wp D)$.
  We use $\umods(A)$ to denote the class of all one-step models based on the set of names $A$.

In this section we define the one-steps logics that we use in the rest of the article, namely: one-step weak monadic second-order logic ($\owmso$), one-step first-order logic with and without equality ($\ofoe$, $\ofo$) and one-step first-order logic with the $\qu$ quantifier ($\olque$). The main theorems of this section prove normal forms for these logics and give syntactical characterizations of the monotonic and continuous fragments that we use in later sections.

\begin{definition}
Let $\fovar$ be a set of (individual) variables. 
Given a set of names $A$ we define the set $\owmso(A)$ of one-step weak monadic second-order formulas as the \emph{sentences} given by the following grammar
%
\begin{align*}
\varphi \ ::=  & \ a(x)
\mid x \foeq y
\mid \neg \varphi
 \mid \varphi \lor \varphi
\mid \exists x.\varphi
\mid \exists a.\varphi 
\end{align*}
where $x,y \in \fovar$, 
$a \in A$.
\end{definition}

\begin{definition}
The set $\ofoe(A)$ of one-step first-order sentences (with equality) is given by the sentences formed by
$$
\varphi ::= a(x)
\mid x \foeq y
\mid \neg \varphi
\mid \varphi \lor \varphi
\mid \exists x.\varphi
$$
where $x,y\in \fovar$, $a \in A$. The one-step logic $\ofo(A)$ is as $\ofoe(A)$ but without equality.
The set $\olque(A)$ of one-step first-order sentences with generalized quantifier $\qu$ (with equality)
is defined analogously by just adding the clauses $\qu x. \varphi$ and $\dqu x. \varphi$.
\end{definition}

Without loss of generality, from now on we always assume that every bounded variable occurring in a sentence is bounded by an unique quantifier (generalized or not).
%
%
Recall that given a one-step logic $\llang_1$ we write $\llang_1^+(A)$ to denote the fragment
where every predicate $a\in A$ occurs only positively. The following observation will allow us to
work with the (one-sorted) language $\ofoe$ instead of the (two-sorted) language $\owmso$.

\begin{fact}[\cite{vaananen77}]
$\owmso(A) \equiv \olque(A)$.
\end{fact}

In the following subsections we provide a detailed model theoretic analysis of the one-step logics that we use in this article, specifically, we give
\begin{itemize}
	\itemsep 0 pt
	\item Normal forms for arbitrary formulas of $\ofo$, $\ofoe$ and $\olque$.
	\item Strong forms of syntactic characterizations for the monotone and continuous fragments of several of the mentioned logics. Namely, for $\llang_1 \in \{\ofo,\ofoe,\olque\}$ we provide
		\begin{enumerate}[(a)]
			\item A fragment $\monot{\llang_1}{a}$ and a translation $(-)^\tmono:\llang_1(A)\to\monot{\llang_1}{a}(A)$ such that for every $\varphi \in\llang_1$ we have $\varphi\equiv\varphi^\tmono$ iff $\varphi$ is monotone in $a \in A$,
		\end{enumerate}
		for $\llang_1 \in \{\ofo,\olque\}$ we provide
		\begin{enumerate}[(a)]
			\item[(b)] A fragment $\cont{\llang_1}{a}$ and a translation $(-)^\tcont:\llang_1(A)\to\cont{\llang_1}{a}(A)$ such that for every $\varphi \in\llang_1$ we have $\varphi\equiv\varphi^\tcont$ iff $\varphi$ is continuous in $a \in A$.\fcwarning{The real cont-translation is from $\monot{\llang_1}{a}(A)$}
		\end{enumerate}
		Moreover, we show that the latter translation also restricts to the fragment $\llang^+_1$, i.e.,
		\begin{enumerate}[(a)]
			\item[(c)] The restriction $(-)^\tcont_+:\llang^+_1(A)\to\cont{\llang^+_1}{a}(A)$ of $(-)^\tcont$ is such that for every $\varphi \in\llang^+_1$ we have $\varphi\equiv\varphi^\tcont_+$ iff $\varphi$ is continuous in $a \in A$.
		\end{enumerate}
	\item Syntactic characterizations of the co-continuous fragments of $\ofo$ and $\olque$.
	\item Normal forms for the monotone and continuous fragments.
\end{itemize}

\subsection{Normal forms}\label{subsec:normalforms}
\input{normalforms.tex}

\subsection{One-step monotonicity}
\input{monotonicity.tex}

\subsection{One-step continuity}\label{subsec:one-stepcont}
\input{continuity.tex}

\subsection{One-step co-continuity and Boolean duals}\label{subsec:one-stepcocont}
\input{cocontinuity.tex}

%% file: normalforms.tex

Given a set of names $A$ and $S \subseteq A$, we introduce the notation
$$
  \tau_{S}(x) := \bigwedge_{a\in S}a(x) \land \bigwedge_{a\in A\setminus S}\lnot a(x).
$$
The formula $\tau_{S}(x)$ is called an \emph{$A$-type}, we usually blur the distinction between $\tau_{S}(x)$ and $S$ and call $S$ an $A$-type as well.
A \emph{positive} $A$-type is defined as $\tau_{S}^+(x) := \bigwedge_{a\in S}a(x)$.
We use the convention that, if $S$ is the empty set, then $\tau_S^+(x)$ is $\top$ and we call it an \emph{empty} positive $A$-type.
Given a one-step model $\osmodel$ we use $|S|_\osmodel$ to denote the number of elements that realize the $A$-type $\tau_S$ in $\osmodel$. Formally, it is defined as $|S|_\osmodel := |\{d\in |\osmodel| : \osmodel \models \tau_S(d) \}|$.

A \emph{partial isomorphism} between two one-step models $\osmodel = (D,\val)$ and $\osmodel' = (D',\val')$ is a \emph{partial} function $f: D \to D'$ which is injective and satisfies $d \in \val(a) \Leftrightarrow f(d) \in \val'(a)$ for all $a\in A$ and $ d\in \Dom(f)$. Given two sequences $[d_1,\dots,d_k] \in D^k$ and $[d'_1,\dots,d'_k] \in {D'}^k$ 
we use
$f: [d_1,\dots,d_k] \mapsto [d'_1,\dots,d'_k]$ to denote the partial function $f:D\to D'$ defined as $f(d_i) := d'_i$. If there exist $d_i,d_j$ such that $d_i = d_j$ but $d'_i \neq d'_j$ then the result is undefined.


\begin{definition}
The quantifier rank $qr(\varphi)$ of $\varphi \in \olque$ (hence also for $\ofo$ and $\ofoe$) is defined as follows
\begin{itemize}
	\itemsep 0 pt
	\item If $\varphi$ is atomic $qr(\varphi) = 0$,
	\item If $\varphi = \lnot\psi$ then $qr(\varphi) = qr(\psi)$,
	\item If $\varphi = \psi_1 \land \psi_2$ or $\varphi = \psi_1 \lor \psi_2$ then $qr(\varphi) = \max\{qr(\psi_1),qr(\psi_2)\}$,
	\item If $\varphi = Qx.\psi$ for $Q \in \{\exists,\forall,\qu,\dqu\}$ then $qr(\varphi) = 1+qr(\psi)$.
\end{itemize}
Given a one-step logic $\llang_1$ we write $\osmodel \equiv_k^{\llang} \osmodel'$ to indicate that the one-step models $\osmodel$ and $\osmodel'$ satisfy exactly the same formulas $\varphi \in \llang_1$ with $qr(\varphi) \leq k$. The logic $\llang$ will be omitted when it is clear from context.
\end{definition}

\subsubsection{Normal form for $\ofo$}

We start by stating a normal form for one-step first-order logic without equality. A formula in \emph{basic form} gives a complete description of the types that are satisfied in a one-step model.

\begin{definition}
A formula $\varphi \in \ofo(A)$ is in \emph{basic form} if $\varphi = \bigvee \dbnfofo{\Sigma}$
where each disjunct is of the form
$$
\dbnfofo{\Sigma} = \bigwedge_{S\in\Sigma} \exists x. \tau_S(x) \land \forall x. \bigvee_{S\in\Sigma} \tau_S(x)
$$
for some set of types $\Sigma\subseteq \wp A$.
\end{definition}

It is easy to prove, using Ehrenfeucht-Fra\"iss\'e games, that every formula of first-order logic without equality over a unary signature (i.e., $\ofo$) is equivalent to a formula in basic form. Proof sketches can be found in~\cite[Lemma 16.23]{ALG02} and~\cite[Proposition 4.14]{Venxx}. We omit a full proof because it is very similar to the following cases.

\begin{fact}
Every formula of $\ofo(A)$ is equivalent to a formula in basic form.
\end{fact}

\subsubsection{Normal form for $\ofoe$}

When considering a normal form for $\ofoe$, the fact that we can `count types' using equality yields a more involved basic form.

\begin{definition}
A formula $\varphi \in \ofoe(A)$ is in \emph{basic form} if $\varphi = \bigvee \dbnfofoe{\vlist{T}}{\Pi}$ where each disjunct is of the form
$$
\dbnfofoe{\vlist{T}}{\Pi} = \exists \vlist{x}.\big(\arediff{\vlist{x}} \land \bigwedge_i \tau_{T_i}(x_i) \land \forall z.(\arediff{\vlist{x},z} \lthen \bigvee_{S\in \Pi} \tau_S(z))\big)
$$
for some set of types $\Pi \subseteq \wp A$ and where each $T_i \subseteq A$. The predicate $\arediff{\vlist{y}}$, which states that the given elements are different, is defined as $\arediff{y_1,\dots,y_n} := \bigwedge_{1\leq m < m^{\prime} <n} (y_m \not\approx y_{m^{\prime}})$.
\end{definition}

We prove that every formula of monadic first-order logic with equality (i.e., $\ofoe$) is equivalent to a formula in basic form. This result seems to be folklore, however, we give a detailed proof because some of its ingredients are used in the following section, when giving a normal form for $\olque$. We start by defining the following relation between one-step models.

\begin{definition}
	Let $\osmodel$ and $\osmodel'$ be one-step models. For every $k \in \nat$ we define 
\begin{eqnarray*}
	\osmodel \sim^=_k \osmodel' & \Longleftrightarrow & \forall S\subseteq A \ \big(
	   |S|_\osmodel = |S|_{\osmodel'} < k \\
	&& \qquad\qquad \text{or } |S|_\osmodel,|S|_{\osmodel'} \geq k \big)
\end{eqnarray*}
\end{definition}

Intuitively, two models are related by $\sim^=_k$ when their type information coincides `modulo~$k$'. Later we will prove that this is the same as saying that they cannot be distinguished by a formula of $\ofoe$ with quantifier rank lower or equal to $k$. For the moment, we prove the following properties of $\sim^=_k$.

\begin{proposition}\label{prop:eqrelofoe} The following hold
	\begin{enumerate}[(i)]
		\itemsep 0 pt
		\item $\sim^=_k$ is an equivalence relation,
		\item $\sim^=_k$ has finite index,
		\item Every $E \in \umods/{\sim^=_k}$ is characterized by a formula $\varphi^=_E \in \ofoe(A)$ with $qr(\varphi^=_E) = k$.
	\end{enumerate}
\end{proposition}
\begin{proof}
	We only prove the last point. Let $E \in \umods/{\sim^=_k}$ and let $\osmodel \in E$ be a representative. Call $S_1,\dots,S_n \subseteq A$ to the types such that $|S_i|_\osmodel = n_i < k$ and $S'_1,\dots,S'_m \subseteq A$ to those satisfying $|S'_i|_\osmodel \geq k$. Now define
	\begin{align*}
	\varphi^=_E :=& \bigwedge_{i\leq n} \big(\exists x_1,\dots,x_{n_i}.\arediff{x_1,\dots,x_{n_i}} \land \bigwedge_{j\leq n_i} \tau_{S_i}(x_j) \land \forall z. \arediff{x_1,\dots,x_{n_i},z} \lthen \lnot\tau_{S_i}(z)\big)\ \land \\
	            & \bigwedge_{i\leq m} \big(\exists x_1,\dots,x_k.\arediff{x_1,\dots,x_k} \land \bigwedge_{j\leq k} \tau_{S'_i}(x_j) \big)
	\end{align*}
	It is easy to see that $qr(\varphi^=_E) = k$ and that $\osmodel' \models \varphi^=_E$ iff $\osmodel' \in E$. Observe that $\varphi^=_E$ gives a specification of $E$ ``type by type''.
\end{proof}

In the following definition we recall a (standard) notion of Ehrenfeucht-Fra\"iss\'e game for $\ofoe$ which will be used to establish the connection between ${\sim^=_k}$ and $\equiv_k^\foe$.

\begin{definition}
	Let $\osmodel_0 = (D_0,\val_0)$ and $\osmodel_1 = (D_1,\val_1)$ be one-step models. We define the game $\efgame^=_k(\osmodel_0,\osmodel_1)$ between \abelard and \eloise. If $\osmodel_i$ is one of the models we use $\osmodel_{-i}$ to denote the other model. A position in this game is a pair of sequences $\vlist{s_0} \in D_0^n$ and $\vlist{s_1} \in D_1^n$ with $n \leq k$. The game consists of $k$ rounds where in round $n+1$ the following steps are made
	\begin{enumerate}[1.]
		\itemsep 0 pt
		\parsep 0 pt
		\item \abelard chooses an element $d_i$ in one of the $\osmodel_i$,
		\item \eloise responds with an element $d_{-i}$ in the model $\osmodel_{-i}$.
		\item Let $\vlist{s_i} \in D_i^n$ be the sequences of elements chosen up to round $n$, they are extended to ${\vlist{s_i}' := \vlist{s_i}\cdot d_i}$. Player \eloise survives the round iff she does not get stuck and the function $f_{n+1}: \vlist{s_0}' \mapsto \vlist{s_1}'$ is a partial isomorphism of one-step models.
	\end{enumerate}
	Player \eloise wins iff she can survive all $k$ rounds.
	%
	%
	Given $n\leq k$ and $\vlist{s_i} \in D_i^n$ such that $f_n:\vlist{s_0}\mapsto\vlist{s_1}$ is a partial isomorphism, we use $\efgame_{k}^=(\osmodel_0,\osmodel_1)@(\vlist{s_0},\vlist{s_1})$ to denote the (initialized) game where $n$ moves have been played and $k-n$ moves are left to be played.
\end{definition}

\begin{lemma}\label{lem:connofoe}
	The following are equivalent
	\begin{enumerate}
		\itemsep 0 pt
		\item\label{lem:connofoe:i} $\osmodel_0 \equiv_k^\foe \osmodel_1$,
		\item\label{lem:connofoe:ii} $\osmodel_0 \sim_k^= \osmodel_1$,
		\item\label{lem:connofoe:iii} \eloise has a winning strategy in $\efgame_k^=(\osmodel_0,\osmodel_1)$.
	\end{enumerate}
\end{lemma}
\begin{proof}
	Step~(\ref{lem:connofoe:i}) to~(\ref{lem:connofoe:ii}) is direct by Proposition~\ref{prop:eqrelofoe}. For~(\ref{lem:connofoe:ii}) to~(\ref{lem:connofoe:iii}) we give a winning strategy for \eloise in $\efgame_k^=(\osmodel_0,\osmodel_1)$. We do it by showing the following claim
	\begin{claimfirst}
	Let $\osmodel_0 \sim_k^= \osmodel_1$ and $\vlist{s_i} \in D_i^n$ be such that $n<k$ and $f_n:\vlist{s_0}\mapsto\vlist{s_1}$ is a partial isomorphism; then \eloise can survive one more round in $\efgame_{k}^=(\osmodel_0,\osmodel_1)@(\vlist{s_0},\vlist{s_1})$.
	\end{claimfirst}
	\begin{pfclaim}
		Let \abelard pick $d_i\in D_i$ such that the type of $d_i$ is $T \subseteq A$. If $d_i$ had already been played then \eloise picks the same element as before and $f_{n+1} = f_n$. If $d_i$ is new and $|T|_{\osmodel_i} \geq k$ then, as at most $n<k$ elements have been played, there is always some new $d_{-i} \in D_{-i}$ that \eloise can choose that matches $d_i$. If $|T|_{\osmodel_i} = m < k$ then we know that $|T|_{\osmodel_{-i}} = m$. Therefore, as $d_i$ is new and $f_n$ is injective, there must be a $d_{-i} \in D_{-i}$ that \eloise can choose. 
	\end{pfclaim}
	
	Step~(\ref{lem:connofoe:iii}) to~(\ref{lem:connofoe:i}) is a standard result~\cite[Corollary 2.2.9]{fmt} which we prove anyway because we will need to extend it later. We prove the following loaded statement
	\begin{claim}
		Let $\vlist{s_i} \in D_i^n$ and $\varphi(z_1,\dots,z_n) \in \ofoe(A)$ be such that $qr(\varphi) \leq k-n$. If \eloise has a winning strategy in $\efgame_k^=(\osmodel_0,\osmodel_1)@(\vlist{s_0},\vlist{s_1})$ then $\osmodel_0 \models \varphi(\vlist{s_0})$ iff $\osmodel_1 \models \varphi(\vlist{s_1})$.
	\end{claim}
	\begin{pfclaim}
		If $\varphi$ is atomic the claim holds because of $f_n:\vlist{s_0}\mapsto \vlist{s_1}$ being a partial isomorphism. Boolean cases are straightforward.
		Let $\varphi(z_1,\dots,z_n) = \exists x. \psi(z_1,\dots,z_n,x)$ and suppose $\osmodel_0 \models \varphi(\vlist{s_0})$. Hence, there exists $d_0 \in D_0$ such that $\osmodel_0 \models \psi(\vlist{s_0},d_0)$.
		By hypothesis we know that \eloise has a winning strategy for $\efgame_k^=(\osmodel_0,\osmodel_1)@(\vlist{s_0},\vlist{s_1})$. Therefore, if \abelard picks $d_0\in D_0$ she can respond with some $d_1\in D_1$ and has a winning strategy for $\efgame_{k}^=(\osmodel_0,\osmodel_1)@(\vlist{s_0}{\cdot}d_0,\vlist{s_1}{\cdot}d_1)$.
		By induction hypothesis, because $qr(\psi) \leq k- (n+1)$, we have that $\osmodel_0 \models \psi(\vlist{s_0},d_0)$ iff $\osmodel_1 \models \psi(\vlist{s_1},d_1)$ and hence $\osmodel_1 \models \exists x.\psi(\vlist{s_1},x)$. The other direction is symmetric. 
		\end{pfclaim}
		Combining these claims finishes the proof of the lemma.
\end{proof}

\begin{theorem}\label{thm:bnfofoe}
Every $\psi \in \ofoe(A)$ is equivalent to a formula in basic form.
\end{theorem}
\begin{proof}
	Let $qr(\psi) = k$ and let $\ext{\psi}$ be the models satisfying $\psi$. As $\umods/{\equiv_k^\foe}$ is the same as $\umods/{\sim_k^=}$ by Lemma~\ref{lem:connofoe}, it is easy to see that $\psi \equiv \bigvee \{ \varphi^=_E \mid E \in \ext{\psi}/{\sim_k^=} \}$. Now it only remains to see that each $\varphi^=_E$ is equivalent to $\dbnfofoe{\vlist{T}}{\Pi}$ for some $\Pi \subseteq \wp A$ and $T_i \subseteq A$.

	The crucial observation is that we will use $\vlist{T}$ and $\Pi$ to give a specification of the types ``element by element''. Let $\osmodel \in E$ be a representative. Call $S_1,\dots,S_n \subseteq A$ to the types such that $|S_i|_\osmodel = n_i < k$ and $S'_1,\dots,S'_m \subseteq A$ to those satisfying $|S'_i|_\osmodel \geq k$. The size of the sequence $\vlist{T}$ is defined to be $(\sum_{i=1}^n n_i) + k\times m$ where $\vlist{T}$ is contains exactly $n_i$ occurrences of type $T_i$ and $k$ occurrences of each $S'_j$. On the other hand $\Pi = \{S'_1,\dots,S'_m\}$. It is straightforward to check that $\varphi^=_E$ is equivalent to $\dbnfofoe{\vlist{T}}{\Pi}$, however, the quantifier rank of the latter is only bounded by $k\times 2^{|A|} + 1$.
\end{proof}

\subsubsection{Normal form for $\olque$}

The logic $\olque$ basically extends $\ofoe$ with the capacity to tear apart finite and infinite sets of elements. This is reflected in the normal form for $\olque$ by adding extra constraints to the normal form of $\ofoe$.

\begin{definition}\label{def:basicform_folque}
A formula $\varphi \in \olque(A)$ is in \emph{basic form} if $\varphi = \bigvee \dbnfolque{\vlist{T}}{\Pi}{\Sigma}$ where each disjunct is of the form
$$
\dbnfolque{\vlist{T}}{\Pi}{\Sigma} = \dbnfofoe{\vlist{T}}{\Pi \cup \Sigma} \land \dbnfinf{\Sigma}
$$
where
$$
\dbnfinf{\Sigma} := \bigwedge_{S\in\Sigma} \qu y.\tau_S(y) \land \dqu y.\bigvee_{S\in\Sigma} \tau_S(y)
$$
for some set of types $\Pi,\Sigma \subseteq \wp A$ and each $T_i \subseteq A$.
\end{definition}

A small argument reveals that, intuitively, every disjunct expresses that each one-step model satisfying it admits a partition of its domain in three parts:
\begin{enumerate}[(i)]
\itemsep 0 pt
\item Distinct elements $t_1,\dots,t_n$ with type $T_1,\dots,T_n$,
\item Finitely many elements whose types belong to $\Pi$, and
\item For each $S\in \Sigma$, infinitely many elements with type $S$.
\end{enumerate}

In the same way as before, we define a relation $\sim^\infty_k$ which can be clearly seen to be a refinement of $\sim^=_{k}$ which adds information about the (in-)finiteness of the types.

\begin{definition}
	Let $\osmodel$ and $\osmodel'$ be one-step models. For every $k\in\nat$ we define 
\begin{eqnarray*}
	\osmodel \sim^\infty_0 \osmodel' & \Longleftrightarrow & \text{always}\\
	\osmodel \sim^\infty_{k+1} \osmodel' & \Longleftrightarrow & \osmodel \sim^=_{k+1} \osmodel' \text{ and }\\
	&& \forall S\subseteq A \ \big(
	|S|_\osmodel,|S|_{\osmodel'} < \omega \text{ or } |S|_\osmodel,|S|_{\osmodel'} \geq \omega \big)
\end{eqnarray*}
\end{definition}

\begin{proposition}\label{prop:eqrelolque} The following hold
	\begin{enumerate}[(i)]
		\itemsep 0 pt
		\item $\sim^\infty_k$ is an equivalence relation,
		\item $\sim^\infty_k$ has finite index,
		\item $\sim^\infty_k$ is a refinement of $\sim^=_k$,
		\item Every $E \in \umods/{\sim^\infty_k}$ is characterized by a formula $\varphi^\infty_E \in \olque(A)$ with $qr(\varphi) = k$.
	\end{enumerate}
\end{proposition}
\begin{proof}
	We only prove the last point, for $k>0$. Let $E \in \umods/{\sim^\infty_k}$ and let $\osmodel \in E$ be a representative of the class. Let $E' \in \umods/{\sim^=_k}$ be the eqivalence class of $\osmodel$ with respect to $\sim^=_k$.
	Call $S_1,\dots,S_n \subseteq A$ to the types such that $|S_i|_\osmodel \geq \omega$ and $S'_1,\dots,S'_m \subseteq A$ to those satisfying $|S'_i|_\osmodel < \omega$. Now define
	$$
	\varphi^\infty_E := \varphi^=_{E'}\ \land \\
		\bigwedge_{i\leq n} \qu x.\tau_{S_i}(x) \land \bigwedge_{i\leq m} \dqu x.\lnot\tau_{S'_i}(x) .
	$$
	It is easy to see that $qr(\varphi^\infty_E) = k$ and that $\osmodel' \models \varphi^\infty_E$ iff $\osmodel' \in E$. 
\end{proof}

Now we give a notion of Ehrenfeucht-Fra\"ss\'e game for $\olque$. In this case the game extends $\efgame^=_k$ with a move for $\qu$.

\begin{definition}
	Let $\osmodel_0 = (D_0,\val_0)$ and $\osmodel_1 = (D_1,\val_1)$ be one-step models. We define the game $\efgame^\infty_k(\osmodel_0,\osmodel_1)$ between \abelard and \eloise. A position in this game is a pair of sequences $\vlist{s_0} \in D_0^n$ and $\vlist{s_1} \in D_1^n$ with $n \leq k$. The game consists of $k$ rounds where in round $n+1$ the following steps are made. First \abelard chooses to perform one of the following types of moves:
	\begin{enumerate}[(a)]
		\itemsep 0 pt
		\parsep 0 pt
		\item Second-order move
		\begin{enumerate}[1.]
			\itemsep 0 pt
			\parsep 0 pt
			\item \abelard chooses an infinite set $X_i \subseteq D_i$,
			\item \eloise responds with an infinite set $X_{-i} \subseteq D_{-i}$,
			\item \abelard chooses an element $x_{-i} \in X_{-i}$,
			\item \eloise responds with an element $x_i \in X_i$.
		\end{enumerate}
		\item First-order move
		\begin{enumerate}[1.]
			\itemsep 0 pt
			\parsep 0 pt
			\item \abelard chooses an element $d_i \in D_i$,
			\item \eloise responds with an element $d_{-i} \in D_{-i}$.
		\end{enumerate}
	\end{enumerate}
	Let $\vlist{s_i} \in D_i^n$ be the sequences of elements chosen up to round $n$, they are extended to ${\vlist{s_i}' := \vlist{s_i}\cdot d_i}$. \eloise survives the round iff she does not get stuck and the function $f_{n+1}: \vlist{s_0}' \mapsto \vlist{s_1}'$ is a partial isomorphism of one-step models.
\end{definition}

This game can be seen as an adaptation of the Ehrenfeucht-Fra\"ss\'e game for monotone generalized quantifiers found in~\cite{Kolaitis199523} to the case of full monadic first-order logic. 

\begin{lemma}\label{lem:connolque}
	The following are equivalent
	\begin{enumerate}
		\itemsep 0 pt
		\item\label{lem:connolque:i} $\osmodel_0 \equiv_k^{\lque} \osmodel_1$,
		\item\label{lem:connolque:ii} $\osmodel_0 \sim_k^\infty \osmodel_1$,
		\item\label{lem:connolque:iii} \eloise has a winning strategy in $\efgame_k^\infty(\osmodel_0,\osmodel_1)$.
	\end{enumerate}
\end{lemma}


\begin{proof}
	Step~(\ref{lem:connolque:i}) to~(\ref{lem:connolque:ii}) is direct by Proposition~\ref{prop:eqrelolque}. For~(\ref{lem:connolque:ii}) to~(\ref{lem:connolque:iii}) we show
	\begin{claimfirst}
	Let $\osmodel_0 \sim_k^\infty \osmodel_1$ and $\vlist{s_i} \in D_i^n$ be such that $n<k$ and $f_n:\vlist{s_0}\mapsto\vlist{s_1}$ is a partial isomorphism; then \eloise can survive one more round in $\efgame_{k}^\infty(\osmodel_0,\osmodel_1)@(\vlist{s_0},\vlist{s_1})$.
	\end{claimfirst}
	\begin{pfclaim}
		We focus on the second-order moves because the first-order moves are the same as in the corresponding Claim of Lemma~\ref{lem:connofoe}. Let \abelard choose an infinite set $X_i \subseteq D_i$, we would like \eloise to choose a set $X_{-i} \subseteq D_{-i}$ such that the following conditions hold:
		\begin{enumerate}[(a)]
			\parskip 0pt
			\item\label{it:piso} Let $\vlist{r_j} := \vlist{s_j}\cap X_j$ be the restriction of $\vlist{s_j}$ to the elements of $X_j$. We want $f':\vlist{r_0}\mapsto\vlist{r_1}$ to be a well-defined partial isomorphism between the restricted models $\osmodel_0{\rest}X_0$ and $\osmodel_1{\rest}X_1$,
			%
			\item\label{it:equiv} $\osmodel_0{\rest}X_0 \sim_{m+1}^= \osmodel_1{\rest}X_1$ where $m = |r_i|$. That is, for all $S \subseteq A$,
			\begin{enumerate}[(i)]
				\itemsep 0pt
				\item if $|S|_{X_i} < m+1$ then $|S|_{X_{-i}} = |S|_{X_i}$,
				\item if $|S|_{X_i} \geq m+1$ then $|S|_{X_{-i}} \geq m+1$,
			\end{enumerate}
			\item\label{it:inf} $X_{-i}$ is infinite.
		\end{enumerate}
		First we prove that such a set exists and then we will use it to prove the claim.

		To satisfy item~\ref{it:piso} she just needs to add to $X_{-i}$ the elements connected to $X_i$ by $f_n$; this is not a problem. For item~\ref{it:equiv} the only hinder is that, while adding elements to $X_{-i}$, we cannot make use of those elements of $D_{-i}$ which are in $\vlist{s_{-i}}$ but not in $\vlist{r_{-i}}$, since that would break condition~\ref{it:piso}. It is not difficult to see that because $m+1 \leq k$ and $\osmodel_0 \sim_k^\infty \osmodel_1$ implies $\osmodel_0 \sim_k^= \osmodel_1$ then we always have enough elements to satisfy the required constraints. For the last item observe that as $X_{i}$ is infinite but there are only finitely many types, hence there must be some $S$ such that $|S|_{X_i} \geq \omega$. It is then safe to add infinitely many elements for $S$ in $X_{-i}$ while considering point~\ref{it:equiv}. Moreover, the existence of infinitely many elements satisfying $S$ in $D_{-i}$ are guaranteed by $\osmodel_0 \sim_k^\infty \osmodel_1$.

		Having shown that \eloise can choose a set $X_{-i}$ satisfying the above conditions, we continue the proof of the claim as follows: as $\osmodel_0{\rest}X_0 \sim_{m+1}^= \osmodel_1{\rest}X_1$ and $f':\vlist{r_0}\mapsto\vlist{r_1}$ is a partial isomorphism between them (with sequences of length $m$) then by the first claim of Lemma~\ref{lem:connofoe}, we know that \eloise can survive one round in $\efgame_{m+1}^=(\osmodel_0{\rest}X_0,\osmodel_1{\rest}X_1)@(\vlist{r_0},\vlist{r_1})$. In particular, she can survive the ``first-order part'' of the second-order move we were considering. This finishes the proof of the claim.
	\end{pfclaim}
	Going back to the proof of the lemma, for step~(\ref{lem:connolque:iii}) to~(\ref{lem:connolque:i}) we prove the following claim.
	\begin{claim}
		Let $\vlist{s_i} \in D_i^n$ and $\varphi(z_1,\dots,z_n) \in \olque(A)$ be such that $qr(\varphi) \leq k-n$. If \eloise has a winning strategy in $\efgame_k^\infty(\osmodel_0,\osmodel_1)@(\vlist{s_0},\vlist{s_1})$ then $\osmodel_0 \models \varphi(\vlist{s_0})$ iff $\osmodel_1 \models \varphi(\vlist{s_1})$.
	\end{claim}
	\begin{pfclaim}
		All the cases involving operators of \ofoe are the same as in Lemma~\ref{lem:connofoe}. We prove the inductive case for the generalized quantifier. Let $\varphi(z_1,\dots,z_n)$ be of the form $\qu x.\psi(z_1,\dots,z_n,x)$ and let $\osmodel_0 \models \varphi(\vlist{s_0})$. Hence, there is an \emph{infinite} set $X_0 \subseteq D_0$ such that
		\begin{equation}\label{eq:ind0ok}
		\osmodel_0 \models \psi(\vlist{s_0},x_0) \text{ if and only if } x_0\in X_0 .
		\end{equation}
		By hypothesis we know that \eloise has a winning strategy for $\efgame_k^\infty(\osmodel_0,\osmodel_1)@(\vlist{s_0},\vlist{s_1})$. Therefore, if \abelard plays a second-order move by picking $X_0 \subseteq D_0$ she can respond with some infinite set $X_1 \subseteq D_1$. We claim that $\osmodel_1 \models \psi(\vlist{s_1},x_1)$ for every $x_1\in X_1$. First observe that if this holds then $X'_1 := \{ d_1 \in D_1 \mid \osmodel_1 \models \psi(\vlist{s_1},d_1)\}$ must be infinite, and hence $\osmodel_1 \models \qu x.\psi(\vlist{s_1},x)$.
		
		Assume, for a contradiction, that $\osmodel_1 \not\models \psi(\vlist{s_1},x'_1)$ for some $x'_1\in X_1$. Let \abelard play that $x'_1$ as the second part of his move. Then, as \eloise has a winning strategy, she will respond with some $x'_0 \in X_0$ such that she has a winning strategy for $\efgame_{k}^\infty(\osmodel_0,\osmodel_1)@(\vlist{s_0}{\cdot}x'_0,\vlist{s_1}{\cdot}x'_1)$. By induction hypothesis, as $qr(\psi) \leq k-(n+1)$, this means that $\osmodel_0 \models \psi(\vlist{s_0},x'_0)$ iff $\osmodel_1 \models \varphi(\vlist{s_1},x'_1)$ which contradicts~(\ref{eq:ind0ok}). The other direction is symmetric.
	\end{pfclaim}
	Combining the claims finishes the proof of the lemma.
\end{proof}

\begin{theorem}\label{thm:bfolque}
Every $\varphi \in \olque(A)$ is equivalent to a formula in basic form.
\end{theorem}
\begin{proof}
	This can be proved using the same technique as in Theorem~\ref{thm:bnfofoe}. Hence we only focus on showing that $\varphi_E^\infty \equiv \dbnfolque{\vlist{T}}{\Pi}{\Sigma}$ for some $\Pi,\Sigma \subseteq \wp A$ and $T_i \subseteq A$. Recall that
	$$
	\varphi^\infty_E = \varphi^=_{E'}\ \land \\
		\bigwedge_{i\leq n} \qu x.\tau_{S_i}(x) \land \bigwedge_{i\leq m} \dqu x.\lnot\tau_{S'_i}(x)
	$$
	and from Theorem~\ref{thm:bnfofoe} we know that this is equivalent to
	$$
	\varphi^\infty_E = \dbnfofoe{\vlist{T}'}{\Pi'} \land \\
		\bigwedge_{i\leq n} \qu x.\tau_{S_i}(x) \land \bigwedge_{i\leq m} \dqu x.\lnot\tau_{S'_i}(x)
	$$
	for some $\Pi', \subseteq \wp A$ and $T'_i \subseteq A$. Now separate $\Pi'$ as $\Pi' = \Pi \uplus \Sigma$ where $\Pi$ is composed of the finite types (i.e., the $S'_i$) and $\Sigma$ is composed of the infinite types (i.e., the $S_i$). We claim that $\varphi_E^\infty \equiv \dbnfolque{\vlist{T}'}{\Pi}{\Sigma}$. After a minor rewritting, we have to show that the following two formulas are equivalent:
	\begin{eqnarray}
	\label{eq:classolque}
	\dbnfofoe{\vlist{T}'}{\Pi'} &\land& \bigwedge_{S \in \Sigma} \qu x.\tau_{S}(x) \land \bigwedge_{S\in \Pi} \dqu x.\lnot\tau_{S}(x) \\
	\label{eq:nfolque}
	\dbnfofoe{\vlist{T}'}{\Pi\cup\Sigma} &\land&	\bigwedge_{S\in\Sigma} \qu y.\tau_S(y) \land \dqu y.\bigvee_{S\in\Sigma} \tau_S(y)
	\end{eqnarray}
	It is clear that $\dbnfofoe{\vlist{T}'}{\Pi'}$ is equivalent to $\dbnfofoe{\vlist{T}'}{\Pi\cup\Sigma}$ and that the infinite-existential parts are also equivalent. Therefore we focus on the last part of both formulas.
	
	\medskip
	\noindent \fbox{$\Rightarrow$} 
	We show that~(\ref{eq:classolque}) implies $\dqu y.\bigvee_{S\in\Sigma} \tau_S(y)$. By contrapositive assume that there are infinitely many elements satisfying $\bigwedge_{S\in\Sigma} \lnot\tau_S(y)$, that is, infinitely many elements with their type outside $\Sigma$. Because of the \foe part, these elements should have their type in $\Pi' = \Pi\cup\Sigma$. Therefore, as the number of types is finite, there must be some type $P \in \Pi$ such that $\qu x.\tau_P(x)$ which contradicts the last part of~(\ref{eq:classolque}).

	\medskip
	\noindent \fbox{$\Leftarrow$}
	We show that~(\ref{eq:nfolque}) implies $\bigwedge_{S\in \Pi} \dqu x.\lnot\tau_{S}(x)$. By contrapositive assume there is some type $P\in \Pi$ such that $\qu x.\tau_P(x)$. Then, in particular, these elements do not have their type in $\Sigma$ and hence there are infinitely many elements satisfying $\bigwedge_{S\in\Sigma} \lnot\tau_S(y)$ which contradicts the last part of~(\ref{eq:nfolque}).
\end{proof}

%% file: monotonicity.tex

Given a one-step logic $\llang_1(A)$ and formula $\varphi \in \llang_1(A)$.
We say that $\varphi$ is \emph{monotone in $a\in A$} if for every one step model $(D,\val:A\to\wp D)$ and assignment $\ass:\fovar\to D$,
$$\text{If } (D,\val),\ass \models \varphi \text{ and } \val(a) \subseteq E \text{ then } (D,\val[a\mapsto E]),\ass \models \varphi.$$
We use $\llang_1^+(A)$ to denote the fragment of $\llang_1(A)$ composed of formulas monotone in all $a\in A$.

Monotonicity is usually tightly related to positivity. If the quantifiers are well-behaved (i.e., monotone) then a formula $\varphi$ will usually be monotone in $a \in A$ iff $a$ has positive polarity in $\varphi$, that is, if it occurs under an even number of negations. This is the case for all one-step logics considered in this article. In this section we give a syntactic characterization of monotonicity for several one-step logics.

\medskip\noindent
\textit{Convention}.
	Given $S\cup\{a\} \subseteq A$ we use $\tau^a_S$ to denote the \emph{$a$-positive} $A$-type defined as
	$$
	\tau^a_S(x) := \bigwedge_{b\in S} b(x) \land \bigwedge_{b\in A\setminus (S\cup\{a\})}\lnot b(x).
	$$
	Intuitively, $\tau^a_S$ is an $A$-type where $a$ can only occur positively and all other predicates occur either positively or negatively. 


\subsubsection{Monotone fragment of $\ofo$}

\begin{theorem}\label{thm:ofomonot}
\fcwarning{Cite van Benthem? check}A formula of $\ofo(A)$ is monotone in ${a\in A}$ iff it is equivalent to a sentence given by
$$
\varphi ::= \psi \mid a(x) \mid \exists x.\varphi(x) \mid \forall x.\varphi(x) \mid \varphi \land \varphi \mid \varphi \lor \varphi
$$
where 
$\psi \in \ofo(A\setminus \{a\})$. We denote this fragment as $\monot{\ofo}{a}(A)$.
\end{theorem}
The result will follow from the following two lemmas. 

\begin{lemma}\label{lem:monofoismonot}
If $\varphi \in \monot{\ofo}{a}(A)$ then $\varphi$ is monotone in $a$.
\end{lemma}
\begin{proof}
	The proof is a routine argument by induction on the complexity of $\varphi$.
\end{proof}

To prove that the fragment is complete we need to show that every formula which is monotone in $a$ is equivalent to some formula in $\monot{\ofo}{a}$. We prove a stronger result: we give a translation that constructively maps arbitrary formulas into $\monot{\ofo}{a}$. The interesting observation is that the translation will preserve truth iff the given formula is monotone in $a$.

\begin{lemma}
There exists a translation $(-)^\tmono:\ofo(A) \to \monot{\ofo}{a}(A)$ such that
a formula ${\varphi \in \ofo(A)}$ is monotone in $a$ if and only if $\varphi\equiv \varphi^\tmono$.
\end{lemma}
\begin{proof}
To define the translation we assume that $\varphi$ is in a normal form $\bigvee \dbnfofo{\Sigma}$ where 
$$
\dbnfofo{\Sigma} := \bigwedge_{S\in\Sigma} \exists x. \tau_S(x) \land \forall x. \bigvee_{S\in\Sigma} \tau_S(x)
$$
for some types $\Sigma \subseteq \wp A$.
For the translation, let
$(\bigvee \dbnfofo{\Sigma})^\tmono:= \bigvee \mondbnfofo{\Sigma}{a}$ and
$$
\mondbnfofo{\Sigma}{a} := \bigwedge_{S\in\Sigma} \exists x. \tau^a_S(x) \land \forall x. \bigvee_{S\in\Sigma} \tau^a_S(x)
$$

From the construction it is clear that $\varphi^\tmono \in \monot{\ofo}{a}(A)$ and therefore the right-to-left direction of the lemma is immediate by Lemma~\ref{lem:monofoismonot}. For the left-to-right direction assume that $\varphi$ is monotone in $a$, we have to prove that $(D,\val) \models \varphi$ iff $(D,\val) \models \varphi^\tmono$.

\bigskip
\noindent \fbox{$\Rightarrow$} This direction is trivial.

\bigskip
\noindent \fbox{$\Leftarrow$} Assume $(D,\val) \models \varphi^\tmono$ and let $\mondbnfofo{\Sigma}{a}$ be such that $(D,\val) \models \mondbnfofo{\Sigma}{a}$. It is safe to assume that the \emph{only} ($a$-positive) types realized in $(D,\val)$ are exactly those in $\Sigma$ and that every type has a (single) distinct witness.
%
%
For every $S \in \Sigma$, let $d_S$ be the witness of the $a$-positive type $\tau^a_S(x)$ in $(D,\val)$. Let $U := \{d_S \mid S\in \Sigma, a\notin S\}$ and $\val' := \val[a \mapsto \val(a) \setminus U]$. 

\begin{claimfirst}
	$(D,\val') \models \dbnfofo{\Sigma}$.
\end{claimfirst}

\begin{pfclaim}
First we show that the existential part of the normal form is satisfied. That is, that for every $S\in \Sigma$ we have a witness for the \emph{full} type $\tau_S(x)$. If $a\in S$ the witness is given by $\varphi^\tmono$, that is, $d_S$. If $a \notin S$ then we specially crafted $d_S$ to be a witness. The universal part is clearly satisfied.
\end{pfclaim}
To finish it is enough to observe that, by monotonicity of $\varphi$, we get $(D,\val) \models \varphi$. 
\end{proof}

Putting together the above lemmas we obtain Theorem~\ref{thm:ofomonot}. Moreover, a careful analysis of the translation gives us the following corollary, providing normal forms for the monotone fragment of $\ofo$.

\begin{corollary}\label{cor:ofopositivenf} Let $\varphi \in \ofo(A)$, the following hold:
	\begin{enumerate}[(i)]
		\item The formula $\varphi$ is monotone in $a \in A$ iff it is equivalent to a formula in the basic form $\bigvee \mondbnfofo{\Sigma}{a}$ for some types $\Sigma \subseteq \wp A$.
		\item The formula $\varphi$ is monotone in every $a\in A$ (i.e., $\varphi\in\ofo^+(A)$) iff $\varphi$ is equivalent to a formula in the basic form $\bigvee \posdbnfofo{\Sigma}$ for some types $\Sigma \subseteq \wp A$, where
		%
		%
		$$
		\posdbnfofo{\Sigma} := \bigwedge_{S\in\Sigma} \exists x. \tau^+_S(x) \land \forall x. \bigvee_{S\in\Sigma} \tau^+_S(x) .
		$$
	\end{enumerate}
\end{corollary}

\subsubsection{Monotone fragment of $\olque$}

\begin{theorem}\label{thm:olquemonot}
A formula of $\olque(A)$ is monotone in ${a\in A}$ iff it is equivalent to a sentence given by
$$
\varphi ::= \psi \mid a(x) \mid \exists x.\varphi(x) \mid \forall x.\varphi(x) \mid \varphi \land \varphi \mid \varphi \lor \varphi \mid \qu x.\varphi(x) \mid \dqu x.\varphi(x)
$$
where $\psi \in \olque(A\setminus \{a\})$. We denote this fragment as $\monot{\olque}{a}(A)$.
\end{theorem}

Observe that $x \foeq y$ and $x \not\foeq y$ are included in the case $\psi \in \olque(A\setminus \{a\})$. The result will follow from the following two lemmas.

\begin{lemma}\label{lem:monolqueismonot}
If $\varphi \in \monot{\olque}{a}(A)$ then $\varphi$ is monotone in $a$.
\end{lemma}
\begin{proof}
The proof is basically the same as Lemma~\ref{lem:monofoismonot}.
That is, we show by induction, that any one-step formula $\varphi$ in the fragment (which may not be a sentence) satisfies, for every one-step model $(D,\val:A\to\wp D)$ and assignment ${\ass:\fovar\to D}$, 
$$\text{If } (D,\val),\ass \models \varphi \text{ and } \val(a) \subseteq E \text{ then } (D,\val[a\mapsto E]),\ass \models \varphi.$$ 
We focus on the cases for the generalized quantifiers. Let $(D,\val),\ass \models \varphi$ and $\val(a) \subseteq E$.
\begin{enumerate}[$\bullet$]
\item Case $\varphi = \qu x.\varphi'(x)$. By definition there exists an infinite set $I\subseteq D$ such that for all $d\in I$ we have $(D,\val),\ass[x\mapsto d] \models \varphi'(x)$. By induction hypothesis $(D,\val[a\mapsto E]),\ass[x\mapsto d] \models \varphi'(x)$ for all $d \in I$. Therefore $(D,\val[a\mapsto E]),\ass \models \qu x.\varphi'(x)$.

\item Case $\varphi = \dqu x.\varphi'(x)$. By definition there exists $I\subseteq D$ such that for all $d\in I$ we have $(D,\val),\ass[x\mapsto d] \models \varphi'(x)$ and $D\setminus I$ is \emph{finite}. By induction hypothesis $(D,\val[a\mapsto E]),\ass[x\mapsto d] \models \varphi'(x)$ for all $d \in I$. Therefore $(D,\val[a\mapsto E]),\ass \models \dqu x.\varphi'(x)$.
\end{enumerate}
\end{proof}

\begin{lemma}
There exists a translation $(-)^\tmono:\olque(A) \to \monot{\olque}{a}(A)$ such that
a formula ${\varphi \in \olque(A)}$ is monotone in $a$ if and only if $\varphi\equiv \varphi^\tmono$.
\end{lemma}
\begin{proof}
We assume that $\varphi$ is in a normal form $\bigvee \dbnfolque{\vlist{T}}{\Pi}{\Sigma}$ where 
$$
\dbnfolque{\vlist{T}}{\Pi}{\Sigma} = \dbnfofoe{\vlist{T}}{\Pi \cup \Sigma} \land \dbnfinf{\Sigma} .
$$
for some sets of types $\Pi,\Sigma \subseteq \wp A$ and each $T_i \subseteq A$.
For the translation, we define
$$(\bigvee \dbnfolque{\vlist{T}}{\Pi}{\Sigma})^\tmono:= \bigvee \mondbnfolque{\vlist{T}}{\Pi}{\Sigma}{a}$$
where
\begin{eqnarray*}
\mondbnfolque{\vlist{T}}{\Pi}{\Sigma}{a} &:=& \mondbnfofoe{\vlist{T}}{\Pi \cup \Sigma}{a} \land \mondbnfinf{\Sigma}{a}\\
\mondbnfofoe{\vlist{T}}{\Lambda}{a} &:=& \exists \vlist{x}.\big(\arediff{\vlist{x}} \land \bigwedge_i \tau^a_{T_i}(x_i) \land \forall z.(\arediff{\vlist{x},z} \lthen \bigvee_{S\in \Lambda} \tau^a_S(z))\big) \\
\mondbnfinf{\Sigma}{a} &:=& \bigwedge_{S\in\Sigma} \qu y.\tau^a_S(y) \land \dqu y.\bigvee_{S\in\Sigma} \tau^a_S(y) .
\end{eqnarray*}
From the construction it is clear that $\varphi^\tmono \in \monot{\olque}{a}(A)$ and therefore the right-to-left direction of the lemma is immediate by Lemma~\ref{lem:monolqueismonot}. For the left-to-right direction assume that $\varphi$ is monotone in $a$, we have to prove that $(D,\val) \models \varphi$ iff $(D,\val) \models \varphi^\tmono$.

\bigskip
\noindent \fbox{$\Rightarrow$} This direction is trivial.

\bigskip
\noindent \fbox{$\Leftarrow$}
Assume $(D,\val) \models \varphi^\tmono$.
Observe that the elements of $D$ can be partitioned as follows:
\begin{enumerate}[(a)]
	\itemsep 0 pt
	\item Elements $t_i \in D$ witnessing the $a$-positive type $T_i \in \vlist{T}$,
	\item\label{it:dpi} Disjoint sets $\compset{D_S \subseteq D \mid S \in \Sigma}$ such that each $D_S$ is \emph{infinite} and every $d \in D_S$ is a witness for the $a$-positive type $S \in \Sigma$,
	\item\label{it:ds} A \emph{finite} set $D_\Pi \subseteq D$ of witnesses of the $a$-positive types in $\Pi$.
\end{enumerate}
In what follows we assume that every $d\in D$ is the witness of some fixed $a$-positive type $S_d \in \vlist{T}\cup \Pi \cup \Sigma$.
We define a new valuation $\val'$ as $\val'(b) = \val(b)$ for all $b \in A\setminus\{b\}$ and $\val'(a) := \compset{d\in D \mid a \in S_d}$. Observe that $\val'(a) \subseteq \val(a)$.
\begin{claimfirst}
	$(D,\val') \models \varphi$.
\end{claimfirst}
\begin{pfclaim}
	First we check that $(D,\val') \models \dbnfofoe{\vlist{T}}{\Pi \cup \Sigma}$. It is easy to see that the elements $t_i$ work as witnesses for the \emph{full} types $T_i$. That is $(D,\val') \models \tau_{T_i}(t_i)$ for every $i$. To prove the universal part of the formula it is enough to show that
	\begin{enumerate}
		\itemsep 0 pt
		\item Every element $d\in D_\Pi$ realizes the full type $S_d \in \Pi$,
		\item For all $S\in \Sigma$, every element of $D_S$ realizes the full type $S$.
	\end{enumerate}
	Let $d$ be an element of either $D_\Pi$ or any of the $D_S$. By~\ref{it:dpi} and~\ref{it:ds} we know $(D,\val) \models \tau_{S_d}^a(d)$. If $a\in S_d$ we can trivially conclude $(D,\val') \models \tau_{S_d}(d)$. If $a\notin S_d$, by definition of $\val'$ we know that $d \notin \val'(a)$ and hence we can also conclude that $(D,\val') \models \tau_{S_d}(d)$.

	To prove that $(D,\val') \models \bigwedge_{S\in\Sigma} \qu y.\tau_S(y) \land \dqu y.\bigvee_{S\in\Sigma} \tau_S(y)$ we only need to observe that the existential part is satisfied because each $D_S$ is infinite by~\ref{it:ds} and the universal part is satisfied because the set $D_\Pi \cup \vlist{T}$ is finite by~\ref{it:dpi}.
\end{pfclaim}
To finish the proof of the lemma, note that by monotonicity of $\varphi$ we get $(D,\val) \models \varphi$.
\end{proof}

Putting together the above lemmas we obtain Theorem~\ref{thm:olquemonot}. Moreover, a careful analysis of the translation gives us the following corollary, providing normal forms for the monotone fragment of $\olque$.

\begin{corollary}\label{cor:olquepositivenf}
	Let $\varphi \in \olque(A)$, the following hold:
	\begin{enumerate}[(i)]
		\item The formula $\varphi$ is monotone in $a \in A$ iff it is equivalent to a formula in the basic form $\bigvee \mondbnfolque{\vlist{T}}{\Pi}{\Sigma}{a}$ for some types $\Pi,\Sigma \subseteq \wp A$ and each $T_i \subseteq A$.
		\item The formula $\varphi$ is monotone in every $a\in A$ (i.e., $\varphi\in{\olque}^+(A)$) iff it is equivalent to a formula in the basic form $\bigvee \posdbnfolque{\vlist{T}}{\Pi}{\Sigma}$ for types $\Pi,\Sigma \subseteq \wp A$ and $T_i \subseteq A$, where
		%
		%
		\begin{align*}
			\posdbnfolque{\vlist{T}}{\Pi}{\Sigma} :=\ & \mondbnfofoe{\vlist{T}}{\Pi \cup \Sigma}{+} \land \posdbnfinf{\Sigma} \\
			\posdbnfofoe{\vlist{T}}{\Lambda} :=\ & \exists \vlist{x}.\big(\arediff{\vlist{x}} \land \bigwedge_i \tau^+_{T_i}(x_i) \land \forall z.(\arediff{\vlist{x},z} \lthen \bigvee_{S\in \Lambda} \tau^+_S(z))\big) \\
			\posdbnfinf{\Sigma} :=\ & \bigwedge_{S\in\Sigma} \qu y.\tau^+_S(y) \land \dqu y.\bigvee_{S\in\Sigma} \tau^+_S(y) .
		\end{align*}
	\end{enumerate}
\end{corollary}

%% file: continuity.tex

Consider a one-step logic $\llang_1(A)$ and formula $\varphi \in \llang_1(A)$.
We say that $\varphi$ is \emph{continuous in $a\in A$} if $\varphi$ is monotone in $a$ and additionally,
for every $(D,\val)$ and assignment $\ass:\fovar\to D$,
$$
\text{if } (D,\val),\ass \models \varphi \text{ then } \exists U \subseteq_\omega \val(a) \text{ such that } (D, \val[a \mapsto U]),\ass \models \varphi.
$$

\noindent It will be useful to give a syntactic characterization of continuity for several one-step logics.

\subsubsection{Continuous fragment of $\ofo$}

\begin{theorem}\label{thm:ofocont}
A formula of $\ofo(A)$ is continuous in ${a\in A}$ iff it is equivalent to a sentence given by
$$
\varphi ::= \psi \mid a(x) \mid \exists x.\varphi(x) \mid \varphi \land \varphi \mid \varphi \lor \varphi
$$
where
$\psi \in \ofo(A\setminus \{a\})$. We denote this fragment as $\cont{\ofo}{a}(A)$.
\end{theorem}
The theorem will follow from the next two lemmas.

\begin{lemma}\label{lem:cofoiscont}
If $\varphi \in \cont{\ofo}{a}(A)$ then $\varphi$ is continuous in $a$.
\end{lemma}
\begin{proof}
First observe that $\varphi$ is monotone in $a$ by Theorem~\ref{thm:ofomonot}.
We show, by induction, that any one-step formula $\varphi$ in the fragment (which may not be a sentence) satisfies, for every one-step model $(D,\val:A\to\wp D)$, assignment ${\ass:\fovar\to D}$,
$$
\text{If } (D,\val),\ass \models \varphi \text{ then } \exists U \subseteq_\omega \val(a) \text{ such that } (D,\val[a\mapsto U]),\ass \models \varphi.
$$
\begin{enumerate}[$\bullet$]
\item If ${\varphi = \psi \in \ofo(A\setminus \{a\})}$ changes in the $a$ part of the valuation will make no difference and hence the condition is trivial. 

\item Case $\varphi = a(x)$: if $(D,\val),\ass \models a(x)$ then $\ass(x)\in \val(a)$. Clearly, $\ass(x) \in \val[a\mapsto \{\ass(x)\}](a)$ and hence $(D, \val[a\mapsto \{\ass(x)\}]),\ass \models a(x)$. 

\item Case $\varphi = \varphi_1 \lor \varphi_2$: assume $(D,\val),\ass \models \varphi$. Without loss of generality we can assume that $(D,\val),\ass \models \varphi_1$ and hence by induction hypothesis we have that there is $U \subseteq_\omega \val(a)$ such that $(D,\val[a\mapsto U]),\ass \models \varphi_1$ which clearly satisfies $(D,\val[a\mapsto U]),\ass \models \varphi$. 

\item Case $\varphi = \varphi_1 \land \varphi_2$: assume $(D,\val),\ass \models \varphi$. By induction hypothesis we have $U_1,U_2 \subseteq_\omega \val(a)$ such that $(D,\val[a\mapsto U_1]),\ass \models \varphi_1$ and $(D,\val[a\mapsto U_2]),\ass \models \varphi_2$. By monotonicity this also holds with $\val[a\mapsto U_1 \cup U_2]$ and therefore $(D,\val[a\mapsto U_1 \cup U_2]),\ass \models \varphi$. 

\item Case $\varphi = \exists x.\varphi'(x)$ and $(D,\val),\ass \models \varphi$. By definition there exists $d\in D$ such that $(D,\val),\ass[x\mapsto d] \models \varphi'(x)$. By induction hypothesis there exists $U \subseteq_\omega \val(a)$ such that $(D,\val[a\mapsto U]),\ass[x\mapsto d] \models \varphi'(x)$ and hence $(D,\val[a\mapsto U]),\ass \models \exists x.\varphi'(x)$.
\end{enumerate}
\end{proof}

\begin{lemma}
There exists a translation $(-)^\tcont:\monot{\ofo}{a}(A) \to \cont{\ofo}{a}(A)$ such that
a formula ${\varphi \in \monot{\ofo}{a}(A)}$ is continuous in $a$ if and only if $\varphi\equiv \varphi^\tcont$.
\end{lemma}
\begin{proof}
To define the translation we assume that $\varphi$ is in basic form $\bigvee \mondbnfofo{\Sigma}{a}$ where
$$
\mondbnfofo{\Sigma}{a} := \bigwedge_{S\in\Sigma} \exists x. \tau^a_S(x) \land \forall x. \bigvee_{S\in\Sigma} \tau^a_S(x)
$$
for some types $\Sigma \subseteq \wp A$. 
For the translation, let
$(\bigvee \mondbnfofo{\Sigma}{a})^\tcont := \bigvee \mondbnfofo{\Sigma}{a}^\tcont$ and
$$
\mondbnfofo{\Sigma}{a}^\tcont := \bigwedge_{S\in\Sigma} \exists x. \tau^a_S(x) \land \forall x. \bigvee_{S\in\Sigma^{-}_{a}} \tau^a_S(x)
$$
where $\Sigma^{-}_{a} := \{S\in \Sigma \mid a\notin S\}$.

\bigskip
From the construction it is clear that $\varphi^\tcont \in \cont{\ofo}{a}(A)$ and therefore the right-to-left direction of the lemma is immediate by Lemma~\ref{lem:cofoiscont}. For the left-to-right direction assume that $\varphi$ is continuous in $a$, we have to prove that $(D,\val) \models \varphi$ iff $(D,\val) \models \varphi^\tcont$, for every one-step model $(D,\val)$. We will take a slightly different but equivalent approach.

It is easy to prove that $(D,\val) \equiv_\fo (D\times \omega,\val_\pi)$ where $D\times\omega$ has countably many copies of each element in $D$ and $\val_\pi(a) := \{(d,k) \mid d\in \val(a), k\in\omega\}$.
Moreover, as $\varphi$ is continuous in $a$ there is $U \subseteq_\omega \val_\pi(a)$ such that $\val'_\pi := \val[a \mapsto U]$ satisfies $(D\times\omega,\val_\pi) \models \varphi$ iff $(D\times\omega,\val'_\pi) \models \varphi$.
Therefore, it is enough to prove that $(D\times\omega,\val'_\pi) \models \varphi$ iff $(D\times\omega,\val'_\pi) \models \varphi^\tcont$. 

\bigskip
\noindent \fbox{$\Rightarrow$}
Let $(D\times\omega,\val'_\pi) \models \dbnfofo{\Sigma}$, we show that $(D\times\omega,\val'_\pi) \models \dbnfofo{\Sigma}^\tcont$. The existential part of $\dbnfofo{\Sigma}^\tcont$ is trivially true. We have to show that every element of $(D\times\omega,\val'_\pi)$ realizes a $a$-positive type in $\Sigma^{-}_{a}$. Let $(d,k) \in D\times\omega$ and $T$ be the (full) type of $(d,k)$. If $a\notin T$ then trivially $T\in \Sigma^{-}_{a}$ and we are done. Suppose $a\in T$. Observe that in $D\times\omega$ we have infinitely many copies of $d\in D$. However, as $\val'_\pi(a)$ is finite, there must be some $(d,k')$ with type $T' := T\setminus\{a\}$.
For $\dbnfofo{\Sigma}$ to be true we must have $T'\in \Sigma$ and hence $T'\in \Sigma^{-}_{a}$. It is easy to see that $(d,k)$ realizes the $a$-positive type $T'$.

\bigskip
\noindent \fbox{$\Leftarrow$}
Let $(D\times\omega,\val'_\pi) \models \dbnfofo{\Sigma}^\tcont$, we show that $(D\times\omega,\val'_\pi) \models \dbnfofo{\Sigma}$. The existential part is again trivial. For the universal part just observe that $\Sigma^{-}_{a} \subseteq \Sigma$.
\end{proof}

Putting together the above lemmas we obtain Theorem~\ref{thm:ofocont}. Moreover, a careful analysis of the translation gives us the following corollary, providing normal forms for the continuous fragment of $\ofo$.

\begin{corollary}\label{cor:ofocontinuousnf}
	Let $\varphi \in \ofo(A)$, the following hold:
	\begin{enumerate}[(i)]
		\item The formula $\varphi$ is continuous in $a \in A$ iff it is equivalent to a formula in the basic form $\bigvee \chi^a_\Sigma$ for some types $\Sigma \subseteq \wp A$, where $\chi^a_\Sigma$ is given by
		$$
		\chi^a_\Sigma := \bigwedge_{S\in\Sigma} \exists x. \tau^a_S(x) \land \forall x. \bigvee_{S\in\Sigma^-_a} \tau^a_S(x) .
		$$
		\item If $\varphi$ is monotone in $A$ (i.e., $\varphi\in\ofo^+(A)$) then $\varphi$ is continuous in $a \in A$ iff it is equivalent to a formula in the basic form $\bigvee \chi^+_\Sigma$ for some types $\Sigma \subseteq \wp A$, where 
		$$
		\chi^+_\Sigma := \bigwedge_{S\in\Sigma} \exists x. \tau^+_S(x) \land \forall x. \bigvee_{S\in\Sigma^-_a} \tau^+_S(x) .
		$$
	\end{enumerate}
\end{corollary}

\subsubsection{Continuous fragment of $\olque$}

\begin{theorem}\label{thm:olquecont}
A formula of $\olque(A)$ is continuous in ${a\in A}$ iff it is equivalent to a sentence given by
$$
\varphi ::= \psi \mid a(x) \mid \exists x.\varphi(x) \mid \varphi \land \varphi \mid \varphi \lor \varphi \mid \wqu x.(\varphi,\psi)
$$
where $\psi \in \olque(A\setminus \{a\})$ and $\wqu x.(\varphi,\psi) := \forall x.(\varphi(x) \lor \psi(x)) \land \dqu x.\psi(x)$. We denote this fragment as $\cont{\olque}{a}(A)$.
\end{theorem}

Universal quantification is usually problematic for preserving continuity because of its potentially infinite nature. However, in the case of $\wqu x.(\varphi,\psi)$, the combination of both quantifiers ensures that all the elements are covered by $\varphi \lor \psi$ but only \emph{finitely many} are required to make $\varphi$ true (which contains $a\in A$). This gives no trouble for continuity in $a$.

The theorem will follow from the next two lemmas.

\begin{lemma}\label{lem:colqueiscont}
If $\varphi \in \cont{\olque}{a}(A)$ then $\varphi$ is continuous in $a$.
\end{lemma}
\begin{proof}
Observe that monotonicity of $\varphi$ is guaranteed by Theorem~\ref{thm:olquemonot}.
We show, by induction, that any formula of the fragment (which may not be a sentence) satisfies, for every one-step model $(D,\val:A\to\wp D)$ and assignment ${\ass:\fovar\to D}$,
$$\text{If } (D,\val),\ass \models \varphi \text{ then } \exists U \subseteq_\omega \val(a) \text{ such that } (D,\val[a\mapsto U]),\ass \models \varphi.$$
We focus on the inductive case of the new quantifier. Let $\varphi' = \wqu x.(\varphi,\psi)$, i.e., 
$$\varphi' = \forall x.\underbrace{(\varphi(x) \lor \psi(x))}_{\alpha(x)} \land \underbrace{\dqu x.\psi(x)}_\beta.$$
Let $(D,\val),\ass \models \varphi'$. By induction hypothesis,
for every $\ass_d := \ass[x\mapsto d]$ which satisfies $(D,\val),\ass_d \models \alpha(x)$ there is $U_d \subseteq_\omega \val(a)$ such that $(D,\val[a\mapsto U_d]),\ass_d \models \alpha(x)$. The crucial observation is that because of $\beta$, 
only finitely many elements of $d$ make $\psi(d)$ false. Let $U := \bigcup \{U_d \mid (D,\val), \ass_d \not\models \psi(x) \}$. Note that $U$ is a finite union of finite sets, hence finite.
\begin{claimfirst}
	Let $\val_U := \val[a\mapsto U]$, we have $(D,\val_U),\ass \models \varphi'$.
\end{claimfirst}
\begin{pfclaim}
	It is clear that $(D,\val_U),\ass \models \beta$ because $\psi$ is $a$-free. To show that the first conjunct is true we have to show that $(D,\val_U),\ass_d \models \varphi(x) \lor \psi(x)$ for every $d\in D$. We consider two cases: (i) if $(D,\val),\ass_d \models \psi(x)$ we are done, again because $\psi$ is $a$-free; (ii) if the former is not the case then $U_d \subseteq U$; moreover, we knew that $(D,\val[a\mapsto U_d]),\ass_d \models \alpha(x)$ and by monotonicity of $\alpha(x)$ we can conclude that $(D,\val_U),\ass_d \models \alpha(x)$.
\end{pfclaim}
This finishes the proof of the lemma.
\end{proof}

\begin{lemma}\label{lem:olquectrans}
	There exists a translation $(-)^\tcont:\monot{\olque}{a}(A) \to \cont{\olque}{a}(A)$ such that
a formula $\varphi \in \monot{\olque}{a}(A)$ is continuous in $a$ if and only if $\varphi\equiv \varphi^\tcont$.
\end{lemma}
\begin{proof} We assume that $\varphi$ is in basic normal form, i.e., $\varphi = \bigvee \mondbnfolque{\vlist{T}}{\Pi}{\Sigma}{a}$.
For the translation let $(\bigvee \mondbnfolque{\vlist{T}}{\Pi}{\Sigma}{a})^\tcont := \bigvee \mondbnfolque{\vlist{T}}{\Pi}{\Sigma}{a}^\tcont$ and
$$
\mondbnfolque{\vlist{T}}{\Pi}{\Sigma}{a}^\tcont :=
\begin{cases}
	\bot &\text{ if } a\in \bigcup \Sigma\\
	\mondbnfolque{\vlist{T}}{\Pi}{\Sigma}{a} &\text{ otherwise}.
\end{cases}
$$

First we prove the right-to-left direction of the lemma. By Lemma~\ref{lem:colqueiscont} it is enough to show that $\varphi^\tcont \in \cont{\olque}{a}(A)$. We focus on the disjuncts of $\varphi^\tcont$. The interesting case is when $a\notin \bigcup \Sigma$. If we rearrange $\mondbnfolque{\vlist{T}}{\Pi}{\Sigma}{a}$ and define the formulas $\varphi', \psi$ as follows:
\begin{align*}
\mondbnfolque{\vlist{T}}{\Pi}{\Sigma}{a} \equiv \exists \vlist{x}.\Big(& \arediff{\vlist{x}} \land \bigwedge_i \tau^a_{T_i}(x_i)\ \land \\
& \forall z.(\underbrace{\lnot\arediff{\vlist{x},z} \lor \bigvee_{S\in \Pi} \tau^a_S(z)}_{\varphi'(\vlist{x},z)} \lor \underbrace{\bigvee_{S\in \Sigma} \tau^a_S(z)}_{\psi(z)})\ \land \\
& \dqu y.\underbrace{\bigvee_{S\in\Sigma} \tau^a_S(y)}_{\psi(y)} \Big) \land \bigwedge_{S\in\Sigma} \qu y.\tau^a_S(y),
\end{align*}
then we get that
$$
\mondbnfolque{\vlist{T}}{\Pi}{\Sigma}{a} \equiv \exists \vlist{x}.\Big(\arediff{\vlist{x}} \land \bigwedge_i \tau^a_{T_i}(x_i) \land \wqu z.(\varphi'(\vlist{x},z),\psi(z)) \Big) \land \bigwedge_{S\in\Sigma} \qu y.\tau^a_S(y)
$$
which, because $a\notin \bigcup \Sigma$, is in the required fragment.

For the left-to-right direction of the lemma we have to prove that $\varphi \equiv \varphi^\tcont$.

\bigskip
\noindent\fbox{$\Leftarrow$} Let $(D,\val) \models \varphi^\tcont$. The only difference between $\varphi$ and $\varphi^\tcont$ is that some disjuncts may have been replaced by $\bot$. Therefore this direction is trivial.

\bigskip
\noindent\fbox{$\Rightarrow$} Let $(D,\val) \models \varphi$. Because $\varphi$ is continuous in $a$ we may assume that $\val(a)$ is finite. Let $\mondbnfolque{\vlist{T}}{\Pi}{\Sigma}{a}$ be a disjunct of $\varphi$ such that $(D,\val) \models \mondbnfolque{\vlist{T}}{\Pi}{\Sigma}{a}$. If $a \notin \bigcup\Sigma$ we trivially conclude that $(D,\val) \models \varphi^\tcont$ because the disjunct remains unchanged. Suppose now that $a\in \bigcup\Sigma$, then there must be some $S\in\Sigma$ with $a\in S$. Because $(D,\val) \models \mondbnfolque{\vlist{T}}{\Pi}{\Sigma}{a}$ we have, in particular, that $(D,\val) \models \qu y.\tau^a_S(x)$ and hence $\val(a)$ must be infinite which is absurd.
\end{proof}

Putting together the above lemmas we obtain Theorem~\ref{thm:olquecont}. Moreover, a careful analysis of the translation gives us the following corollary, providing normal forms for the continuous fragment of $\olque$.

\begin{corollary}\label{cor:olquecontinuousnf}
Let $\varphi \in \olque(A)$, the following hold:
	\begin{enumerate}[(i)]
		\item The formula $\varphi$ is continuous in $a \in A$ iff it is equivalent to a formula in the basic form $\bigvee \mondbnfolque{\vlist{T}}{\Pi}{\Sigma}{a}$ for some types $\Pi,\Sigma \subseteq \wp A$ and $T_i \subseteq A$ such that $a\notin \bigcup\Sigma$.
		\item If $\varphi$ is monotone in every element of $A$ (i.e., $\varphi\in{\olque}^+(A)$) then $\varphi$ is continuous in $a \in A$ iff it is equivalent to a formula in the basic form $\bigvee \posdbnfolque{\vlist{T}}{\Pi}{\Sigma}$ for some types $\Pi,\Sigma \subseteq \wp A$ and $T_i \subseteq A$ such that $a\notin \bigcup\Sigma$.
	\end{enumerate}
\end{corollary} 

%% file: cocontinuity.tex

Consider a one-step logic $\llang_1(A)$ and formula $\varphi \in \llang_1(A)$.
We say that $\varphi$ is \emph{co-continuous in $a\in A$} if $\varphi$ is monotone in $a$ and,
for all $(D,\val)$ and $\ass:\fovar\to D$,
$$
\text{if } (D,\val),\ass \not\models \varphi \text{ then } \exists U \subseteq_\omega \val(a) \text{ such that } (D, \val[a \mapsto D\setminus U]),\ass \not\models \varphi.
$$

Observe that, already with the abstract definition of Boolean dual given in Definition~\ref{d:bdual1} we can prove the expected relationship between the notions of continuity and co-continuity.

\begin{proposition}\label{prop:contdualcocont}
	Let $\varphi \in \llang_1(A)$, $\varphi$ is continuous in $a\in A$ iff $\varphi^\delta$ is co-continuous in $a$.
\end{proposition}
\begin{proof}
Let $\varphi$ be continuous in $a$, we prove that $\varphi^\delta$ is co-continuous in $a$:
\begin{align*}
(D,\val) \not\models \varphi^\delta
& \text{ iff } (D,\val^c) \models \varphi & \tag{Definition~\ref{d:bdual1}} \\
& \text{ iff } \exists U \subseteq_\omega \val^c(a) \text{ such that } (D,\val^c[a\mapsto U]) \models \varphi & \tag{$\varphi$ continuous} \\
& \text{ iff } \exists U' \subseteq_\omega \val(a) \text{ such that } (D,\val[a\mapsto D\setminus U']) \not\models \varphi & \tag{Definition of $\val^c$}
\end{align*}
The proof of the other direction is analogous.
\end{proof}

To define a syntactic notion of co-continuity we first give a concrete definition of the dualization operator of Definition~\ref{d:bdual1} and then show that the one-step language $\olque$ is closed under Boolean duals.

\begin{definition}\label{DEF_dual} 
Let $\varphi \in {\olque}(A)$. 
The \emph{dual} $\varphi^{\delta} \in {\olque}(A)$ of $\varphi$ is defined 
as follows.
\begin{align*}
 (a(x))^{\delta} & :=  a(x) 
\\ (\top)^{\delta} & :=  \bot 
  & (\bot)^{\delta} & :=  \top 
\\  (x \approx y)^{\delta} & :=  x \not\approx y 
  & (x \not\approx y)^{\delta}& :=  x \approx y 
\\ (\varphi \wedge \psi)^{\delta} &:=  (\varphi)^{\delta} \vee (\psi)^{\delta} 
  &(\varphi \vee \psi)^{\delta}& :=  (\varphi)^{\delta} \wedge (\psi)^{\delta}
\\ (\exists x.\psi)^{\delta} &:=  \forall x.(\psi)^{\delta} 
  &(\forall x.\psi)^{\delta} &:=  \exists x.(\psi)^{\delta} 
\\ (\exists^{\infty} x.\psi)^{\delta} &:= \forall^{\infty} x.(\psi)^{\delta} 
  &(\forall^{\infty} x.\psi)^{\delta} &:=  \exists^{\infty} x.(\psi)^{\delta}
\end{align*} 
\end{definition}

\begin{remark}
	Observe that if $\varphi \in {\ofo}(A)$ then $\varphi^{\delta} \in {\ofo}(A)$ and that the operator preserves positivity of the predicates. That is, if $\varphi \in {\olque}^+(A)$ then $\varphi^{\delta} \in {\olque}^+(A)$ and the same occurs with $\ofo^+(A)$.
\end{remark}

The proof of the following Proposition is a routine check.

\begin{proposition}\label{prop:duals}
The sentences $\phi$ and $\phi^{\delta}$ are Boolean duals, for every $\phi 
\in \olque(A)$.
\end{proposition}

We are now ready to give the syntactic definition of a co-continuous fragment for the one-step logics into consideration.

\begin{definition}\label{def:cocontfrag}
	Let $A$ be a set of names. The syntactic fragments of $\olque(A)$ and $\ofo(A)$ which are \emph{co-continuous} in $a\in A$ are given by
	\begin{align*}
		\cocont{\olque}{a}(A) &:= \{\varphi \mid \varphi^\delta \in \cont{\olque}{a}(A)\} &
		\cocont{\ofo}{a}(A) &:= \{\varphi \mid \varphi^\delta \in \cont{\ofo}{a}(A)\} .
	\end{align*}
\end{definition}

\begin{proposition}
	A formula $\varphi \in \olque(A)$ is co-continuous in $a\in A$ iff $\varphi \in \cocont{\olque}{a}(A)$.
\end{proposition}

\begin{proof} This is a consequence of Proposition~\ref{prop:contdualcocont}, Theorem~\ref{thm:olquecont} and Definition~\ref{def:cocontfrag}.
\end{proof}

%% file: msoaut.tex

In this section we provide an automata-theoretic characterization of $\wmso$.
As argued in the introduction, $\yvWMSO$-automata will be defined as the
automata in $\yvAut(\olque)$ that satisfy two additional properties: weakness
and continuity.
We now briefly discuss these properties in a slightly more general setting,
starting with weakness.

\begin{definition}
\label{def:weak}
Let $\yvLo$ be a one-step language, and let $\bbA = \tup{A,\tmap,\pmap,a_I}$
be in $\yvAut(\yvLo)$.
Given two states $a,b$ of $\bbA$,
we say that there is a transition from $a$ to $b$, notation: $a \leadsto b$,
if $b$ occurs in $\De(a,c)$ for some $c \in C$.
We let the \emph{reachability} relation $\ord$ denote the reflexive-transitive
closure of the relation $\leadsto$.

A \emph{strongly connected $\ord$-component} ($\ord$-SCC) is a subset $M\subseteq A$ such that for every $a,b \in M$ we have $a \ord b$ and $b \ord c$. The SCC is called \emph{maximal} (MSCC) when $M\cup\{a\}$ ceases to be a SCC for any choice of $a \in A\setminus M$.

We say that $\Om$ is a \emph{weak} parity condition, and $\bbA$ is a
\emph{weak} parity automaton if we have
\begin{description}
\item[(weakness)] if $a \ord b$ and $b \ord a$ then $\pmap(a) = \pmap(b)$.
\end{description}
\end{definition}

\begin{remark}
Any weak parity automaton $\bbA$ is equivalent to a weak parity automaton
$\bbA'$ with $\pmap: A' \to \{0,1\}$. From now on we therefore consider only
weak parity automata with priorities $0$ and $1$.
\end{remark}

As explained in the introduction, the leading intuition is that weak parity automata are those unable to register non-trivial properties concerning the vertical `dimension' of input trees. Indeed on trees they characterize $\WFMSO$, that is, the fragment of $\MSO$ where the quantification is restricted to \emph{well-founded} subsets of trees (corresponding to the notion of \emph{noetherian} subset if we consider arbitrary transition systems). We refer to the literature on weak automata \cite{MullerSaoudiSchupp92} and on $\WFMSO$ \cite{DBLP:conf/lics/FacchiniVZ13,Zanasi:Thesis:2012} for more details.

We now turn to the second condition that we will be interested in,
viz., continuity. Intuitively, this property expresses a constraint on how much of the horizontal `dimension' of an input tree the automaton is allowed to process. It will be instrumental in showing that the class of automata that we are going to shape characterizes $\wmso$. The idea is that, as $\wmso$-quantifiers range over finite sets, the weakeness condition corresponds to those sets being `vertically' finite (i.e. included in well-founded subtrees), whereas the continuity condition corresponds to them being `horizontally' finite (i.e. included in finitely branching subtrees).

First we formulate our continuity condition abstractly in the setting of $\yvAut(\yvLo)$.
Given the semantics of the one-step language $\yvLo$, the (semantic) notion
of (co-)continuity applies to one-step formulas (see for instance
section~\ref{subsec:one-stepcont}). We can then formulate the following requirement on automata from $\yvAut(\yvLo)$:

\begin{description}
\item[(continuity)] let $a,b$ be states such that both $a\ord b$ and
$b \ord a$, and let $c\in C$;
    if ${\pmap(b)}=1$ then $\tmap(b,c)$ is continuous in $a$.
    If $\pmap(b)=0$, then $\tmap(b,c)$ is co-continuous in $a$.
\end{description}

For the automata used in this article we need to combine the constraints for the horizontal and vertical dimensions, yielding automata with both the weakness and continuity constraints.

\begin{definition}
A \emph{continuous-weak parity automaton} is an automaton $\aut \in \yvAut(\yvLo)$ addittionally satisfying both the \textbf{(weakness)} and \textbf{(continuity)} conditions.
We let $\yvcwAut(\llang_1)$ denote the class of such automata.
\end{definition}

Observe that, so far, the continuity condition is given semantically.
However, given that the one-step languages that we are interested in have
a \emph{syntactic characterization} of continuity (see for example Theorem~\ref{thm:olquecont})
we will give concrete definitions of these automata that take advantage of the mentioned characterizations.

\begin{definition}
A \emph{$\wmso$-automaton} $\aut = \tup{A,\Delta,\Omega,a_I}$ is an automaton $\aut \in \yvAut(\olque)$ such that for all states $a,b \in A$ with $a \ord b$ and $b\ord a$ the following conditions hold:
\begin{description}
	\itemsep 0 pt
	\item[(weakness)] $\pmap(a)=\pmap(b)$,
	\item[(continuity)] if $\pmap(a)$ is odd (resp. even) then, for each $c\in C$ we have
	   $\tmap(a,c) \in \cont{{\olque}^+}{b}(A)$ (resp. $\tmap(a,c) \in \cocont{{\olque}^+}{b}(A)$).
\end{description}
As the class of such automata coincides with $\yvcwAut(\olque)$ we use the same notation for it.
\end{definition}

We have now arrived at the main theorem of this section, which takes care of
one direction of Theorem~\ref{t:m1}. The main theorem of this section states that

\begin{theorem}
\label{t:wmsoauto}
There is an effective construction transforming a $\yvWMSO$-formula $\phi$
into a $\yvWMSO$-automaton $\bbA_{\phi}$ that is equivalent
to $\phi$ on the class of trees.
That is, for any tree $\bbT$,
\begin{equation}
\bbA_{\phi} \text{ accepts } \bbT \text{ iff } \bbT \models {\phi}.
\end{equation}
\end{theorem}

As usual, the proof of this theorem proceeds by induction on the complexity of
$\phi$.
For the inductive steps of this proof, we need to verify that the class of
\wmso-automata is closed under the boolean operations and finite projection.
Clearly, the latter closure property requires most of the work; we first
provide a simulation theorem that put $\wmso$-automata in a suitable shape
for the projection construction.
The proofs in the next two sections are (nontrivial) modifications of the
analogous results proved in~\cite{DBLP:conf/lics/FacchiniVZ13}.
It will be convenient to use the following terminology.

\begin{definition}
The \emph{tree language $\trees(\aut)$ recognized} by an automaton $\bbA$ is
defined as the collection of trees that are accepted by $\bbA$.
A class of trees is \emph{$\yvWMSO$-automaton recognizable} (or
\emph{recognizable}, if clear from context), if it is of the form
$\trees(\aut)$ for a $\yvWMSO$-automaton $\bbA$.
\end{definition}

%% file: autchar.tex

\subsection{Simulation theorem}
\input{simulation.tex}

\subsection{Closure Properties}

In this subsection we prove that $\yvWMSO$-automata are closed under the 
operations corresponding to the connectives of $\yvMSO$, that is: union, 
complementation and projection with respect to finite sets.
We start with the latter.

\input{projection.tex}

\input{complementation.tex}

\subsection{Proof of Theorem \ref{t:wmsoauto}}

\begin{proof} The proof is by induction on $\varphi$.
\begin{itemize}
  \item For the base case $\varphi = p \inc q$, the corresponding 
  $\wmso$-automaton is provided in \cite[Ex. 2.6]{Zanasi:Thesis:2012}. 
  For the base case $\varphi = R(p,q)$, we give the corresponding 
  $\wmso$-automaton $\aut_{R(p,q)} = \tup{A,\Delta,\Omega,a_I}$ below:
\begin{eqnarray*}
        A &:=& \{a_0,a_1\}\\
        a_I &:=& a_0\\
  \Delta(a_0,c) &:=& \left\{
	\begin{array}{ll}
           \exists x. a_1(x) \wedge \forall y. a_0(y) & \mbox{if }p \in c 
	\\ \forall x\ (a_0(x)) & \mbox{otherwise}
	\end{array}
\right. \\
  \Delta(a_1,c) &:=& \left\{
	\begin{array}{ll}
        \top & \mbox{if }q \in c \\
        \bot & \mbox{otherwise}
	\end{array}
\right. \\
    \Omega(a_0) &:=& 0\\
    \Omega(a_1) &:=& 1.
\end{eqnarray*}
Note that the $\mso$-automaton for $R(p,q)$ provided in 
\cite[Ex. 2.5]{Zanasi:Thesis:2012} is \emph{not} a $\wmso$-automaton, as the 
continuity property does not hold.

\item
For the Boolean cases, where $\varphi = \psi_1 \vee \psi_2$ or $\phi = \neg\psi$
we refer to the closure properties of recognizable tree languages, see 
Theorem~\ref{t:cl-dis} and Theorem~\ref{t:cl-cmp}, 
respectivel.
  
\item 
For the case $\varphi = \exists p. \psi$, consider the following chain of
equivalences, where $\aut_{\psi}$ is given by the inductive hypothesis and 
${{\exists}_F p}.\aut_{\psi}$ is constructed according to 
Definition \ref{DEF_fin_projection}:
\begin{alignat*}{2}
{{\exists}_F p}.\aut_{\psi} \text{ accepts }\mb{T} 
   & \text{ iff }
     \aut_{\psi} \text{ accepts } \mb{T}[p \mapsto X], 
     \text{ for some } X \sse_{\om} T
   & \quad\text{(Lemma~\ref{PROP_fin_projection})}
\\ & \text{ iff }
     \mb{T}[p \mapsto X] \models \psi,
     \text{ for some } X \sse_{\om} T
   & \quad\text{(induction hyp.)}
\\ & \text{ iff }
    \mb{T} \models \exists p. \psi
   & \quad\text{(semantics $\wmso$)}
\end{alignat*}
\end{itemize}
\end{proof}

%% file: simulation.tex

In this section we show that any $\wmso$-automaton $\aut$ can be simulated by a ``two-sorted'' $\wmso$-automaton $\mb{A}^{\f}$. The leading intuition is that $\mb{A}^{\f}$ will consist of one copy of $\aut$ (based on a set of states $A$) and a variant of its powerset construction, which will be based both on states from $A$ and ``macro-states'' from $\pw (A \times A)$.\footnote{It is customary for powerset constructions on parity automata to encode macro-states as binary relations between states (from $\pw (A \times A)$) instead of plain sets (from $\pw A$). Such additional structure is needed to correctly associate with a run on macro-states the corresponding bundle of runs of the original automaton $\aut$. We refer to the standard literature on parity automata (e.g. \cite{Walukiewicz96,ALG02}) for further details.} Successful runs of $\mb{A}^{\f}$ will have the property of processing only a \emph{finite} amount of the input with $\mb{A}^{\f}$ being in a macro-state and all the rest with $\mb{A}^{\f}$ behaving exactly as $\aut$.

To achieve this result, we first need some preliminary definitions. The following is a notion of lifting for types on states that is instrumental in defining a translation to types on macro-states. The distinction between empty and non-empty subsets of $A$ is to make sure that empty types on $A$ are lifted to empty types on $\pw A$.

\begin{definition}\label{def:typelifting}
Let $A$ be a set of unary predicates.
Given a set $\Sigma \subseteq \wp A$, its \emph{lifting} is the set
$\lift{\Sigma} \subseteq \wp \wp A$ defined as follows:
\begin{eqnarray*}
\lift{\Sigma} & := & \{\{S\} \mid S \in \Sigma \wedge S \neq \emptyset\} \cup
    \{\emptyset \mid \emptyset \in \Sigma \}.
\end{eqnarray*}
\end{definition}

The next step is to define a translation on the sentences associated with the
transition function of the original $\wmso$-automaton, say with set of states $A$. Following the intuition given above, the idea is that we want to work with sentences that can be made true by assigning macro-states (from $\p(A \times A)$) to finitely many nodes in the model, and ordinary states (from $A$) to all the other nodes. Henceforth, we use the notation $\shA$ for the set $\p(A \times A)$.

\begin{definition}\label{DEF_finitary_lifting}
Let $\varphi \in {\olque}^+(A \times A)$ be a formula of shape $\posdbnfolque{\vlist{T}}{\Pi}{\Sigma}$ for some $\Pi,\Sigma \subseteq \shA$ and $\vlist{T} = \{T_1,\dots,T_k\} \subseteq \shA$. Let $\widetilde{\Sigma}\subseteq \wp A$ be $\widetilde{\Sigma} := \{\Ran(S) \mid S \in \Sigma\}$. We define $\varphi^{\f} \in {\olque}^+(A \cup \shA )$ as follows:
$$(\posdbnfolque{\vlist{T}}{\Pi}{\Sigma})^{\f} := \posdbnfolque{\lift{\vlist{T}}}{\lift{\Pi} \cup \lift{\Sigma}}{\widetilde{\Sigma}}. $$
Observe that each ${\tau}^{+}_{P}$ with $P \in \widetilde{\Sigma}$ appearing in $\varphi^{\f}$ is a (positive) $A$-type, as $P = \Ran(S) \subseteq A$ for some $S \in \Sigma$.
\end{definition}

\noindent Our desiderata on the translation $(\cdot)^{\f}$ concern the notions of \emph{continuity} and \emph{functionality}.

\begin{definition} Given a set $A$ of unary predicates and $B \subseteq A$, we say that a sentence $\varphi \in {\olque}^+(A)$ is \emph{functionally continuous in $B$} if, for every model $(D,\val \: A \to \p(D))$,
\begin{align*}
\text{if } (D,\val),\ass \models \varphi \text{ then } & \exists\ \val' \: A \to \p(D) \text{ such that } (D, \val'),\ass \models \varphi, \\
& \val'(a)\subseteq \val(a) \text{ for all } a \in A, \tag{$\val'$ is a restriction of $\val$}\\
 & \val'(b) \text{ is finite for all }b \in B \text{ and } \tag{continuity in $B$}\\
 & \val'(b)\cap \val'(a) = \emptyset \text{ for all } a \in A\setminus\{b\} \text{ and }b \in B\tag{functionality in $B$}.
\end{align*}
\end{definition}
In words, $\varphi$ is functionally continuous in $B$ if it is continuous in each $b \in B$ and, for each model $(D,\val)$ where $\varphi$ is true, there is a restriction $\val'$ of $\val$ which both witnesses continuity and does not assign any other $a \in A$ to the elements marked with some $b \in B$.

\begin{lemma}\label{LEM_cont}
Let $\varphi \in {\olque}^+(A \times A)$ and $\varphi^{\f}\in {\olque}^+(A\cup \shA )$ be given as in Definition~\ref{DEF_finitary_lifting}. Then $\varphi^{\f}$ is functionally continuous in $\shA$.
 \end{lemma}

\begin{proof}
We first unfold the definition of $\varphi^{\f}$ as follows:
\begin{align*}
\varphi^{\f} =\ &
\underbrace{
    \exists \vlist{x}.\big(\arediff{\vlist{x}} \land \bigwedge_{0 \leq i \leq n} \tau^+_{\lift{T}_i}(x_i)
}_{\psi_1}
\land \underbrace{
    \forall z.(\arediff{\vlist{x},z} \lthen \bigvee_{S\in \lift{\Pi} \cup \lift{\Sigma} \cup \widetilde{\Sigma}} \tau^+_S(z))\big)
}_{\psi_2}
\land
\\ & \underbrace{
    \bigwedge_{P\in\widetilde{\Sigma}} \qu y.{\tau}^{+}_P(y)
}_{\psi_3} \land
 \underbrace{
    \dqu y.\bigvee_{P\in\widetilde{\Sigma}} {\tau}^{+}_P(y)
}_{\psi_4} .
\end{align*}
Observe that $\psi_1 \land \psi_2$ is just $\mondbnfofoe{\lift{\vlist{T}}}{\lift{\Pi} \cup \lift{\Sigma} \cup \widetilde{\Sigma}}{+}$. Now suppose that $(D,\val \: (A \cup \shA ) \to \p(D))$ is a model where $\varphi^{\f}$ is true. This amounts to the truth of subformulas $\psi_1$, $\psi_2$, $\psi_3$ and $\psi_4$ whose syntactic shape yields information on the types of elements of $D$. In particular, we can define a partition of $D$ into subsets $D_1$, $D_2$, $D'_2$ as follows:
\begin{itemize}
  \item As $\psi_1$ is true, we can pick $n$ distinct elements $s_1,\dots,s_n$ of $D$ such that $s_i$ witnesses the positive type $\lift{T}_i$, 
   that is, $s_i \in \val(S)$ for each $S \in \lift{T}_i$. We define $D_1 := \{s_1,\dots,s_n\}$.
  \item  As $\psi_2$ is true, we can cover all the elements not in $D_1$ with two disjoint sets $D_2$ and $D'_2$ given as follows. The set $D_2$ is defined to contain all the elements not in $D_1$ witnessing a type ${\tau}^{+}_P(z)$ with $P \in \widetilde{\Sigma}$. The set $D'_2$ is just the complement of $D_1 \cup D_2$: by syntactic shape of $\psi_2$, all elements of $D'_2$ witness a positive type ${\tau}^{+}_S$ with
  $S \in \lift{\Pi} \cup \lift{\Sigma}$.
  \item The truth of the subformula $\psi_4$ yields the information that the set $D_1 \cup D'_2$ is finite. If $\widetilde{\Sigma}$ is non-empty, the truth of $\psi_3$ implies that the set $D_2$ is infinite.
 \end{itemize}
This partition uniquely associates with each $s \in D$ a type ${\tau}^{+}_S$ witnessed by $s$ and thus a set of unary predicates $S_s := S \subseteq A \cup \shA$. We can then define a valuation $\val'$ assigning to each element $s$ of $D$ exactly the set $S_s$.

We now check the properties of $\val'$. As the partition inducing $\val'$ follows the syntactic shape of $\varphi^{\f}$, one can observe that $\val'$ is a restriction of $\val$ and $(D,\val')$ makes $\varphi^{\f}$ true. By definition of the partition, $\val'$ assigns unary predicates from $\shA$ only to elements in the finite set $D_1 \cup D'_2$, meaning that $\varphi^{\f}$ is continuous in $\shA$. Furthermore, $\val'$ assigns at most one unary predicate from $\shA$ to each element of $D_1 \cup D'_2$, because $\lift{\vlist{T}} \cup \lift{\Pi} \cup \lift{\Sigma}$ is defined as the lifting of $\vlist{T} \cup \Pi \cup \Sigma$. It follows that $\varphi^{\f}$ is also functional in $\shA$. Since the same restriction $\val'$ yields both properties, $\varphi^{\f}$ is functionally continuous in $\shA$.
\end{proof}

\begin{remark} As $\varphi^{\f}$ is of shape $\posdbnfolque{\lift{\vlist{T}}}{\lift{\Pi} \cup \lift{\Sigma}}{\widetilde{\Sigma}}$ with $R \not\in \bigcup\widetilde{\Sigma}$ for each $R \in \shA$, by application of Corollary \ref{cor:olquecontinuousnf} we would immediately get that $\varphi^{\f}$ is continuous in each $R \in \shA$. However, we do not use this observation in proving Lemma \ref{LEM_cont} and propose instead a more direct argument, allowing to show both continuity and functionality at once.
\end{remark}

The next definition is standard (see e.g.  \cite{Walukiewicz96,Ven08}) as an intermediate step to define the transition function of the powerset construct for parity automata.

\begin{definition}\label{DEF_delta star} Let $\mb{A} = \tup{A,\Delta,\Omega,a_I}$ be a $\wmso$-automaton. Fix $a \in A$ and $c \in C$. The sentence $\Delta^{\star}(a,c)$ is defined as
\begin{eqnarray*}
        \Delta^{\star}(a,c) &:=& \Delta(a,c)[b \mapsto (a,b) \mid b \in A].
      \end{eqnarray*}
\end{definition}

 Next we combine the previous definitions to characterize the transition function associated with the macro-states.

\begin{definition}\label{PROP_DeltaPowerset}
Let $\aut = \tup{A,\Delta,\Omega,a_I}$ be a $\wmso$-automaton. Let $c \in C$ be a label and $Q \in \shA$ a binary relation on $A$. By Corollary \ref{cor:olquepositivenf}, for some $\Pi,\Sigma \subseteq \shA$ and $T_i \subseteq A \times A$, there is a sentence $\Psi_{Q,c} \in {\olque}^+(A\times A)$ in the basic form $\bigvee \posdbnfolque{\vlist{T}}{\Pi}{\Sigma}$ such that
\begin{eqnarray*}
  \bigwedge_{a \in \Ran(Q)} \Delta^{\star}(a,c) &\equiv& \Psi_{Q,c}.
\end{eqnarray*}
By definition $\Psi_{Q,c}$ is of the form $\bigvee_{i}\varphi_i$, with each $\phi_{i}$ of shape $\posdbnfolque{\vlist{T}}{\Pi}{\Sigma}$. We put $\shDe(Q,c) := \bigvee_{i}\varphi_i^{\f}$, where the translation $(\cdot)^{\f}$ is given as in definition \ref{DEF_finitary_lifting}. Observe that $\shDe(Q,c)$ is of type ${\olque}^+(A \cup \shA)$.
\end{definition}

We have now all the ingredients to define our two-sorted automaton.

\begin{definition}\label{def:finitaryconstruct}
Let $\aut = \tup{A,\Delta,\Omega,a_I}$ be a {\wmso-automaton}. We define the \emph{finitary construct over $\mb{A}$} as the automaton $\aut^{\f} = \tup{A^{\f},\Delta^{\f},\Omega^{\f},a_I^{\f}}$ given by
\begin{eqnarray*}
        A^{\f} &:=& A \cup \shA \\
        a_I^{\f} &:=& \{(a_I,a_I)\}\\
        \Delta^{\f}(a,c) &:=& \Delta(a,c)\\
        \Delta^{\f}(R,c) &:=& \shDe(R,c) \vee \bigwedge_{a \in \Ran(R)} \Delta(a,c)\\
        \Omega^{\f}(a) &:=& \Omega(a)\\
        \Omega^{\f}(R) &:=& 1.
      \end{eqnarray*}
\end{definition}

The underlying idea of Definition \ref{def:finitaryconstruct} is the same of the \emph{two-sorted construction} (\emph{cf.} \cite[Def.~3.7]{Zanasi:Thesis:2012}, \cite{DBLP:conf/lics/FacchiniVZ13}) for weak $\mso$-automata. In both cases we want that macro-states process just a \emph{well-founded} portion of any accepted tree: this is guaranteed by associating all macro-states with the odd parity value $1$. However, for the finitary construction we aim at the stronger condition that such a portion is \emph{finite}. To achieve this, the key difference with the two-sorted construction is in the use of the translation $(\cdot)^{\f}$ to define $\shDe$: as $\Delta$ may be specified using quantifiers $\qu$ and $\dqu$, it serves the purpose of tracking the cardinality constraints of the original $\wmso$-automaton and ensure that they are not lifted to constraints on macro-states.

The next proposition establishes the desired properties of the finitary
construct. To this aim, we first introduce the notions of functional and finitary strategy.

\begin{definition}\label{def:StratfunctionalFinitary}
Given a $\wmso$-automaton $\bbA = \tup{A,\tmap,\pmap,a_I}$ and transition system $\bbT$, a strategy $f$ for \eloise in $\mathcal{A}(\bbA,\model)$ is \emph{functional in $B \subseteq A$} (or simply functional, if $B=A$) if for each node $s$ in $\bbT$ there is at most one $b \in B$ such that $(b,s)$ is a reachable position in an $f$-guided match. Also $f$ is \emph{finitary} in $B$ if there are only finitely many nodes $s$ in $\bbT$ for which a position $(b,s)$ with $b \in B$ is reachable in an $f$-guided match.
\end{definition}


\begin{lemma}\label{PROP_facts_finConstr} Let $\mb{A}$ be a $\wmso$-automaton and $\mb{A}^{\f}$ its finitary construct over $\mb{A}$. The following holds:
\begin{enumerate}
  \itemsep 0 pt
  \item $\mb{A}^{\f}$ is a $\wmso$-automaton. \label{point:finConstrAut}
  \item For any $\mb{T}$, if $\exists$ has a winning strategy in
  the game $\mathcal{A}(\aut^{\f},\model)@(a_I^{\f},s_I)$, then she has a winning strategy in the same game which is both functional and finitary in $\shA$. \label{point:finConstrStrategy}
  \item $\mb{A} \equiv \mb{A}^{\f}$. \label{point:finConstrEquiv}
  \end{enumerate}
\end{lemma}
\begin{proof} 
We address each point separately.
\begin{enumerate}
  \item We need to show that $\mb{A}^{\f}$ is weak and respects the continuity condition. For this purpose, we fix the following observation:
      \begin{itemize}
      \item[($\star$)] by definition of $\Delta^{\f}$, for any macro-state $R \in \shA$ and state $a \in A$, it is never the case that $a \preceq R$.
      \end{itemize}
      This means that, when considering a strongly connected component of $\mb{A}^{\f}$, we may assume that all states involved are either from $A$ or from $\shA$.

      In order to prove our claim, let $q_1,q_2 \in A^{\f}$ be two states of $\mb{A}^{\f}$ such that $q_1 \preceq q_2$ and $q_2 \preceq q_1$. By observation ($\star$), we can distinguish the following two cases:
      \begin{enumerate}[(\roman*)]
        \item if $q_1$ and $q_2$ are states from $A$, then for $i \in \{1,2\}$ the value of $\Omega^{\f}(q_i)$ and $\Delta^{\f}(q_i,c)$ is defined respectively as $\Omega(q_i)$ and $\Delta(q_i,c)$ like in the $\wmso$-automaton $\mb{A}$. It follows that they satisfy both the continuity and weakness condition.
        \item Otherwise, $q_1$ and $q_2$ are macro-states in $\shA$. For the weakness condition, observe that all macro-states in $\aut^{\f}$ have the same parity value. For the continuity condition, suppose that $q_1$ occurs in $\Delta^{\f}(q_2,c)$ for some $c \in C$. By definition $q_1$ can only appear in the disjunct $\shDe(q_2,c) = \bigvee_{i}\varphi_i^{\f}$ of $\Delta^{\f}(q_2,c)$. By Lemma \ref{LEM_cont}, we know that each $\varphi_i^{\f}$ is continuous in $\shA$. Then in particular $\Delta^{\f}(q_2,c)$ is continuous in $q_1$. By definition $\Omega^{\f}(q_1) =1$ is odd, meaning that the continuity condition holds. The case in which $q_2$ appears in $\Delta^{\f}(q_1,c)$ is just symmetric.
      \end{enumerate}
  \item  Let $f$ be a winning strategy for $\exists$ in $\mathcal{A}(\mb{A}^{\f},\model)@(a_I^{\f},s_I)$. We define a strategy $f'$ for $\exists$ in the same game as follows:
      \begin{enumerate}[label=(\alph*),ref=\alph*]
        \item on basic positions of the form $(a,s) \in A\times T$, let $\val$ be the valuation suggested by $f$. We let the valuation suggested by $f'$ be the restriction $\val'$ of $\val$ to $A$. Observe that, as no predicate from $A^{\f}\setminus A =\shA$ occurs in $\Delta^{\f}(a,\V(s)) = \Delta(a,\V(s))$, then $\val'$ also makes that sentence true in $\R{s}$.
        \label{point:stat2point1}
        \item for basic positions of the form $(R,s) \in \shA \times T$, let $\val_{R,s}$ be the valuation suggested by $f$. As $f$ is winning, $\Delta^{\f}(R,\V(s))$ is true in the model $\val_{R,s}$. If this is because the disjunct $\bigwedge_{a \in \Ran(R)} \Delta(a,\V(s))$ is made true, then we can let $f'$ suggest the restriction to $A$ of $\val_{R,s}$, for the same reason as in \eqref{point:stat2point1}. Otherwise, the disjunct $\shDe(R,\V(s)) = \bigvee_{i}\varphi_i^{\f}$ is made true. This means that, for some $i$,
             $$(R[s], \val_{R,s}) \models \varphi_i^{\f}.$$
             By Lemma \ref{LEM_cont} $\varphi_i^{\f}$ is functionally continuous in $\shA$, meaning that we have a restriction $\val_{R,s}'$ of $\val_{R,s}$ that verifies $\varphi_i^{\f}$, assigns finitely many nodes to predicates from $\shA$ and associates with each node at most one predicate from $\shA$. We let $\val_{R,s}'$ be the suggestion of $f'$ from position $(R,s)$.
      \end{enumerate}
      The strategy $f'$ defined as above is immediately seen to be
      surviving for $\exists$. It is also winning, because the set of
      basic positions on which $f'$ is defined is a subset of the one
      of the winning strategy $f$. By this observation it also follows that any $f'$-conform match visits basic positions of the form $(R,s) \in \shA \times C$ only finitely many times, as those have odd parity. By definition, the valuation suggested by $f'$ only assigns finitely many nodes to predicates in $\shA$ from positions of that shape, and no nodes from other positions. It follows that $f'$ is finitary in $\shA$. Functionality in $\shA$ also follows immediately by definition of $f'$.
  \item The proof is entirely analogous to the one presented in \cite[Prop. 3.9]{Zanasi:Thesis:2012} for the two-sorted construction. For the direction from left to right, it is immediate by definition of $\mb{A}^{\f}$ that a winning strategy for $\exists$ in $\mc{G} = \mathcal{A}(\aut,\model)@(a_I,s_I)$ is also winning for $\exists$ in $\mc{G}^{\f} = \mathcal{A}(\mb{A}^{\f},\model)@(a_I^{\f},s_I)$.

      For the direction from right to left, let $f$ be a winning strategy for $\exists$ in $\mc{G}^{\f}$. The idea is to define a strategy $f'$ for $\exists$ in stages, while playing a match $\pi'$ in $\mc{G}$. In parallel to $\pi'$, a shadow match $\pi$ in $\mc{G}^{\f}$ is maintained, where $\exists$ plays according to the strategy $f$. For each round $z_i$, we want to keep the following relation between the two matches:
\smallskip
\begin{center}
\fbox{\parbox{12cm}{
Either
\begin{enumerate}[label=(\arabic*),ref=\arabic*]
  \item basic positions of the form $(Q,s) \in \shA \times T$ and $(a,s) \in A \times T$ occur respectively in $\pi$ and $\pi'$, with $a \in \Ran(Q)$,
\end{enumerate}
or
\begin{enumerate}[label=(\arabic*),ref=\arabic*]
  \item[(2)] the same basic position of the form $(a,s) \in A \times T$ occurs in both matches.
\end{enumerate}
}}\hspace*{0.3cm}($\ddag$)
\end{center}
\smallskip
The key observation is that, because $f$ is winning, a basic position of the form $(Q,s) \in \shA \times T$ can occur only for finitely many initial rounds $z_0,\dots,z_n$ that are played in $\pi$, whereas for all successive rounds $z_n,z_{n+1},\dots$ only basic positions of the form $(a,s) \in A \times T$ are encountered. Indeed, if this was not the case then either $\exists$ would get stuck or the minimum parity occurring infinitely often would be odd, since states from $\shA$ have parity $1$.

It follows that enforcing a relation between the two matches as in ($\ddag$) suffices to prove that the defined strategy $f'$ is winning for $\exists$ in $\pi'$. For this purpose, first observe that $(\ddag).1$ holds at the initial round, where the positions visited in $\pi'$ and $\pi$ are respectively $(a_I,s_I) \in A \times T$ and $(\{(a_I,a_I)\},s_I) \in A^{\f} \times T$. Inductively, consider any round $z_i$ that is played in $\pi'$ and $\pi$, respectively with basic positions $(a,s) \in A \times T$ and $(q,s) \in A^{\f} \times T$. In order to define the suggestion of $f'$ in $\pi'$, we distinguish two cases.
\begin{itemize}
  \item First suppose that $(q,s)$ is of the form $(Q,s) \in
  \shA\times T$. By ($\ddag$) we can assume that $a$ is in $\Ran(Q)$. Let $\val_{Q,s} :A^{\f} \rightarrow \p(\R{s})$ be the valuation suggested by $f$, verifying the sentence $\Delta^{\f}(Q,\V(s))$. We distinguish two further cases, depending on which disjunct of $\Delta^{\f}(Q,\V(s))$ is made true by $\val_{Q,s}$.
      \begin{enumerate}[label=(\roman*), ref=\roman*]
        \item If $(\R{s},\val_{Q,s})\models \bigwedge_{b \in \Ran(Q)} \Delta(b,\V(s))$, then we let $\exists$ pick the restriction to $A$ of the valuation $\val_{Q,s}$. \label{point:valuation1}
        \item If $(\R{s},\val_{Q,s})\models \shDe(Q,\V(s))$, we let $\exists$ pick a valuation $\val_{a,s}:A \rightarrow \p (\R{s})$ defined by putting, for each $b \in A$:
            \begin{align*}
               \val_{a,s}(b)\ :=\ \bigcup_{b \in \Ran(Q')} &\{t \in \R{s} \mid t \in \val_{Q,s}(Q')\} \\
               \cup\ \ \ \ \ & \{t \in \R{s} \mid t \in \val_{Q,s}(b)\} .
            \end{align*} \label{point:valuation2}
      \end{enumerate}
      It can be readily checked that the suggested move is admissible for $\exists$ in $\pi$, i.e. it makes $\Delta(a,\V(s))$ true in $\R{s}$. For case \eqref{point:valuation2}, one has to observe how $\shDe$ is defined in terms of $\Delta$. In particular, the nodes assigned to $b$ by $\val_{Q,s}$ have to be assigned to $b$ also by $\val_{a,s}$, as they may be necessary to fulfill the condition, expressed with $\qu$ and $\dqu$, that infinitely many nodes witness (or that finitely many nodes do not witness) some type.

      We now show that $(\ddag)$ holds at round $z_{i+1}$. If \eqref{point:valuation1} is the case, any next position $(b,t)\in A \times T$ picked by player $\forall$ in $\pi'$ is also available for $\forall$ in $\pi$, and we end up in case $(\ddag .2)$. Suppose instead that \eqref{point:valuation2} is the case. Given the choice $(b,t) \in A \times T$ of $\forall$, by definition of $\val_{a,s}$ there are two possibilities. First, $(b,t)$ is also an available choice for $\forall$ in $\pi$, and we end up in case $(\ddag .2)$ as before. Otherwise, there is some $Q' \in \shA$ such that $b$ is in $\Ran(Q')$ and $\forall$ can choose $(Q',t)$ in the shadow match $\pi$. By letting $\pi$ advance at round $z_{i+1}$ with such a move, we are able to maintain $(\ddag .1)$ also in $z_{i+1}$.
  \item In the remaining case, inductively we are given the same basic position $(a,s) \in A\times T$ both in $\pi$ and in $\pi'$. The valuation $\val$ suggested by $f$ in $\pi$ verifies $\Delta^{\f}(a,\V(s)) = \Delta(a,\V(s))$, thus we can let the restriction of $\val$ to $A$ be the valuation chosen by $\exists$ in the match $\pi'$. It is immediate that any next move of $\forall$ in $\pi'$ can be mirrored by the same move in $\pi$, meaning that we are able to maintain the same position --whence the relation $(\ddag.1)$-- also in the next round.
\end{itemize}
In both cases, the suggestion of strategy $f'$ was a legitimate move for $\exists$ maintaining the relation $(\ddag)$ between the two matches for any next round $z_{i+1}$. It follows that $f'$ is a winning strategy for $\exists$ in $\mc{G}$.
\end{enumerate}
\end{proof}


\begin{remark}
While the finitary construction is a variant of the two-sorted one given
in~\cite{DBLP:conf/lics/FacchiniVZ13}, it is
worth noticing that the latter would have not been suitable for our purposes.
Indeed, suppose to define the two-sorted construct $\mb{A}^{2S}$ over a
$\wmso$-automaton $\mb{A}$, analogously to the case of weak $\MSO$-automata.
Then $\mb{A}^{2S}$ will generally \emph{not} be a $\wmso$-automaton.
The problem lies in the \textbf{(continuity)} condition: since all macro-states in
$\mb{A}^{2S}$ have parity $1$, whenever two of them, say $R$ and $Q$, are such that $R \preceq Q$ and $Q \preceq R$, then the sentence $\Delta^{2S}(R,c)$ should be continuous in $Q$. But this is not necessarily the case, since the truth of $\Delta^{2S}(R,c)$ may depend upon the truth of a subformula of the form $\exists^{\infty}x.Q(x)$, requiring $Q$ to be interpreted over infinitely many nodes. (This problem is overcome in the finitary construction by using the translation $(\cdot)^{\f}$ to define $\Delta^{\f}$.)

As a consequence, we cannot use the two-sorted construction to show that
$\wmso$-automata are closed under noetherian projection
(\cite[Def. 3]{DBLP:conf/lics/FacchiniVZ13}).
This observation is coherent with the fact that $\yvWFMSO$ is \emph{not} a fragment of $\yvWMSO$. Similarly, the simulation theorem for $\mso$-automata \cite{Walukiewicz96} preserves neither the \textbf{(weakness)} nor the \textbf{(continuity)} condition and thus it cannot show
closure under (arbitrary) projection for $\wmso$-automata.\end{remark} 

%% file: projection.tex

\subsubsection{Closure under Finitary Projection}

\begin{definition}\label{DEF_fin_projection}
Let $\aut = \tup{A, \Delta, \Omega, a_I}$ be a $\wmso$-automaton on alphabet $\p(\prop \cup \{p\})$. Let $\aut^{\f}$
denote its finitary construct.
We define the automaton ${{\exists}_F p}.\mb{A} = \langle A^{\f}, a_I^{\f},
\DeltaProj, \Omega^{\f}\rangle$ on alphabet $\p\prop$ by putting
\begin{eqnarray*}
  \DeltaProj(a,c) &:=& \Delta^{\f}(a,c)\\
  \DeltaProj(R,c) &:=& \Delta^{\f}(R,c) \vee \Delta^{\f}(R,c\cup\{p\}).
\end{eqnarray*}
The automaton ${{\exists}_F p}.\mb{A}$ is called the \emph{finitary projection
construct of $\mb{A}$ over $p$}.
\end{definition}

Our projection construction corresponds to a suitable closure operation on tree languages, modeling the semantics of $\wmso$ existential quantification.

\begin{definition}\label{def:tree_finproj} Let $p$ be a propositional letter and $L$ a tree language of $\p (\prop\cup\{p\})$-labeled trees. The \emph{finite projection} of $L$ over $p$ is the language ${\exists}_F p.L$ of $C$-labeled trees defined as
\begin{equation*}
    {\exists}_F p.L = \{\model \mid \text{there is a $p$-variant } \model[p\mapsto S] \text{ of } \model \text{ such that } \model[p\mapsto S] \in L \text{ and } S \text{ is finite} \}.
\end{equation*}\hfill
\end{definition}

\begin{lemma}\label{PROP_fin_projection}
For each $\wmso$-automaton $\aut$ on alphabet $\p (\prop \cup \{p\})$,
we have that
$$\trees({{\exists}_F p}.\mb{A}) \ \equiv\
{{\exists}_F p}.\trees(\mb{A}).
$$
\end{lemma}

\begin{proof}
What we need to show is that for any tree $\model$:
\begin{eqnarray*}
  {{\exists}_F p}.\mb{A} \text{ accepts } \mathbb{T} & \text{ iff }& \text{there is a finite $p$-variant }\model' \\
   & & \text{of }\mathbb{T}\text{  such that }\aut\text{  accepts }\model'.
\end{eqnarray*}
For direction from left to right, we first observe that the properties stated by Lemma~\ref{PROP_facts_finConstr} hold for ${{\exists}_F p}.\mb{A}$ as well, since the latter is defined in terms of $\mb{A}^{\f}$. Then we can assume that the given winning strategy $f$ for $\exists$ in $\mc{G_{\exists}} = \mc{A}({{\exists}_F p}.\mb{A},\model)@(a_I^{\f},s_I)$ is functional and finitary in $\shA$. Functionality allows us to associate with each node $s$ either none or a unique state $Q_s \in \shA$ (\emph{cf.} \cite[Prop. 3.12]{Zanasi:Thesis:2012}). We now want to isolate  the nodes that $f$ treats ``as if they were labeled with $p$''. For this purpose, let $\val_{s}$ be the valuation suggested by $f$ from a position $(Q_s,s) \in \shA \times T$. As $f$ is winning, $\val_{s}$ makes $\DeltaProj(Q,\tscolors(s))$ true in $\R{s}$. We define a $p$-variant $\model'$ of $\model$ by coloring with $p$ all nodes in the following set:
 \begin{equation}\label{eq:X_p}
   X_p\ :=\ \{s \in T\mid (\R{s},\widetilde{\val}_{s}) \models \Delta^{\f}(Q_s,\tscolors(s)\cup\{p\})\}.
\end{equation}
The fact that the strategy of $\exists$ is finitary in $\shA$ guarantees that $X_p$ is finite, whence $\model'$ is a finite $p$-variant. The argument showing that $\mb{A}^{\f}$ (and thus also $\mb{A}$, by Lemma \ref{PROP_facts_finConstr}(1)) accepts $\model'$ is a routine adaptation of the analogous proof for the noetherian projection of weak $\mso$-automata, for which we refer to \cite[Prop. 3.12]{Zanasi:Thesis:2012}.
\medskip

For the direction from right to left, let $\model'$ be a finite $p$-variant of
$\model$, with labeling function $\tscolors'$, and $g$ a winning strategy for $\exists$ in $\mc{G} = \mathcal{A}(\aut,\model')@(a_I,s_I)$. Our goal is to define a strategy $g'$ for $\exists$ in $\mc{G_{\exists}}$. As usual, $g'$ will be constructed in stages, while playing a match $\pi'$ in $\mc{G_{\exists}}$. In parallel to $\pi'$, a \emph{bundle} $\mc{B}$ of $g$-guided shadow matches in $\mc{G}$ is maintained, with the following condition enforced for each round $z_i$ (\emph{cf.} \cite[Prop.~ 3.12]{Zanasi:Thesis:2012}) :
\smallskip
\begin{center}
\fbox{\parbox{13cm}{
\begin{enumerate}
  \item If the current (i.e. at round $z_i$) basic position in $\pi'$ is of the form $(Q,s) \in \shA \times T$, then for each $a \in\Ran(Q)$ there is an $g$-guided (partial) shadow match $\pi_a$ at basic position $(a,s) \in A\times T$ in the current bundle $\mc{B}_i$. Also, either $\model'_s$ is not $p$-free (i.e., it does contain a node $s'$ with $p \in \tscolors'(s')$) or $s$ has some sibling $t$ such that $\model'_t$ is not $p$-free.
  \item Otherwise, the current basic position in $\pi'$ is of the form $(a,s) \in A \times T$ and $\model'_s$ is $p$-free (i.e., it does not contain any node $s'$ with $p \in \tscolors'(s')$). Also, the bundle $\mc{B}_i$ only consists of a single $g$-guided match $\pi_a$ whose current basic position is also $(a,s)$.
\end{enumerate}
}}\hspace*{0.3cm}($\ddag$)
\end{center}
\smallskip
We briefly recall the idea behind condition ($\ddag$). Point ($\ddag.1$) describes the part of match $\pi'$ where it is still possible to encounter nodes which are labeled with $p$ in $\model'$. As $\DeltaProj$ only takes the letter $p$ into account when defined on macro-states in $\shA$, we want $\pi'$ to visit only positions of the form $(R,s) \in \shA \times T$ in that situation. Anytime we visit such a position $(R,s)$ in $\pi'$, the role of the bundle is to provide one $g$-guided shadow match at position $(a,s)$ for each $a \in \Ran(R)$.
Then $g'$ is defined in terms of what $g$ suggests from those positions.

 Point ($\ddag.2$) describes how we want the match $\pi'$ to be
 played on a $p$-free subtree: as any node that one might encounter has the same label in $\model$ and $\model'$,
it is safe to let ${{\exists}_F p}.\mb{A}$ behave as $\aut$ in such situation. Provided that the two matches visit the same basic positions, of the form $(a,s)\times A \times T$, we can let $g'$ just copy $g$.

The key observation is that, as $\model'$ is a \emph{finite} $p$-variant of $\model$, nodes labeled with $p$ are reachable only for finitely many rounds of $\pi'$. This means that, provided that ($\ddag$) hold at each round, ($\ddag.1$) will describe an initial segment of $\pi'$, whereas ($\ddag.2$) will describe the remaining part. Thus our proof that $g'$ is a winning strategy for $\exists$ in $\mc{G}_{\exists}$ is concluded by showing that ($\ddag$) holds for each stage of construction of $\pi'$ and $\mc{B}$.

\medskip

For this purpose, we initialize $\pi'$ from position $(\shai,s) \in \shA\times T$ and the bundle $\mc{B}$ as $\mc{B}_0 = \{\pi_{a_I}\}$, with $\pi_{a_I}$ the partial $g$-guided match consisting only of the position $(a_I,s)\in A\times T$. The situation described by ($\ddag .1$) holds at the initial stage of the construction.
Inductively, suppose that at round $z_i$ we are given a position $(q,s) \in A^{\f} \times T$ in $\pi^{\f}$ and a bundle $\mc{B}_i$ as in ($\ddag$). To show that ($\ddag$) can be maintained at round $z_{i+1}$, we distinguish two cases, corresponding respectively to situation ($\ddag.1$) and ($\ddag.2$) holding at round $z_i$.
\begin{enumerate}[label = (\Alph*), ref = \Alph*]
  \item If $(q,s)$ is of the form $(Q,s) \in \shA \times T$, by inductive hypothesis we are given with $g$-guided shadow matches $\{\pi_a\}_{a \in \Ran(Q)}$ in $\mc{B}_i$. For each match $\pi_a$ in the bundle, we are provided with a valuation $\val_{a,s}: A \rightarrow \p (\R{s})$ making $\Delta(a,\tscolors'(s))$ true. Then we further distinguish the following two cases.
\begin{enumerate}[label = (\roman*), ref = \roman*]
  \item \label{point:TsNotPFree} Suppose first that $\model'_s$ is not $p$-free. We let the suggestion $\val' \: A^{\f} \to \p (\R{s})$ of $g'$ from position $(Q,s)$ be defined as follows:
       \begin{align*}
       \val'(q')\ :=\ \begin{cases}
               \bigcap\limits_{\substack{(a,b) \in q',\\ a \in \Ran(Q)}}\{t\ \in \R{s} \mid t \in \val_{a,s}(b)\}               & q' \in \shA \\[2em]
               \bigcup\limits_{a \in \Ran(Q)} \{t\ \in \R{s} \mid t \in \val_{a,s}(q') \text{ and }\model'.t\text{ is $p$-free}\}              & q' \in A.
           \end{cases}
       \end{align*}
       The definition of $\val'$ on $q' \in \shA$ is standard (\emph{cf.}~\cite[Prop. 2.21]{Zanasi:Thesis:2012}) and guarantees a correspondence between the states assigned by the markings $\{\val_{a,s}\}_{a \in \Ran(Q)}$ and the macro-states assigned by $\val'$. The definition of $\val'$ on $q' \in A$ aims at fulfilling the conditions, expressed via $\qu$ and $\dqu$, on the number of nodes in $\R{s}$ witnessing (or not) some $A$-types. Those conditions are the ones that $\shDe(Q,\tscolors'(s))$ --and thus also $\Delta^{\f}(Q,\tscolors'(s))$-- ``inherits'' by $\bigwedge_{a \in \Ran(R)} \Delta(a,\tscolors'(s))$, by definition of $\shDe$. Notice that we restrict $\val'(q')$ to the nodes $t \in \val_{a,s}(q')$ such that $\model'.t$ is $p$-free. As $\model'$ is a \emph{finite} $p$-variant, only \emph{finitely many} nodes in $\val_{a,s}(q')$ will not have this property. Therefore their exclusion, which is crucial for maintaining condition ($\ddag$) (\emph{cf.}~case \eqref{point:ddag2CardfromMacro} below), does not influence the fulfilling of the cardinality conditions expressed via $\qu$ and $\dqu$ in $\shDe(Q,\tscolors'(s))$.

       On the base of these observations, one can check that $\val'$ makes $\shDe(Q,\tscolors'(s))$--and thus also $\Delta^{\f}(Q,\tscolors'(s))$--true in $\R{s}$. In fact, to be a legitimate move for $\exists$ in $\pi'$, $\val'$ should make $\DeltaProj(Q,\tscolors(s))$ true: this is the case, for $\Delta^{\f}(Q,\tscolors'(s))$ is either equal to $\Delta^{\f}(Q,\tscolors(s))$, if $p \not\in \tscolors'(s)$, or to $\Delta^{\f}(Q,\tscolors(s)\cup\{p\})$ otherwise. In order to check that we can maintain $(\ddag)$, let $(q',t) \in A^{\f} \times T$ be any next position picked by $\forall$ in $\pi'$ at round $z_{i+1}$. As before, we distinguish two cases:
       \begin{enumerate}[label = (\alph*), ref = \alph*]
         \item If $q'$ is in $A$, then, by definition of $\val'$, $\forall$ can choose $(q',t)$ in some shadow match $\pi_a$ in the bundle $\mc{B}_i$. We dismiss the bundle --i.e. make it a singleton-- and bring only $\pi_a$ to the next round in the same position $(q',t)$. Observe that, by definition of $\val'$, $\model'.t$ is $p$-free and thus ($\ddag.2$) holds at round $z_{i+1}$. \label{point:ddag2CardfromMacro}
         \item Otherwise, $q'$ is in $\shA$. The new bundle $\mc{B}_{i+1}$ is given in terms of the bundle $\mc{B}_i$: for each $\pi_a \in \mc{B}_i$ with $a\in \Ran(Q)$, we look if for some $b \in \Ran(q')$ the position $(b,t)$ is a legitimate move for $\forall$ at round $z_{i+1}$; if so, then we bring $\pi_a$ to round $z_{i+1}$ at position $(b,t)$ and put the resulting (partial) shadow match $\pi_b$ in $\mc{B}_{i+1}$. Observe that, if $\forall$ is able to pick such position $(q',t)$ in $\pi'$, then by definition of $\val'$ the new bundle $\mc{B}_{i+1}$ is non-empty and consists of an $g$-guided (partial) shadow match $\pi_b$ for each $b \in \Ran(q')$. In this way we are able to keep condition ($\ddag.1$) at round $z_{i+1}$.
       \end{enumerate}
    \item Let us now consider the case in which $\model'_s$ is $p$-free. We let $g'$ suggest the valuation $\val'$ that assigns to each node $t \in \R{s}$ all states in $\bigcup_{a \in \Ran(Q)}\{b \in A\ |\ t \in \val_{a,s}(b)\}$. It can be checked that $\val'$ makes $\bigwedge_{a \in \Ran(Q)} \Delta(a,\tscolors'(s))$ -- and then also $\Delta^{\f}(Q,\tscolors'(s))$ -- true in $\R{s}$. As $p \not\in \tscolors(s)=\tscolors'(s)$, it follows that $\val'$ also makes $\DeltaProj(Q,\tscolors(s))$ true, whence it is a legitimate choice for $\exists$ in $\pi'$. Any next basic position picked by $\forall$ in $\pi'$ is of the form $(b,t) \in A \times T$, and thus condition ($\ddag.2$) holds at round $z_{i+1}$ as shown in (i.a). 
  \end{enumerate}
  \item In the remaining case, $(q,s)$ is of the form $(a,s) \in A \times T$ and by inductive hypothesis we are given with a bundle $\mc{B}_i$ consisting of a single $f$-guided (partial) shadow match $\pi_a$ at the same position $(a,s)$. Let $\val_{a,s}$ be the suggestion of $\exists$ from position $(a,s)$ in $\pi_a$. Since by assumption $s$ is $p$-free, we have that $\tscolors'(s) = \tscolors(s)$, meaning that $\DeltaProj(a,\tscolors(s))$ is just $\Delta(a,\tscolors(s)) = \Delta(a,\tscolors'(s))$. Thus the restriction $\val'$ of $\val$ to $A$ makes $\Delta(a,\tscolors'(t))$ true and we let it be the choice for $\exists$ in $\tilde{\pi}$. It follows that any next move made by $\forall$ in $\tilde{\pi}$ can be mirrored by $\forall$ in the shadow match $\pi_a$.
\end{enumerate}

\end{proof}

%% file: complementation.tex

\subsubsection{Closure under Boolean operations}

In this section we will show that the class of $\wmso$-automaton recognizable
tree languages is closed under the Boolean operations.
Start with closure under union, we just mention the following result, without
providing the (completely routine) proof.

\begin{theorem}
\label{t:cl-dis}
Let $\bbA_{0}$ and $\bbA_{1}$ be $\yvWMSO$-automata. 
Then there is a $\yvWMSO$-automaton $\bbA$ such that $\trees(\bbA)$ is the 
union of $\trees(\bbA_{0})$ and $\trees(\bbA_{1})$.
\end{theorem}

In order to prove closure under complementation, we crucially use that the 
one-step language $\olque$ is closed under Boolean duals 
(cf.~Proposition~\ref{prop:duals}).

\begin{theorem}
\label{t:cl-cmp}
Let $\bbA$ be an $\yvWMSO$-automaton.
Then the automaton $\overline{\aut}$ defined in Definition~\ref{d:caut} is a
$\yvWMSO$-automaton recognizing the complement of $\trees(\bbA)$.
\end{theorem}

\begin{proof}
Since we already know that $\overline{\bbA}$ accepts exactly the transition
systems that are rejected by $\bbA$, we only need to check that 
$\overline{\bbA}$ is indeed a $\yvWMSO$-automaton.
But this is straightforward: for instance, continuity can be checked by 
observing the self-dual nature of this property.
\end{proof}

%% file: aut-to-formula_wmso.tex

In what follows, we verify that WMSO-automata capture exactly the expressive power of WMSO on the class of tree models. Since we already proved the direction from formulas into automata (Theorem~\ref{t:wmsoauto}), we just have to verify that there is a sound translation going in the other direction.
For this purpose, we first introduce a fixpoint extension of first-order logic.

\subsection{Fixpoint extension of  first-order logic}

Let our first-order signature be composed of a set $\prop$ of monadic predicates (denoted with capital latin letters) and an unique binary predicate $R$.
%
%
%
%
Analogously to the modal $\mu$-calculus, the fixpoint extension of $\lque(\prop)$ is defined by adding a fixpoint construction clause.

\begin{definition}
The fixed point logic $\mlque(\prop)$ is given by:
$$
\varphi ::= q(x) \mid R(x,y) \mid x \foeq y \mid \exists x.\varphi \mid \qu x.\varphi \mid \lnot\varphi \mid \varphi \land \varphi \mid \mu p.\varphi(p,x)
$$
where $p,q\in\prop$, $x,y\in\fovar$; moreover $p$ occurs only positively in $\varphi(p,x)$ and $x$ is the only free variable in $\varphi(p,x)$.
\end{definition}



The semantics of the fixpoint formula $\mu p. \phi(p, x)$is the expected one. Given a model $\model$ and $s \in T$,  $\model \models \mu p. \phi(p, s)$ iff $s$ is in the least fixpoint of the  operator $F_\phi:\wp(T)\to \wp(T)$ defined as $F_\phi(S) := \{t \in T \mid \model[p \mapsto S] \models \phi(p, t) \}$.

Formulas of $\mlque$ may be also classified according to their alternation depth as it happens for the modal $\mu$-calculus.
The alternation-free fragment of $\mlque$ is thence defined as the collection of $\mlque$-formulas $\phi$
without nesting of greatest and least fixpoint operators, i.e. such that, for any two subformulas $\mu p.\psi_1(p,y)$ and $\nu q. \psi_2(q,z)$, predicates $p$ and $q$ do not occur free respectively in $\psi_2(q,z)$ and $\psi_1(p,y)$.

\begin{definition}
Given $p \in \prop$, we say that $\varphi \in \mlque(\prop)$ is
\begin{itemize}
\item \emph{monotone in the predicate $p$} iff for every LTS $\model$ and assignment $\ass$, \[ \text{if }\model, \ass \models \varphi \text{ and $\tsval(p) \subseteq E$, then }\model[p \mapsto E], g\models \phi\]

\item \emph{continuous in the predicate $p$} iff for every LTS $\model$ and assignment $\ass$ there exists some finite $S \subseteq_\omega \tsval(p)$ such that
$$
\model, \ass \models \varphi \quad\text{iff}\quad \model[p \mapsto S], \ass \models \varphi .
$$
\end{itemize}
\end{definition}

In the next definition, we provide a definition of the continuous fragment of $\mlque$, reminiscent of the one defined in Theorem~\ref{thm:olquecont}.
\begin{definition}
Let $\mathsf{Q}\subseteq \prop$ be a set of monadic predicates. The fragment $\cont{\mlque}{\mathsf{Q}}(\prop)$ is defined by the following rules:
$$
\varphi ::= \psi \mid q(x) \mid \exists x.\varphi(x) \mid \varphi \land \varphi \mid \varphi \lor \varphi \mid \wqu x.(\varphi,\psi) \mid \mu p. \phi'(p, x)
$$
where $q \in \mathsf{Q}$, $\psi \in \mlque(\prop\setminus \mathsf{Q})$, $p \in \prop \setminus \mathsf{Q}$, $\wqu x.(\varphi,\psi) := \forall x.(\varphi(x) \lor \psi(x)) \land \dqu x.\psi(x)$ and $\phi'(p,x)$ is a formula with only $x$ free such that $\phi'(p,x) \in \cont{\mlque}{\mathsf{Q} \cup\{p\}}(\prop)$.

\end{definition}


\begin{lemma}\label{lem:colqueiscont_mu}
If $\varphi \in \cont{\mlque}{\mathsf{Q}}(\prop)$ then $\varphi$ is continuous in (each predicate from) $\mathsf{Q}$.
\end{lemma}
\begin{proof} First, notice that If $\varphi \in \cont{\mlque}{\mathsf{Q}}(\prop)$ then $\varphi$ is monotone  in (each predicate from) $\mathsf{Q}$. 
The proof goes then by induction on the complexity of $\varphi$. For the all the cases except the fixpoint one, the proof is the same as the one for Lemma~\ref{lem:colqueiscont}. For $\phi=\mu p. \phi'(p, x)$, with $\phi'(p,x) \in \cont{\mlque}{\mathsf{Q} \cup\{p\}}(\prop)$, the argument is the same as in~\cite[Lemma 1]{Fontaine08}.
\end{proof}


As for the modal $\mu$-calculus, we define the fragment $\clque$ of $\mlque$ as the one where the use of the least fixed point operator is restricted to the continuous fragment. 

\begin{definition}
The fragment $\clque(\prop)$ of $\mlque(\prop)$ is given by the following restriction of the fixpoint operator to the contiuous fragment:
{\small%
$$
\varphi ::= q(x) \mid R(x,y) \mid x \foeq y \mid \exists x.\varphi \mid \qu x.\varphi \mid \lnot\varphi \mid \varphi \land \varphi \mid \mu p.\varphi'(p,x)
$$}%
where $p,q\in\prop$, $x,y\in\fovar$; and $\varphi'(p,x) \in \cont{\mlque}{\{p\}}(\prop)$ is such that $p$ occurs only positively in $\varphi'$ and $x$ is the only free variable in $\varphi'$.
\end{definition}

We now recall a useful property of fixpoint and continuity. Let $\phi(p,x)$ a formula with only $x$ free.
Given a LTS $\model$, for every ordinal $\alpha$, we define by induction the following sets:
\begin{itemize}
	\itemsep 0 pt
	\item $\phi^0(\emptyset):= \emptyset$,
	\item $\phi^{\alpha+1}(\emptyset):= \{ s \in T \mid \model[p \mapsto \phi^\alpha(\emptyset)] \models \phi(p, s)\}$,
	\item $\phi^{\lambda}(\emptyset):= \bigcup_{\alpha < \lambda} \phi^{\alpha}(\emptyset)$, with $\lambda$ limit.
\end{itemize}
If $\phi$ is monotone in $p$, it is possible to show that $\phi^{\beta+1}(\emptyset)= \phi^{\beta}(\emptyset)$, for some ordinal $\beta$. Moreover, the set $\phi^{\beta}(\emptyset)$ is the least fixpoint of $F_\phi$ (cf. for instance \cite{ArnoldN01}).

A formula $\phi(p, x)$ is said to be \emph{constructive} in $p$ if its least fixpoint is reached in at most $\omega$ steps, i.e., if for every model $\model$, the least fixpoint of $F_\phi$ equals to $\bigcup_{\alpha < \omega} \phi^{\alpha}(\emptyset)$. From a local perspective, this means that a formula $\phi(p, x)$ constructive in $p$ if for every model $\model$,  every node $s \in T$, whenever $\mu p. \phi(p,x)$ is true at $s$, then $s$ belongs to some finite approximant $\phi^{i+1}(\emptyset)$ of the least fixpoint of $F_\phi$.
The next proposition is easily verified:

\begin{proposition}\label{prop:constructivity}
Let $\phi(p,x)$ be a $\mlque$-formula with only $x$ free. If $\phi(p,x)$ is continuous in $p$, then for every LTS $\model$, and every node $s \in T$, there is $i < \omega$ such that
\[\model \models \mu p. \phi(p,s) \text{ iff } s \in \phi^{i+1}(\emptyset).\]
\end{proposition}

From the fact that sets $\phi^{i+1}(\emptyset)$ are essentially defined as finite unfoldings and the previous Proposition~\ref{prop:constructivity}, we obtain the following.\fcwarning{More intuition on this?}

\begin{proposition}\label{prop:cor_constructivity}
Let $\phi(p,x)$ be a $\mlque$-formula with only $x$ free and such that $\phi(p,x)$ is continuous in $p$. Let $\model$ be a LTS, and $s \in T$. Then
$\model \models \mu p. \phi(p,s)$ iff there is a finite set $p^\model \subseteq_\omega T$ such that $s\in p^\model$ and $\model[p\mapsto p^\model] \models \phi(p,t)$  for every $t \in p^\model$.
\end{proposition}
 \begin{proof}
 For the direction from left to right, assume that $\model \models \mu p. \phi(p,s)$. By Proposition~\ref{prop:constructivity}, we know that  there is $i< \omega$ such that $\model[p \mapsto \phi^i(\emptyset)] \models \phi(p, s)$. The set $\phi^i(\emptyset)$ need not to be finite. However,
 using this information, we are going construct a finite tree whose nodes $t$ are labelled by finite sets $X^m_j$, where $m$ is a node of $\model$ and $j \leq i$, satisfying the following condition:
 \begin{enumerate}
\item  if $t$ is the root, then $t$ is labelled by $X_i^s$,
\item  if $t$ is labelled by $X_j^m=\{s_1, \dots, s_\ell\}$ and $j>0$, then $t$ has $\ell$  children and for every $s_i \in X_j^m$ there is an unique child $t'$ of $t$ labelled by $X_{j-1}^{n_i}$ where $m$ is a node,
\item for every node $t$ of the tree, if $t$ is labelled by $X_j^m$, then it holds that $X_j^m \subseteq \phi^{j}(\emptyset)$.
\end{enumerate}
If we verify that $\model[p\mapsto p^\model] \models \phi(p,s)$ holds by taking as $p^\model$ the union of all labels of the nodes of the constructed tree, we can conclude for the proof of this direction.

As starting point of the inductive construction, we start by the empty tree.  Recall that we know that  $\model[p \mapsto \phi^i(\emptyset)] \models \phi(p, s)$. Since $\phi(p,x)$ is continuous in $p$, there is a finite set $X^s_i \subseteq \phi^i(\emptyset)$ such that $\model[p \mapsto X^s_i] \models \phi(p, s)$. We then add a root to our tree and label it by $X^s_i$.
 Assume that at a leaf $s$ of our tree is labelled by $X^m_j$, for some $j < i$. If $X^m_j$ is empty, than we stop, else we proceed as follows. We know that $X^m_j\subseteq \phi^{j}(\emptyset)$. This means that $\model[p \mapsto \phi^{j-1}(\emptyset)] \models \phi(p, r)$, for every $r \in X_j^m$. By continuity, for each such $r$, there is a finite set $X^m_{j-1} \subseteq  \phi^{j-1}(\emptyset)$ such that $\model[p \mapsto X^m_{j-1}(\emptyset)] \models \phi(p, r)$. For each $r \in X^m_j$ we thus add a child to $m$ and label it with $X^r_{j-1}$. By definition of $\phi^{i+1}(\emptyset)$, the tree is finite. Let $X$ be the union of all labels of the constructed tree. $X$ is finite, and by monotonicity of $\phi(p,x)$ we have that for every $m \in X \cup \{s\}$, $\model[p \mapsto X \cup \{s\}] \models \phi(p,m)$.

For the other direction, it's enough to notice that the smallest finite set $p^\model \subseteq T$ such that $\model[p\mapsto p^\model] \models \phi(p,s)$ and $\model[p\mapsto p^\model] \models \phi(p,m)$ for all $m \in p^\model$ is the least fixpoint $F_\varphi$. 
 \end{proof}

\noindent Proposition~\ref{prop:cor_constructivity} naturally suggests the following translation $\mgFOETr{-}:\mlque(\prop)\to\wmso(\prop)$,

\begin{itemize}
	\itemsep 0 pt
	\item $\mgFOETr{p(x)}=p(x)$,
	\item $\mgFOETr{R(x,y)}=R(x,y)$
	\item $\mgFOETr{x\foeq y}= (x \foeq y)$
	\item $\mgFOETr{\varphi \land \psi}=\mgFOETr{\varphi} \land \mgFOETr{\psi}$,
	\item $\mgFOETr{\lnot \varphi}= \lnot \mgFOETr{\varphi}$,
	\item $\mgFOETr{\exists x. \varphi}= \exists x. \mgFOETr{\varphi}$,
	\item $\mgFOETr{\qu x. \varphi}= \forall p.\exists x. (\lnot p(x) \land \mgFOETr{\varphi})$,
	\item $\mgFOETr{\mu p. \varphi(p,x)}= \exists p ( p(x) \land \forall y ( p(y) \to \mgFOETr{\varphi(p,y) }))$.
\end{itemize}
%
The following theorem 
is then an immediate corollary of Proposition~\ref{prop:cor_constructivity}.

\begin{theorem}\label{thm:guard_wmso}
For every alternation-free formula $\phi$ in the $\clque$, every LTS $\model$, and assignment $\ass$, we have $\model, \ass \models \varphi$ iff $\model, \ass \models \mgFOETr{\varphi}$.
%
%
\end{theorem}
\begin{proof}
The proof goes by induction on the complexity of $\varphi$, the only critical step being the least fixpoint operator one. But this follows by applying Proposition \ref{prop:cor_constructivity} and the induction hypothesis.
\end{proof}


\subsection{Translating automata into formulas}
We are now ready to prove the second main result of the paper.

\begin{theorem}\label{thm:wmso_autofor}
There is an effective procedure that given an automaton in $\yvcwAut({\olque})$, returns an equivalent WMSO-formula.
\end{theorem}
\begin{proofsketch}
The argument   is
 essentially a refinement of the standard proof showing that any automaton in $\yvAut(\ofo)$ can be translated into an equivalent $\mu$-formula
$\xi_\aut$ (cf. e.g. \cite{Ven08}).
The idea is the following. We see a $\yvWMSO$-automaton as a system of equations expressed in terms of $\lque$-formulas: each state corresponds to a monadic predicate variable and the parity of a state corresponds to the least and greatest fixpoint that we seek for the associated variable, etc. One then solves this system of equations via the same inductive procedure used to obtain the formula of the modal $\mu$-calculus from the system associated with a  $\yvAut(\ofo)$-automaton (see e.g. \cite{ArnoldN01} for a description of the solution procedure). Because of the (weakness) and (continuity) conditions on the starting $\wmso$-automaton $\aut$, it is thence possible to verify that the resulting fixpoint formula $\xi_\aut$ belongs to $\clque$.
\end{proofsketch}

\begin{remark}
As a corollary of the automata characterization on trees of \wmso, we obtain the equivalence on this class of structures between \wmso and $\clque$. This consequence should be compared to the analogous result obtained by Walukiewicz in~\cite{Walukiewicz96} for FPL (fixpoint extension of $\foe$) and MSO on trees.
\end{remark}

%% file: jw_intro.tex

In this section we are going to prove Theorem \ref{t:m1}.
Our proof of the first item of the theorem crucially involves automata.
In the previous section we saw that on trees, $\yvWMSO$ effectively corresponds
to the automata class $\yvcwAut(\olque)$.
 We will now relate this class to the one of
parity automata based on $\ofo$ and satisfying similar weakness and continuity conditions.

\begin{definition}
A \emph{$\contAFMC$-automaton} $\aut = \tup{A,\Delta,\Omega,a_I}$ is an automaton $\aut \in \yvAut(\ofo)$ such that for all states $a,b \in A$ with $a \ord b$ and $b\ord a$ the following conditions hold:
\begin{description}
	\itemsep 0 pt
	\item[(weakness)] $\pmap(a)=\pmap(b)$,
	\item[(continuity)] if $\pmap(a)$ is odd (resp. even) then, for each $c\in C$ we have
	   $\tmap(a,c) \in \cont{\ofo^+}{b}(A)$ (resp. $\tmap(a,c) \in \cocont{\ofo^+}{b}(A)$).
\end{description}
As the class of such automata coincides with $\yvcwAut(\ofo)$ we use the same name to denote it.
\end{definition}

As the key technical result of our paper, in subsection~\ref{pinvariant-fragment}
we will provide a construction $(-)^{\bullet}: \yvcwAut(\olque) \to 
\yvcwAut(\ofo)$, such that for all $\bbA$ and $\bbT$ we have
\begin{equation}
\label{eq:crux}
\bbA^{\bullet} \text{ accepts } \bbT \text{ iff } \bbA \text{ accepts 
} \bbT^{\om},
\end{equation}
where $\bbT^{\om}$ is the $\om$-unravelling of $\bbT$.
As we shall see, the map $(-)^{\bullet}$ is completely determined at the 
one-step level, that is, by some model-theoretic connection between 
$\olque$ and $\ofo$.

The second fact, to be discussed in
subsection \ref{aut-to-formula}, is that for each $\contAFMC$-automaton $\bbA$  we can effectively construct an equivalent $\yvF$-formula $\xi_{\bbA}$.

On the basis of the above observations we show that those results are enough to prove Theorem~\ref{t:m1}(i) as follows:

\begin{proofof}{Theorem~\ref{t:m1}}
\textbf{(1)} Given a \wmso-formula $\phi$, let $\phi^{\bullet} \isdef
\xi_{\aut_{\phi}^{\bullet}}$.
We verify that $\phi$ is bisimulation invariant iff $\phi$ and $\phi^{\bullet}$
are equivalent.
The direction from right to left is immediate by the observation that
$\phi^{\bullet}$ is a formula in $\MC$.
The opposite direction follows from the following chain of equivalences:
\begin{align*}
\model \models \phi
  & \text{ iff } \bbT^{\om} \models \phi
  & \tag{$\phi$ bisimulation invariant}
\\ & \text{ iff }  \bbA_{\phi} \text{ accepts } \bbT^{\om}
  & \tag{$ \phi \equiv \aut_{\phi}$ on trees}
\\ & \text{ iff } \bbA_{\phi}^{\bullet} \text{ accepts } \bbT
& \tag{\ref{eq:crux}}
\\ & \text{ iff }  \bbT \models \xi_{ \aut_{\phi}^{\bullet}}
& \tag{$\aut_{\phi}^{\bullet}\equiv \xi_{ \aut_{\phi}^{\bullet}}$}
\end{align*}
\textbf{(2)} For the second part of Theorem \ref{t:m1},
We 
first define, for every first-order variable $x$, a translation $ST_x$ from
the $\mu$-calculus into the set of $\mlque$-formulas with only $x$ free:

\begin{itemize}
\itemsep 0 pt
\item $ST_x(p)=p(x)$
\item $ST_x(\varphi \land \psi)=ST_x(\varphi) \land ST_x(\psi)$,
\item $ST_x(\varphi \lor \psi)=ST_x(\varphi) \lor ST_x(\psi)$,
\item $ST_x(\lnot \varphi)= \lnot ST_x(\varphi)$,
\item $ST_x(\Diamond \varphi)=\exists y (R(x,y) \land ST_y(\varphi)$),
\item $ST_x(\mu p. \varphi)= \mu p. ST_x(\varphi)$,
\end{itemize}
Clearly, every formula of the $\yvF$-fragment of the $\mu$-calclus is mapped to a logically equivalent formula of the $\yvF$-fragment of $\mlque$. Let
 $(-)_{\bullet}:\yvF\to\wmso$ defined as the composite $\mgFOETr{-} \circ ST_x$. By Theorem \ref{thm:guard_wmso}
we obtain that $\psi \equiv \psi_{\bullet}$, for all $\psi \in
\yvF$.
\end{proofof}

%% file: pinvariant-fragment.tex

In this subsection we will define a construction that transforms an arbitrary automaton
$\bbA$ in $\yvcwAut(\olque)$ into an automaton $\bbA^{\bullet}$ in 
$\yvcwAut(\ofo)$, such that $\bbA$ and $\bbA^{\bullet}$ are related as 
in~\eqref{eq:crux}.
This construction is completely determined by the following translation at the
one-step level.

\begin{definition}
Using the fact
that by Corollary~\ref{cor:olquepositivenf}, any formula in ${\olque}^+(A)$ is 
equivalent to a disjunction of formulas of the form 
$\posdbnfolque{\vlist{T}}{\Pi}{\Sigma}$, we define the translation 
$(-)^{\bullet} : {\olque}^+(A) \to \ofo^+(A)$ as follows.
We set
\[
\Big( \posdbnfolque{\vlist{T}}{\Pi}{\Sigma} \Big)^{\bullet} \isdef
\bigwedge_{i} \exists x_i. \tau^+_{T_i}(x_i) \land \forall x. \bigvee_{S\in\Sigma} \tau^+_S(x)
\]
and for $\al = \bigvee_{i} \al_{i}$ we define $\al^{\bullet} \isdef \bigvee 
\al_{i}^{\bullet}$.
\end{definition}

The key property of this translation is the following.

\begin{proposition}
\label{p-1P}
For every one-step model $(D,V)$ and every $\al \in {\olque}^+(A)$ we have
\begin{equation}
\label{eq-1cr}
(D,V) \models \alpha^{\bullet} \text{ iff } (D\times \om,V_\pi) \models \alpha,
\end{equation}
where $V_{\pi}$ 
 is the induced valuation given by 
$V_{\pi}(a) \isdef \{ (d,k) \mid d \in V(a), k\in\omega\}$.
\end{proposition}

\begin{proof}
Clearly it suffices to prove \eqref{eq-1cr} for formulas of the form
$\al = \posdbnfolque{\vlist{T}}{\Pi}{\Sigma}$.
\smallskip

\noindent\fbox{$\Rightarrow$} 
Assume $(D,\val) \models \alpha^{\bullet}$, we will show that 
$(D\times \omega,\val_\pi) \models \posdbnfolque{\vlist{T}}{\Pi}{\Sigma}$.
Let $d_i$ be such that $\tau_{T_i}^+(d_i)$ in $(D,\val)$. 
It is clear that the $(d_i,i)$ provide \emph{distinct} elements satisfying 
$\tau_{T_i}^+((d_i,i))$ in $(D\times\omega,\val_{\pi})$ and therefore the 
first-order existential part of $\alpha$ is satisfied. 
With a similar but easier argument it is straightforward that the existential 
generalized quantifier part of $\alpha$ is also satisfied.
%
For the universal parts of $\posdbnfolque{\vlist{T}}{\Pi}{\Sigma}$ it is enough to observe that, because of the universal part of $\alpha^\bullet$, \emph{every} $d\in D$ realizes a positive type in $\Sigma$. By construction, the same applies to $(D\times\omega,\val_{\pi})$, 
therefore this takes care of both universal quantifiers.
\medskip
		
\noindent\fbox{$\Leftarrow$} 
Assuming that $(D\times \omega,\val_\pi) \models 
\posdbnfolque{\vlist{T}}{\Pi}{\Sigma}$,
we will show that $(D,\val) \models \alpha^\bullet$. 
The existential part of $\alpha^{\bullet}$ is trivial. 
For the universal part we have to show that every element of $D$ realizes the 
positive part of a type in $\Sigma$. 
Suppose not, and let $d\in D$ be such that $\lnot\tau_S^+(d)$ for all $S\in 
\Sigma$. 
Then we have $(D\times\omega,\val_\pi) \not\models \tau_S^+((d,k))$ for all $k$.
That is, there are infinitely many elements not realizing the positive part of 
any type in $\Sigma$. 
Hence we have $(D\times\omega,\val_\pi) \not\models \dqu y.\bigvee_{S\in\Sigma} 
\tau_S^+(y)$. 
Absurd, because that is part of $\posdbnfolque{\vlist{T}}{\Pi}{\Sigma}$.
\end{proof}

As a consequence of Proposition~\ref{p-1P} we obtain the following.

\begin{definition}
Given an automaton $\bbA = \tup{A,\De,\Om,a_{I}}$ in $\yvAut(\olque)$, define 
the automaton $\bbA^{\bullet} \isdef \tup{A,\De^{\bullet},\Om,a_{I}}$ in 
$\yvAut(\ofo)$ by putting, for each $(a,c) \in A \times C$:
\[
\De^{\bullet}(a,c) \isdef (\De(a,c))^{\bullet}.
\]
\end{definition}

\begin{proposition}
For any automaton $\bbA = \tup{A,\De,\Om,a_{I}}$ in $\yvAut(\olque)$, and any
model $\bbT$, $\bbA$ and $\bbT$ satisfy \eqref{eq:crux}.
\end{proposition}

\begin{proof}
The proof of this proposition is based on a fairly routine comparison of the 
acceptance games $\mathcal{A}(\bbA^{\bullet},\bbT)$ and 
$\mathcal{A}(\bbA,\bbT^{\om})$.
In a slightly more general setting, the details of this proof can be found 
in~\cite{Venxx}.
\end{proof}
\medskip

It remains to be checked that the construction $(-)^{\bullet}$, which has
been defined for arbitrary automata in $\yvAut(\olque)$, transforms 
$\yvWMSO$-automata into automata in the right class, viz., $\yvcwAut(\ofo)$.

\begin{proposition}
Let $\bbA \in \yvAut(\olque)$.
If $\bbA \in \yvcwAut(\olque)$, then $\bbA^{\bullet} \in \yvcwAut(\ofo)$.
\end{proposition}

\begin{proof}
This proposition can be verified by a straightforward inspection, at the 
one-step level, that if a formula $\al \in {\olque}^+(A)$ belongs to the fragment 
$\cont{{\olque}^+}{a}(A)$, then its translation $\al^{\bullet}$ lands in 
the fragment $\cont{\ofo^+}{a}(A)$.
\end{proof}

\begin{remark}{\rm
As a corollary of the previous two propositions we find that 
\begin{itemize}
	\itemsep 0 pt
	\item $\yvAut(\ofo) \equiv \yvAut(\ofo)/{\bis}$, and
	\item $\yvcwAut(\ofo) \equiv \yvcwAut(\ofo)/{\bis}$.
\end{itemize}
In fact, we are dealing here with an instantiation of a more general phenomenon 
that is essentially coalgebraic in nature.
In~\cite{Venxx} it is proved that if $\yvLo$ and $\yvLo'$ are two one-step
languages that are connected by a translation $(-)^{\bullet}: \yvLo' \to 
\yvLo$ satisfying a condition similar to \eqref{eq-1cr}, then we find that 
$\yvAut(\yvLo)$ corresponds to the bisimulation-invariant fragment of 
$\yvAut(\yvLo')$: $\yvAut(\yvLo) \equiv \yvAut(\yvLo')/{\bis}$.
This subsection can be generalized to prove similar results relating
$\yvwAut(\yvLo)$ to $\yvwAut(\yvLo')$, and $\yvcwAut(\yvLo)$ to 
$\yvcwAut(\yvLo')$.
}\end{remark}

%% file: aut-to-formula.tex

%
%

In this subsection we focus on the following theorem.

\begin{theorem}\label{t:autofor}
There is an effective procedure that, given an automaton $\bbA$ in
$\yvcwAut(\ofo)$, returns an equivalent formula $\xi_{\bbA}$ of the fragment 
$\contAFMC$ of the modal $\mu$-calculus.
\end{theorem}

\begin{proof}
The argument  is a refinement of the standard proof showing that any automaton 
$\aut$ from $\yvAut(\ofo)$ can be translated into an equivalent $\mu$-formula 
$\xi_\aut$ (cf. e.g. \cite{Ven08}), and it is essentially a special case of the argument proving Theorem \ref{thm:wmso_autofor}.

From now on, we always assume that a formula $\tmap(a,c)$ is in normal form. 
Following~\cite{Ven08}, we introduce another type of automata, called $(\prop,X)$-automata, which operate on $\p{(\prop \cup X)}$-trees.
They differ from automata whose one-step language is defined over predicates in $(A \cup X)$ in that\footnote{Parity automata based on $\ofo$ are thus simply $(\prop,\emptyset)$-automata.}
\begin{itemize}
\item Monadic predicate letters from $X$ can occur in the scope of a 
quantifier and only there, meaning that $(\prop,X)$-automata have transition $\tmap(a,c) \in \ofo^+(A\cup X)$
\item The transition function is uniquely determined by the restriction of the coloring to $\prop$, that is, for every $a \in A$, and $c_1, c_2 \in C$, if $c_1 \cap \prop = c_2 \cap \prop$ then $\tmap(a, c_1)= \tmap(a, c_2)$.
\end{itemize}
We also assume that
for every $x \in X$ there is a unique $a \in A$ and an unique $c \in C$ such that $x$ occurs in $\tmap(a,c)$.
The notion of acceptance is defined as expected, the only difference with being that during the acceptance game \'Eloise has to provide a valuation only for predicates in $A$ making formula given by the transition function true. 
It is then enough to prove the following claim.
\begin{claimfirst}\label{c:1}
There is an effective procedure that, given a $(\prop,X)$-automaton $\aut$ gives an equivalent  formula $\xi_{\aut} \in \contAFMC$ in which all occurrences of variables in $X$ are positive.
\end{claimfirst}
\begin{pfclaim} 
Without loss of generality, we can assume that:
\begin{itemize}
\itemsep 0 pt
\item Every (maximal) strongly connected component (SCC) in the graph of $\ord$ has an unique entrance point,
\item The directed acyclic graph (DAG) of the SCCs of $\ord$ is a tree, and more specifically,
\item 
$\{c \in A \mid a \leadsto c, c \prec a\}  \cap \{c \in A \mid b \leadsto c, c \prec b\}  = \emptyset$ whenever $a,b$ are in the same SCC, with $a\neq b$.
\end{itemize}

Given a $(\prop,X)$-automaton $\aut$, we are now going to define a function $\delta_\aut: A \to \ML (A \cup X \cup \prop)$
that assigns to each state $a$ of $\aut$ a modal formula $\delta_\aut(a)$ over  $A \cup X \cup \prop$ representing all possible transitions from $a$ in the modal language with the property that if $b \in A$ is in the same $\ord$-cycle of $a$ and $\pmap(a)=1$, then  $\delta_\aut(a)$ is continuous in $b$. Dually for $\pmap(a)=0$.

Let $c \in C$, and assume $\Delta(a,c)$ is in positive basic form $\bigvee \posdbnfofo{\Sigma}$. We define a first translation $TR_1$ taking as argument $\Delta(a,c)$ and giving as result a formula from the (guarded fragment of) first-order logic over $A \cup X \cup \prop$ as follows. 
With every disjunct
$$
\posdbnfofo{\Sigma} = \bigwedge_{S\in\Sigma} \exists x. \tau^+_{S}(x) \land \forall z.( \bigvee_{S\in \Sigma} \tau^+_S(z)),
$$
%
we associate the formula

$$
TR_1(\posdbnfofo{\Sigma}) := \bigwedge_{S\in\Sigma} \exists y. (R(x,y) \land \tau^+_{S}(y)) \land \forall z.( R(x,z) \to \bigvee_{S\in \Sigma} \tau^+_S(z)).
$$
\fcwarning{Why do we need to go through this?}The formula $TR_1(\posdbnfofo{\Sigma})$ is bisimulation invariant, and it is equivalent to the modal formula

$$
TR_2(\posdbnfofo{\Sigma}) := \bigwedge_{S\in\Sigma}  \Diamond(\bigwedge S) \land \Box \bigvee_{S\in \Sigma} (\bigwedge S).
$$
\fcwarning{Maybe better to define this directly}
Let $TR_3(\Delta(a,c))= \bigwedge (c \cap \prop) \land \bigwedge_{p \in \prop \setminus c} \lnot p \land \bigvee TR_2(\posdbnfofo{\Sigma})$.
 The modal formula $\delta_\aut(a)$ is then defined as

 \[
 \bigvee_{c \in C} TR_3(\Delta(a,c)).
 \]
By construction we have, for every $\model$,
%
    \begin{eqnarray*}
    \model[x\mapsto s_I] \models \bigvee_{c \in C} \big (\tau_{(c \cap \prop)}(y) \land \bigvee TR_1(\posdbnfofo{\Sigma})\big) & \text{iff} &\model \mmodels \delta_\aut(a).
    \end{eqnarray*}
%
A modal automaton over $\prop$ is an automaton $ \tup{A, \delta, \pmap, a_I}$ such that $\delta : A \to \ML^+ (A)$, where $\ML^+ (A)$ is the set of all modal formulas over propositions $A \cup \prop$ such that elements from $A$ appear only positively.
The acceptance game associated with such an automaton and a tree $\model$ is determined by the (symmetric) acceptance game defined according to the rules of Table~\ref{symmetric_modal_game}.
This means that we can equivalently see the automaton $\aut$ as a modal automaton $\tup{A, \delta_\aut, \pmap, a_I}$ whose transition function satisfies the weakness and continuity conditions. Thus, from now on we see $\aut$ as the equivalent modal automaton we have just described.

\begin{table}[h]
  \centering
\begin{tabular}{|l|c|l|c|}
 \hline
  Position & Player & Admissible moves & Parity\\
   \hline
  $(a,s) \in A \times S$ & $\exists$ & $\{(\delta_\aut(a),s)\}$ & $\pmap(a)$\\
  $(\psi_1 \vee \psi_2,s)$ & $\exists$ & $\{(\psi_1,s),(\psi_2,s) \}$ & $-$ \\
  $(\psi_1 \wedge \psi_2,s)$ & $\forall$ & $\{(\psi_1,s),(\psi_2,s) \}$ & $-$ \\
  $(\Diamond\varphi,s)$ & $\exists$ & $\{(\varphi,t)\ |\ t \in R[s] \}$ & $-$ \\
  $(\Box\varphi,s)$ & $\forall$ & $\{(\varphi,t)\ |\ t \in R[s] \}$ & $-$ \\
  $(\lnot p,s) \in \prop \times S$ and $p \notin \tscolors(s)$ & $\forall$ & $\emptyset$ & $-$\\
  $(\lnot p,s) \in \prop \times S$ and $p \in \tscolors(s)$ & $\exists$ & $\emptyset$ & $-$\\
  $(p,s) \in \prop \times S$ and $p \in \tscolors(s)$ & $\forall$ & $\emptyset$ & $-$\\
  $(p,s) \in \prop \times S$ and $p \notin \tscolors(s)$ & $\exists$ & $\emptyset$ & $-$\\

  \hline
\end{tabular}
 \caption{(Symmetric) acceptance game for modal automata}
 \label{symmetric_modal_game}
\end{table}

For the construction of $\xi_\aut$, we proceed by induction on the tree height of the DAG $t$ of SCC. If the height is 1, that is the DAG is a single point graph, we reason as follows.
 We have two cases to consider: either the SCC is trivial (i.e. it consists of a single non looping node), or not.
In the first case, $A=\{a_I\}$ and $\aut$ is equivalent to $\xi_\aut:=\delta_\aut(a_I)$.

For the second case, let us assume that $\pmap(a_I)=1$, the case when it is $0$ being, mutatis mutandis, the same.
Let $A=\{a_0, \dots, a_\ell\}$, and $a_I=a_\ell$.
Since $\aut$ is a weak modal automaton satisfying the continuity condition, given $a,b \in A$, if $b$ occurs in $\delta_\aut(a)$, then $b$ is only in the scope of $\Diamond$ operator. 
We can now see the automaton $\bbA$
as a system of modal equations, and by applying the standard inductive procedure \emph{solves} this system of equations and
construct the least fixpoint formula $\xi_\aut$ equivalent to $\aut$, 

%
Here the key observation is that the weakness and continuity conditions on
strongly connected components of the automaton ensure that when we execute
a single step in solving the system of equations, we 
may work within the 
(syntactically) continuous fragment of the modal $\mu$-calculus.
From this, we can deduce that $\xi_\aut \in \contAFMC$. Clearly the procedure preserves the polarity of each $x \in X$, meaning that all variables in $X$ are positive in $\xi_\aut$.

For the induction step, assume the successors of the root of $t$ are $(n_1, \dots, n_\ell)$. For each $i \in \{1,\dots,\ell\}$, let $a_i$ be the entrance point of the SCC of $n_i$, and let $\aut_{i}$ be the automaton $\aut$ but having as an initial state $a_i$. If we do not consider the states that are not reachable by $a_i$, the DAG of the SCC of $\aut_{i}$ is $t.{n_i}$ (the subtree of $t$ starting at $n_i$).
Let $Y=\{a_1, \dots, a_\ell\}$, and $M$ be the set of states of $\aut$ that belong to the root of $t$. We assume $X \cap Y = \emptyset$. The structure
\[
\aut_M := \tup{M, \delta_\aut|_{M}, \pmap|_{M}, a_I}
\] is a $(\prop, X \cup Y)$-automaton.

The inductive hypothesis applies to automata $\aut_M, \aut_{1}, \dots, \aut_{\ell}$. Thus we obtain fixpoint formulas $\xi_M, \xi_1, \dots, \xi_\ell$, the former taking free variables in $\prop \cup X \cup Y$, all the remaining in $\prop \cup X$, equivalent to $\aut_M, \aut_1, \dots, \aut_\ell$ respectively.
Notice that:
\begin{itemize}
\item by induction hypothesis, $\xi_M, \xi_1, \dots, \xi_\ell \in \contAFMC$, every variable $x \in X \cup Y$ is positive in each of those formulas (if $x$ occurs in it),
\item by construction, if $x$ is free in $\xi_i$, then $x$ is not bounded in $\xi_{M}$.
\end{itemize}
 We can therefore deduce that
 $\xi_\aut= \xi_{M}[a_1\mapsto\xi_{1}, \dots, a_\ell\mapsto\xi_{\ell}] \in \contAFMC$, and that each variable from $X$ occurs positively in $\xi_\aut$.

 We verify that $\model \mmodels \xi_{M}[a_1\mapsto\xi_{1}, \dots, a_\ell\mapsto\xi_{\ell}]$ iff $\model \in \trees(\aut)$, for every model $\model$.
But this follows by the following two facts:
\begin{enumerate}
\item $\model\mmodels \xi_{M}[a_1\mapsto\xi_{1}, \dots, a_\ell\mapsto\xi_{\ell}]$ iff $\model[a_1\mapsto \ext{\xi_1}^{\model}, \dots, a_\ell\mapsto \ext{\xi_\ell}^{\model}  ] \mmodels \xi_M$,
\item $\model \in \trees(\aut)$  iff $\model[a_1\mapsto \ext{\aut_1}^{\model}, \dots, a_\ell\mapsto \ext{\aut_\ell}^{\model}  ] \in \trees(\aut_M)$
\end{enumerate}
where 
$\ext{\aut_i}^{\model} := \{ s \in \model \mid \model.s \in \trees(\aut_i) \}$. 
\end{pfclaim}
This finishes the proof of Theorem. 
\end{proof}

%% file: conclusions.tex

\paragraph{Overview.} In this work we have presented three main contributions.
First, we proved that the bisimulation-invariant fragment of \WMSO is the
fragment $\contAFMC$ of $\MC$ 
where the application of the fixpoint operator $\mu p$ is restricted to
formulas that are continuous in $p$.
Our result sheds light on the relationship between $\MSO$, \WMSO and $\MC$.
In particular, it provides a positive answer to the question whether
$\wmso / {\bis} \leq \AFMC$ on trees of arbitrary branching degree, left open
in~\cite{DBLP:conf/lics/FacchiniVZ13}.
This may also be read as the statement that the formulas that separate
$\wmso$ from $\mso$ are not bisimulation invariant (and hence, irrelevant in
the sense mentioned in the introduction).

To achieve this result, we shaped $\WMSO$-automata, a special kind of parity automata satisfying additional \emph{continuity} and \emph{weakness} conditions, with transition map given by the monadic logic $\olque$. Our second main contribution was to show that they characterize \WMSO on tree models.


As our third main contribution we gave a detailed model-theoretic analysis of the monotone and continuous fragments of $\olque$. We provide strong normal forms and syntactic characterizations that, besides being of independent interest, are critical for the development of the aforementioned results.

\paragraph{Future work.}  
A first line of research is directly inspired by the methods employed in this work. $\yvWMSO$-automata and $\yvF$-automata are essentially obtained by imposing conditions on the appropriate one-step logic $\llang_1$ and transition structure of automata belonging to $\yvAut(\llang_1)$. 
Following this approach, one could take aim at the automata-theoretic counterpart of other fragments of the modal $\mu$-calculus, like PDL, CTL or CTL$^*$. 

Another direction of investigation is based on the observation that, from a
topological point of view, all $\yvWMSO$-definable properties are Borel.
Since we do not have examples of  Borel $\yvMSO$-definable properties that
are not $\yvWMSO$-definable, is tempting to conjecture that $\yvWMSO$ is the
Borel fragment of $\yvMSO$ and analogously for $\yvF$ and $\MC$.

\paragraph{Acknowledgements.}
The second author is supported by the \emph{Expressiveness of Modal Fixpoint Logics} project realized within the 5/2012 Homing Plus programme of the Foundation for Polish Science, co-financed by the European Union from the Regional Development Fund within the Operational Programme Innovative Economy (``Grants for Innovation''). The fourth author is supported by the project ANR 12IS02001 PACE.